\documentclass[12pt, a4paper]{article}
\usepackage{natbib}
\setcitestyle{authoryear,open={(},close={)},citesep={,}}
\usepackage[]{hyperref}
\usepackage{xcolor}
\hypersetup{colorlinks = true, citecolor = blue, linkcolor=blue, linkbordercolor = {white}, urlcolor=blue}
\usepackage{amsmath}
\usepackage{mathrsfs}
\usepackage{amsthm}
\usepackage{amssymb}
\usepackage{mathtools}
\usepackage{booktabs}
\usepackage{bm}
\usepackage{graphicx}
\usepackage{caption}
\usepackage{subcaption}
\usepackage{multirow}
\usepackage{dsfont}
\usepackage{cleveref}
\usepackage[normalem]{ulem}
\usepackage{authblk}
\usepackage{enumitem}

\usepackage[T1]{fontenc}
\usepackage{lmodern}

\usepackage
[
        a4paper,
        left=2.5cm,
        right=2.5cm,
        top=2.5cm,
        bottom=2.5cm,
]
{geometry}

\newtheorem{theorem}{Theorem}[section]
\newtheorem{lemma}{Lemma}[section]
\newtheorem{corollary}{Corollary}[section]

\newtheorem{remark}{Remark}[section]

\newtheorem{assumption}{Assumption}[section]

\renewcommand{\P}{{\mathbb{P}}}
\newcommand{\E}{{\rm E}}

\newcommand{\tran}{^{\mathstrut\scriptscriptstyle{\top}}}

\newcommand{\Exp}{{\rm Exp}}
\newcommand{\Log}{{\rm Log}}

\numberwithin{equation}{section}

\makeatletter
\renewcommand{\@maketitle}{%
  \newpage
  \begingroup
    \linespread{1}\selectfont
    \centering
    \let\footnote\thanks
    {\LARGE\@title\par}
    \vspace{2em}
    {\@author\par}
    \vspace{0.8em}
  \endgroup
}
\makeatother

\begin{document}

\title{Functional Principal Component Analysis for Manifold-Indexed Data}

\author[1]{Chang Jun Im}
\author[2]{Jeong Min Jeon\thanks{Corresponding author (jm.jeon@snu.ac.kr)}}

\affil[1]{The Institute for Data Innovation in Science,
Seoul National University}
\affil[2]{Department of Statistics and School of Transdisciplinary
Innovations, Seoul National University}

\maketitle

\begin{abstract}
Functional principal component analysis has traditionally focused on processes indexed by Euclidean domains, although functional observations increasingly arise on curved domains. We develop functional principal component analysis for real-valued processes observed at discrete locations on a general Riemannian manifold. We estimate the mean and covariance functions using intrinsic kernel smoothing based on geodesic distance and a normal-coordinate volume correction, with a general weighting scheme accommodating unequal numbers of observations across subjects. We establish $L^2$ and uniform convergence rates for these estimators and derive corresponding error bounds for the estimated eigenvalues and eigenfunctions. These results quantify how manifold dimension, sampling frequency and the weighting scheme affect estimation accuracy and characterize the phase transition between sparse and dense sampling regimes. Simulation studies on the circle and sphere, together with real-data applications, illustrate the finite-sample performance of the proposed methodology.
\end{abstract}

\small
\begin{quotation}
\end{quotation}

\normalsize

\section{Introduction} \label{sec:intro}

Functional data analysis (FDA) provides statistical methods for observations represented by curves, surfaces or other functions. Its central goals include characterizing population variation, reducing dimensionality and recovering individual trajectories from discrete and noisy measurements. Functional principal component analysis (FPCA) addresses these goals through the spectral decomposition of the covariance operator. Classical developments primarily concern real-valued processes indexed by compact intervals, with methods and theory accommodating both densely observed curves and sparse longitudinal measurements (\citet{RiceSilverman1991,YaoEtAl2005,HallHosseiniNasab2006,HallEtAl2006,LiHsing2010}). More recently, \cite{ZhangWang2016} studied mean and covariance estimation under a general weighting scheme, and \cite{Petersen2024} extended this framework to high-dimensional functional data.

In modern FDA, functional observations are often indexed by curved domains. Examples include neuron-specific firing-rate functions indexed by three-dimensional head orientation (\cite{AngelakiEtAl2020}), neuroretinal rim thickness profiles indexed by angular position around the optic disc (\cite{AliEtAl2021}) and sea-surface temperature anomalies observed at irregularly spaced ocean locations (\cite{CaponeraEtAl2022}). These functional data can be indexed by the rotation group ${\rm SO}(3):=\{s\in\mathbb{R}^{3\times 3}:s~\text{is an orthogonal matrix with determinant of 1}\}$, the circle $\mathbb{S}^1:=\{s\in\mathbb{R}^2:\|s\|=1\}$ and the sphere $\mathbb{S}^2:=\{s\in\mathbb{R}^3:\|s\|=1\}$, respectively. 

In this paper, we consider real-valued functions indexed by a general manifold, a setting distinct from manifold-valued functional data analysis, in which the function values lie on a manifold (\citet{DaiMuller2018,LinYao2019,ShaoEtAl2022,JeonVan2025}). The formulation of FPCA remains applicable in this setting, but estimation must account for the intrinsic distance and volume measure of the indexing manifold, which govern local neighborhoods and integration. Treating a coordinate representation of the manifold (e.g., the angle of $s\in\mathbb{S}^1$ or the longitude and latitude of $s\in\mathbb{S}^2$) as an ordinary Euclidean domain can distort distances and introduce discontinuities across coordinate boundaries, affecting the estimation of mean and covariance functions. These considerations motivate an FPCA framework that respects the geometry of the indexing manifold. Respecting the underlying geometry has likewise guided the development of statistical methods for other problems involving manifold-valued data (\citet{YuanEtAl2012,Di MarzioEtAl2014,CorneaEtAl2017,PetersenMuller2019,JeonEtAl2021,JeonEtAl2022,ImEtAl2025}).

Relatively few studies have developed statistical methods for functional data indexed by manifolds. \citet{LilaEtAl2016} developed smooth FPCA for real-valued functions defined on two-dimensional Riemannian manifolds in $\mathbb{R}^3$. Their method jointly estimates principal component functions and subject-specific scores through penalized approximation rather than kernel smoothing, combining a Laplace-Beltrami roughness penalty with surface finite elements. \citet{AliEtAl2021} proposed a framework for clustering real-valued functions defined on $\mathbb{S}^1$ by representing the functions using Fourier bases and applying a discriminative functional mixture model to their coefficients. \citet{CaponeraEtAl2022} developed estimators of mean and covariance functions for a real-valued process indexed by $\mathbb{S}^2$ using generalized Tikhonov regularization in reproducing kernel Hilbert spaces (RKHS). They established $L^2$ convergence rates for both independent and serially dependent functional observations. \citet{ReimherrEtAl2023} studied function-on-scalar regression for functional responses indexed by Riemannian manifolds. They proposed an RKHS-based penalized estimator and established minimax optimal convergence rates. \citet{OlsenEtAl2027} developed a permutation-based procedure for local inference on coefficient functions in linear regression models with functional responses indexed by Riemannian manifolds. However, a formal asymptotic analysis is missing in \citet{LilaEtAl2016}, \citet{AliEtAl2021} and \citet{OlsenEtAl2027}.

We develop an intrinsic kernel-based FPCA framework for real-valued processes indexed by a general compact connected Riemannian manifold without boundary. Allowing for noisy observations at subject-specific locations, we estimate the mean and covariance functions based on geodesic distance and a normal-coordinate volume correction, and obtain principal component estimates through spectral decomposition of the estimated covariance operator. We extend the general weighting scheme of \citet{ZhangWang2016} to manifold domains in order to effectively accommodate varying numbers of observations across subjects. We establish $L^2$ and uniform convergence rates for the mean and covariance estimators, together with error bounds for the estimated eigenvalues and eigenfunctions. The resulting bounds quantify how manifold dimension, sampling frequency and the weighting scheme affect estimation accuracy and characterize the phase transition between sparse and dense sampling regimes. To the best of our knowledge, this work provides the first asymptotic theory for FPCA on general manifold domains, the first kernel-smoothing-based FPCA framework for manifold-indexed functional data, and the first application of a general weighting scheme to such data.

We assess the finite-sample performance of the proposed method under sparse and dense sampling regimes through simulations on $\mathbb{S}^1$ and $\mathbb{S}^2$. We further illustrate the method through applications to circular profiles of retinal nerve fiber layer thickness measured by optical coherence tomography (\cite{AliEtAl2023}) and to head-related transfer functions indexed by sound-source directions on $\mathbb{S}^2$ (\cite{EngelEtAl2023}).

This paper is organized as follows. Section~\ref{sec:method} introduces the problem setting, mathematical preliminaries and estimation procedure, and Section~\ref{sec:theory} establishes the asymptotic properties of the proposed estimators. Section~\ref{sec:simulation} reports the simulation studies, and Section~\ref{sec:realdata} presents the two real-data applications. The Supplementary Material contains all technical proofs.

\section{Methodology} \label{sec:method}

\subsection{Problem setting}

Let $\mathcal M$ be a compact connected $d$-dimensional Riemannian manifold without boundary, where $d\in\mathbb N$. Examples include the hyperspheres $\mathbb{S}^d:=\{s\in\mathbb{R}^{d+1}:\|s\|=1\}$, tori $\mathbb{S}^1\times\cdots\times\mathbb{S}^1$, planar shape spaces (\citet{LeKendall1993,Stocker et al. (2023)}) and special orthogonal groups ${\rm SO}(k):=\{s\in\mathbb{R}^{k\times k}:s~\text{is an orthogonal matrix with determinant of 1}\}$ whose dimension is $d=k(k-1)/2$ for $k\ge2$. 

We consider a stochastic process $X:\mathcal{M}\to\mathbb{R}$ such that $\E(X^2(s))<\infty$ for every $s\in\mathcal{M}$. Define its mean and covariance functions by
\begin{align*}
    \mu(s):=\E(X(s)), \quad C(s,t):=\E((X(s)-\mu(s))(X(t)-\mu(t))), \quad s,t\in\mathcal{M},
\end{align*}
respectively. Let $X_1,\ldots,X_n$ be independent copies of $X$. We observe
\begin{align*}
    Y_{ij}:=X_i(s_{ij})+\epsilon_{ij},
    \quad i=1,\ldots,n,
    \quad j=1,\ldots,N_i,
\end{align*}
where the numbers of observations $N_i$ and the observation locations $s_{ij}\in\mathcal M$ may vary across subjects. The measurement errors $\epsilon_{ij}$ are i.i.d. with mean 0 and variance $\sigma^2<\infty$.

\subsection{Mathematical preliminaries}

We introduce some notions and facts in Riemannian geometry needed to understand our methodology. For further details, we refer to the standard textbooks \cite{DoCarmo1992} and \cite{Chavel2006}.

Let $T_s\mathcal M$ denote the tangent space at $s\in\mathcal M$. It is a vector space equipped with an inner product $\langle\cdot,\cdot\rangle_s$. For $r>0$ and the norm $\|z\|_s:=\langle z,z\rangle_s^{1/2}$, write $B_s(0_s,r):=\{z\in T_s\mathcal M:\|z\|_s<r\}$, where $0_s$ is the zero vector of $T_s\mathcal M$. Let $\gamma_{s,z}:[0,1]\to\mathcal{M}$ denote the geodesic satisfying $\gamma_{s,z}(0)=s$ and $\dot{\gamma}_{s,z}(0)=z$. The Riemannian exponential map $\Exp_s:T_s\mathcal M\to\mathcal M$ at $s$ is defined by $\Exp_s(z):=\gamma_{s,z}(1)$. Also, the injectivity radius $\operatorname{inj}(s)$ at $s$ is defined by
\begin{align*}
    \operatorname{inj}(s):=\sup\bigl\{r>0 : \Exp_s \text{ is a diffeomorphism from } B_s(0_s,r) \text{ onto its image}\bigr\}.
\end{align*}
Since $\mathcal M$ is compact, its global injectivity radius $\iota_{\mathcal M}:=\inf_{s\in\mathcal M}\operatorname{inj}(s)$ is finite and strictly positive. Let $d_{\mathcal M}$ denote the geodesic distance on $\mathcal{M}$. For $u\in\mathcal M$, the distance $d_{\mathcal M}(s,u)$ equals the infimum of the lengths $\int_0^1\|\dot\gamma(x)\|_{\gamma(x)}\,\mathrm{d}x$ over all piecewise smooth curves $\gamma:[0,1]\to\mathcal M$ satisfying $\gamma(0)=s$ and $\gamma(1)=u$. Write $B_{\mathcal M}(s,r):=\{u\in\mathcal M:d_{\mathcal M}(s,u)<r\}$. The Riemannian logarithmic map $\Log_s:B_{\mathcal M}(s,\iota_{\mathcal M})\to T_s\mathcal M$ at $s$ is defined by $\Log_s(u):=\Exp_s^{-1}(u)$.

We now introduce the tangent-space volume measure and volume density. Let $\operatorname{vol}_s$ denote the push-forward of the Lebesgue measure on $\mathbb R^d$ under any linear isometry from $\mathbb R^d$ onto $T_s\mathcal M$. This definition is independent of the choice of isometry. For any $r>0$ and measurable set $A\subset T_s\mathcal M$, the tangent-space volume measure $\operatorname{vol}_s$ satisfies
\begin{align}\label{eq:tangent_volume_scaling}
    \operatorname{vol}_s(rA)=r^d\operatorname{vol}_s(A).
\end{align}
Fix an orthonormal basis $\{e_1,\ldots,e_d\}$ of $T_s\mathcal M$. For $u\in B_{\mathcal M}(s,\iota_{\mathcal M})$, let
$z=\Log_s(u)$ and define the volume density
\begin{align*}
    \theta_s(u):=\sqrt{\det\left(\left[\left\langle(d\Exp_s)_z(e_k),(d\Exp_s)_z (e_\ell)\right\rangle_u\right]_{k,\ell=1}^d\right)},
\end{align*}
where $(d\Exp_s)_z:T_s\mathcal M\to T_u\mathcal M$ denotes the differential of $\Exp_s$ at $z$. The value of $\theta_s(u)$ is independent of the choice of orthonormal basis, and $\theta_s(u)$ is the Jacobian factor for the change of variables $u=\Exp_s(z)$. In particular, for any $0<r<\iota_{\mathcal M}$ and integrable function $f:\mathcal M\to\mathbb R$ with respect to the Riemannian volume measure $V$ on $\mathcal M$, it holds that
\begin{align}\label{eq:change_of_variables}
    \int_{B_{\mathcal M}(s,r)}f(u)\,\mathrm{d}V(u)
    =\int_{B_s(0_s,r)}f(\Exp_s(z))\theta_s(\Exp_s(z))\,\mathrm{d}\operatorname{vol}_s(z).
\end{align}

For example, if $\mathcal M=\mathbb S^d$, then we may identify its tangent space at $s$
with the $d$-dimensional subspace $\{z\in\mathbb R^{d+1}:\langle z,s\rangle=0\}$.
In this case, the Riemannian exponential map is given by
\begin{align*}
    \Exp_s(z)=
    \begin{cases}
        \cos(\|z\|_s)s+
        \dfrac{\sin(\|z\|_s)}{\|z\|_s}z,
        & z\neq 0,\\
        s, & z=0,
    \end{cases}
\end{align*} 
the global injectivity radius is $\iota_{\mathbb S^d}=\pi$,
the geodesic distance is
$d_{\mathbb S^d}(s,u)=\arccos(\langle s,u\rangle)$
and the volume density is given by
\begin{align*}
    \theta_s(u)=
    \begin{cases}
        \left(
            \dfrac{\sin(d_{\mathbb S^d}(s,u))}
                  {d_{\mathbb S^d}(s,u)}
        \right)^{d-1},
        & 0<d_{\mathbb S^d}(s,u)<\pi,\\
        1,
        & d_{\mathbb S^d}(s,u)=0.
    \end{cases}
\end{align*}
For explicit expressions of these notions on other manifolds, we refer to \cite{FalorsiEtAl2019} and \citet{JeonEtAl2021}.

\subsection{Estimation of mean and covariance functions}\label{sec:estimation method}


We introduce intrinsic kernel estimators for the mean and covariance functions of $X$. Let $K:[0,\infty)\to[0,\infty)$ be a Lipschitz continuous kernel that vanishes on $(1,\infty)$ and satisfies
\begin{align}\label{kernel normalization}
    \int_{\mathbb R^d}K(\|\bm z\|)\,\mathrm d\bm z=1,\quad\int_{\mathbb R^d}\bm zK(\|\bm z\|)\,\mathrm d\bm z=(0,\ldots,0)\tran.  
\end{align}
For a bandwidth $h\in(0,\iota_{\mathcal M})$, define the kernel weight
\begin{align*}
    \mathcal L_{s,h}(u)
    :=
    \frac{1}{h^d\,\theta_s(u)}
    K\!\left(\frac{d_{\mathcal M}(s,u)}{h}\right)
\end{align*}
whenever $d_{\mathcal M}(s,u)\le h$, and set $\mathcal L_{s,h}(u):=0$ otherwise. This kernel weight was introduced by \cite{Pelletier2005}.

To accommodate heterogeneous sampling frequencies across subjects using a general weighting scheme, let $w_i\ge0$ be the weights satisfying
\begin{align} \label{eq:weight_mean_normalization}
    \sum_{i=1}^n N_iw_i=1.
\end{align}
For a bandwidth $h_\mu\in(0,\iota_{\mathcal M})$, we estimate the mean function $\mu$ by
\begin{align*}
    \hat{\mu}(s) := \frac{\sum_{i=1}^n w_i\sum_{j=1}^{N_i} \mathcal{L}_{s,h_\mu}(s_{ij})Y_{ij}}{\sum_{i=1}^n w_i\sum_{j=1}^{N_i} \mathcal{L}_{s,h_\mu}(s_{ij})}, \quad s\in\mathcal{M},
\end{align*}
which is a weighted Nadaraya-Watson-type estimator. 

For the estimation of the covariance function $C$, define
\begin{align*}
    \mathcal{I}_C:=\{i:N_i\ge2\}, \quad n_C:=|\mathcal{I}_C|.
\end{align*}
Let $v_i\ge0$ for $i\in\mathcal{I}_C$ be the weights such that
\begin{align} \label{eq:weight_covariance_normalization}
    \sum_{i\in\mathcal I_C} N_i(N_i-1)v_i = 1.
\end{align}
Define the residuals
\begin{align*}
    R_{ij}:=Y_{ij}-\hat{\mu}(s_{ij}), \quad i=1,\ldots,n, \quad j=1,\ldots,N_i.
\end{align*}
For a bandwidth $h_C\in(0,\iota_{\mathcal M})$, we estimate $C$ by
\begin{align}
\label{eq:cov_est}
    \hat{C}(s,t)
    :=
    \frac{
        \sum_{i\in\mathcal{I}_C} v_i
        \sum_{1\le j\ne k\le N_i}
        \mathcal{L}_{s,h_C}(s_{ij})\,\mathcal{L}_{t,h_C}(s_{ik})\,R_{ij}R_{ik}
    }{
        \sum_{i\in\mathcal{I}_C} v_i
    \sum_{1\le j\ne k\le N_i}
    \mathcal{L}_{s,h_C}(s_{ij})\,\mathcal{L}_{t,h_C}(s_{ik})
    }.
\end{align}
Excluding $j=k$ in the above sums removes the contribution of the measurement error variance, so that $\hat{C}$ targets $C$ rather than the covariance of the noisy observations.

Note that subjects with only one observation contribute to mean estimation
but do not contribute directly to the covariance smoothing sums. Unlike standard approaches to functional data analysis on a compact interval, our method does not use local-linear-type estimators for $\mu$ and $C$. This is because $\mathcal{M}$ has no boundary, and thus boundary effects do not arise.

Our general weighting scheme covers two commonly used schemes. The first is the equal-weight-per-observation (OBS) scheme defined by
\begin{align*}
    w_i = \frac{1}{\sum_{\ell=1}^n N_\ell}, \quad 1\leq i\leq n,\quad
    v_i = \dfrac{1}{\sum_{\ell\in\mathcal{I}_C} N_\ell(N_\ell-1)}, \quad i\in\mathcal{I}_C.
\end{align*}
Under this scheme, all observations $Y_{ij}$ receive equal weight in estimating $\mu$, and all $R_{ij}R_{ik}$ receive equal weight in estimating $C$. Consequently, the OBS scheme assigns greater total weight to subjects with larger observation counts.

The second is the equal-weight-per-subject (SUBJ) scheme defined by
\begin{align*}
    w_i = \frac{1}{nN_i}, \quad 1\leq i\leq n,\quad
    v_i = \dfrac{1}{n_C\,N_i(N_i-1)}, \quad i\in\mathcal{I}_C.
\end{align*}
These weights ensure that each subject contributes the same total weight. 

\begin{remark}\label{rmk1}
    By limiting the influence of subjects with larger observation counts, SUBJ can be advantageous under dense sampling, where dependence among repeated observations limits the additional information about population variation gained from further measurements within a subject. Under sparse sampling, however, OBS may be more efficient because the additional observations from more frequently sampled subjects can provide substantial information.
\end{remark}

\subsection{Covariance operator and spectral decomposition}

Let $L^2(\mathcal M)$ denote the space of square-integrable real-valued functions defined on $\mathcal M$, equipped with the inner product
\begin{align*}
    \langle f,g\rangle_2:=\int_{\mathcal M}f(u)g(u)\,\mathrm{d}V(u).
\end{align*}
Similarly, let $L^2(\mathcal M\times\mathcal M)$ denote the space of square-integrable real-valued functions defined on $\mathcal M\times\mathcal M$, equipped with the inner product
\begin{align*}
    \langle F,G\rangle_2:=\int_{\mathcal M}\int_{\mathcal M}F(u,v)G(u,v)\,\mathrm{d}V(u)\,\mathrm{d}V(v).
\end{align*}
Let $\|\cdot\|_2$ denote the norms induced by these inner products.

Suppose that $X$ is an $L^2(\mathcal M)$-valued
random element satisfying $\E(\|X\|_2^2)<\infty$. The covariance operator $\mathcal C:L^2(\mathcal M)\to L^2(\mathcal M)$ of $X$ is defined by
\begin{align*}
    (\mathcal C\psi)(s)
    :=\int_{\mathcal M}C(s,u)\psi(u)\,\mathrm d V(u).
\end{align*}
It is compact, self-adjoint and positive-semidefinite. Let $\lambda_1\ge\lambda_2\ge\cdots\ge0$ denote its eigenvalues,
counted with multiplicity, with corresponding orthonormal
eigenfunctions $\phi_k$. Then,
\begin{align*}
    C(s,t)
    =\sum_{k=1}^{\infty}\lambda_k\phi_k(s)\phi_k(t),\quad\sum_{k=1}^{\infty}\lambda_k=\int_{\mathcal M}C(u,u)\,\mathrm d V(u)
      =\E(\|X-\mu\|_2^2),
\end{align*}
where the first equality holds in the $L^2(\mathcal M\times\mathcal M)$ sense.

The estimator $\hat C$ in \eqref{eq:cov_est} induces the compact self-adjoint operator
\begin{align*}
    (\hat{\mathcal C}\psi)(s)
    :=\int_{\mathcal M}\hat C(s,u)\psi(u)\,\mathrm d V(u).
\end{align*}
The operator $\hat{\mathcal C}$ need not be positive-semidefinite,
a familiar issue with covariance smoothing (\cite{HallEtAl2006}). Let $\hat\lambda_1\ge\hat\lambda_2\ge\cdots>0$ denote its positive eigenvalues, arranged
in non-increasing order and counted with multiplicity. For each $\hat\lambda_k>0$, choose a corresponding eigenfunction
$\hat\phi_k$ so that the resulting eigenfunctions are orthonormal
and satisfy $\langle\hat\phi_k,\phi_k\rangle_2\ge0$.

\section{Asymptotic theory}\label{sec:theory}

We establish convergence rates for the mean and covariance estimators and derive the corresponding spectral error bounds. We then specialize the results to the OBS and SUBJ schemes. Throughout this section, we assume that the counts $N_i\ge1$ are deterministic and may depend on $n$. For random counts, the results hold conditionally on the count array $\{N_i:1\leq i\leq n\}$ if the assumptions in this section hold conditionally. We treat $L^2$ convergence rates first and then consider uniform convergence rates. Deriving these rates requires delicate technical arguments to account for the geometry of \(\mathcal{M}\) while accommodating arbitrary dimensions \(d\geq1\) and flexible choices of bandwidths and weights.

\subsection{\texorpdfstring{$L^2$}{L2} results}

\subsubsection{Assumptions}

In this section, we state basic assumptions.

\begin{assumption}\label{ass:sampling}
The array $\{s_{ij}:1\le i\le n,\ 1\le j\le N_i\}$ of observation locations, the array $\{\epsilon_{ij}:1\le i\le n,\ 1\le j\le N_i\}$ of measurement errors and the processes $\{X_i:1\leq i\leq n\}$ are mutually independent. The observation locations are i.i.d. with density $f$ with respect to $V$. The density $f$ is twice continuously differentiable on $\mathcal{M}$ and there exist positive constants $c_f$ and $C_f$ such that
\begin{align*}
    c_f\le f(s)\le C_f, \quad s\in\mathcal M.
\end{align*}
\end{assumption}

\begin{assumption}\label{ass:process}
The functions $\mu$ and $C$ are twice continuously differentiable on $\mathcal M$ and $\mathcal M\times\mathcal M$, respectively. It holds that
\[
    \E\!\left[\sup_{s\in\mathcal M}|X_i(s)|^4\right]<\infty,
    \quad \E(|\epsilon_{ij}|^4)<\infty.
\]
\end{assumption}

For integers $m, q\ge1$, define the falling factorial $(m)_q:=m(m-1)\cdots(m-q+1)$ if $m\geq q$ and define $(m)_q:=0$ otherwise.

\begin{assumption}\label{ass:weights}
It holds that
\[  n_C\to\infty,\quad
    n\max_i(N_iw_i)=O(1),\quad
    n_C\max_{i\in\mathcal I_C}\{(N_i)_2v_i\}=O(1).
\]
\end{assumption}

A similar maximal-weight condition to Assumption \ref{ass:weights} was employed in \citet{ZhangWang2016}.
Note that \citet{ZhangWang2016} treated functional data indexed by a compact interval and assumed $n_C=n$ throughout their work. In our work, we allow $n_C$ to differ from $n$.

\begin{remark}\label{rmk2}
The two maximal-weight conditions in Assumption~\ref{ass:weights} limit the contribution of any single subject and hold automatically under the SUBJ scheme. Under the OBS scheme, they are equivalent to $n\max_i N_i/\sum_i N_i=O(1)$ and $n_C\max_{i\in\mathcal I_C}(N_i)_2/\sum_{i\in\mathcal I_C}(N_i)_2=O(1)$, respectively. These bounds hold when observation counts are uniformly comparable across subjects, but can fail under severe sampling imbalance. For example, if $N_1\to\infty$ while $N_i=2$ for all $i\ge2$, then both $n\max_i(N_iw_i)$ and $n_C\max_{i\in\mathcal I_C}\{(N_i)_2v_i\}$ diverge under OBS as $n\to\infty$. Thus, a small group of heavily sampled subjects can cause the OBS scheme to violate these conditions.
\end{remark}

Define
\begin{alignat*}{2}
    V_{\mu,q}&:=\sum_{i=1}^n(N_i)_qw_i^2,
    \quad &V_{C,q}&:=\sum_{i\in\mathcal I_C}(N_i)_qv_i^2,\\
    \rho_{\mu,2}^2&:=\frac{V_{\mu,1}}{h_\mu^d}+V_{\mu,2},
    \quad &\rho_{C,2}^2&:=\frac{V_{C,2}}{h_C^{2d}}
                      +\frac{V_{C,3}}{h_C^d}+V_{C,4}.
\end{alignat*}

\begin{assumption}\label{ass:bandwidth}
For $\kappa\in\{\mu,C\}$, it holds that
\[
    h_\kappa=o(1),\quad
    \log(1/h_\kappa)=O(\log n),\quad
    (\log n)\rho_{\kappa,2}^2=o(1).
\]
\end{assumption}

The last condition in Assumption~\ref{ass:bandwidth} ensures that
the smoothing denominators are bounded away from zero with probability
tending to one. It excludes settings in which too few subjects contribute to covariance estimation. In particular, it requires $n_C/\log n\to\infty$.

\subsubsection{Rates of convergence}\label{sec:rates}

We derive convergence rates and compare them with those obtained by \cite{ZhangWang2016}. We first state the convergence rate of $\hat\mu$. Let $r_{\mu,2}:=h_\mu^2+\rho_{\mu,2}$.

\begin{theorem}\label{thm:mean_l2_rate}
Under Assumptions~\ref{ass:sampling}--\ref{ass:weights} and the mean-side conditions in Assumption~\ref{ass:bandwidth}, it holds that
\[
    \|\hat\mu-\mu\|_2=O_p(r_{\mu,2}).
\]
\end{theorem}

When $d=1$, this bound has the same form as the mean error bound
in Theorem 4.1 of \cite{ZhangWang2016}. For the convergence rate of $\hat{C}$, define
\begin{equation}\label{eq:mean_cov_leverage}
    \Gamma_{\mu C}:=
    \frac{\sum_{i\in\mathcal I_C}(N_i)_2v_iw_i^2}{h_\mu^{2d}}
    +\frac{\sum_{i\in\mathcal I_C}N_i(N_i-1)^2v_iw_i^2}{h_\mu^d}
\end{equation}
and let
\[
    r_{\mu C,2}:=r_{\mu,2}+\Gamma_{\mu C}^{1/2},
    \quad r_{C,2}:=r_{\mu C,2}+h_C^2+\rho_{C,2}.
\]

\begin{theorem}\label{thm:cov_l2_rate}
Under Assumptions~\ref{ass:sampling}--\ref{ass:bandwidth}, it holds that
\[
    \|\hat C-C\|_2=O_p(r_{C,2}).
\]
\end{theorem}

When $d=1$, the error bound $h_C^2+\rho_{C,2}$ has the same
form as the covariance error bound in Theorem 4.2 of
\cite{ZhangWang2016}, where the mean function is assumed known.
We do not assume that $\mu$ is known and this produces the additional term $r_{\mu C,2}$. Note that the same observations are used to estimate the mean and to form the residual products for covariance estimation. The term $\Gamma_{\mu C}$ accounts for this dependence. For the eigenfunction error bound, we introduce a typical assumption.

\begin{assumption}\label{ass:eigengap}
It holds that $\lambda_1>\lambda_2>\cdots>\lambda_{K^*}>\lambda_{K^*+1}\geq0$ for some fixed $K^*\ge1$.
\end{assumption}

\begin{theorem}\label{thm:eigen_rates}
Under the condition $r_{C,2}\to0$ and Assumptions~\ref{ass:sampling}--\ref{ass:bandwidth}, it holds that
\begin{align*}
    |\hat\lambda_k-\lambda_k|=O_p(r_{C,2}),\quad k\ge1.
\end{align*}
If Assumption~\ref{ass:eigengap} further holds, then it holds that
\begin{align*}
    \max_{1\le k\le K^*}\|\hat\phi_k-\phi_k\|_2
       =O_p(r_{C,2}).
\end{align*}
\end{theorem}

\subsubsection{Case of OBS and SUBJ}\label{sec:balanced_rates}

We specialize the rates in Section~\ref{sec:rates} to the OBS and SUBJ schemes. We use the subscripts $O$ and $S$ to stand for OBS and SUBJ. Define the effective frequencies by
\begin{alignat*}{2}
     a_{\mu,O}&:=\frac{1}{n}\sum_{i=1}^n N_i,
    \quad & a_{C,O}&:=\frac{1}{n_C}\sum_{i\in\mathcal I_C}(N_i)_2,\\
     a_{\mu,S}&:=\left(\frac1n\sum_{i=1}^n N_i^{-1}\right)^{-1},
    \quad & a_{C,S}&:=\left(\frac1{n_C}\sum_{i\in\mathcal I_C}(N_i)_2^{-1}\right)^{-1}.
\end{alignat*}
For $W\in\{O,S\}$, we consider bandwidths satisfying
\begin{align}
    h_{\mu,W}\asymp\max\{(na_{\mu,W})^{-1/(d+4)},n^{-1/4}\},\quad h_{C,W}\asymp\max\{(n_Ca_{C,W})^{-1/(2d+4)},n_C^{-1/4}\},
    \label{eq:scheme_bandwidths}
\end{align}
where $\asymp$ denotes equality of asymptotic orders, meaning that the ratio is bounded above and away from zero as $n\to\infty$. The same bandwidth rates with $d=1$ were employed in Corollaries 4.2 and 4.4 of \cite{ZhangWang2016}. 

Under \eqref{eq:scheme_bandwidths} and Assumption~\ref{ass:weights}, Assumption~\ref{ass:bandwidth} holds for both schemes $W\in\{O,S\}$ whenever
\begin{align}\label{eq:sufficient}
    \frac{n_C}{(\log n)^{(d+2)/2}}\longrightarrow\infty.    
\end{align}
Under \eqref{eq:scheme_bandwidths} and Assumption~\ref{ass:weights}, we have $r_{C,2}\to0$ for $W=O$. Under the same assumptions and that $n_C\cdot h_{\mu,S}^d$ is bounded away from zero, we also have $r_{C,2}\to0$ for $W=S$. Hence, the proposed estimators are consistent under weak assumptions.

For bounded observation counts and $n_C\asymp n$, both schemes give $a_{\mu,W},a_{C,W}\asymp 1$. In this case, $h_{\mu,W}^2\asymp n^{-2/(d+4)}$ and $h_{C,W}^2\asymp n^{-1/(d+2)}$, which are typical nonparametric rates. If $N_i\asymp n^{d/4}$ uniformly in $i$ and $n_C\asymp n$, then $a_{\mu,W}\asymp n^{d/4}$ and $a_{C,W}\asymp n^{d/2}$. In this case, $h_{\mu,W}^2,h_{C,W}^2\asymp n^{-1/2}$. Now, we state the rates for OBS.

\begin{corollary}\label{cor:obs_rates}
Under \eqref{eq:scheme_bandwidths}, \eqref{eq:sufficient} and Assumptions~\ref{ass:sampling}--\ref{ass:weights} with the OBS scheme, it holds that
\[
    \|\hat\mu-\mu\|_2=O_p(h_{\mu,O}^2),\quad
    \|\hat C-C\|_2=O_p(h_{\mu,O}^2+h_{C,O}^2),\quad|\hat\lambda_k-\lambda_k|=O_p(h_{\mu,O}^2+h_{C,O}^2),\quad k\geq1.
\]
If Assumption~\ref{ass:eigengap} further holds, then $\max_{1\le k\le K^*}\|\hat\phi_k-\phi_k\|_2=O_p(h_{\mu,O}^2+h_{C,O}^2)$.
\end{corollary}

Now, we state the rates for SUBJ.

\begin{corollary}\label{cor:subj_rates}
Assume that $n_C\cdot h_{\mu,S}^d$ is bounded away from zero. Under \eqref{eq:scheme_bandwidths}, \eqref{eq:sufficient} and Assumptions~\ref{ass:sampling}--\ref{ass:weights} with  the SUBJ scheme, it holds that
\[
    \|\hat\mu-\mu\|_2=O_p(h_{\mu,S}^2),\quad
    \|\hat C-C\|_2=O_p(h_{\mu,S}^2+h_{C,S}^2),\quad|\hat\lambda_k-\lambda_k|=O_p(h_{\mu,S}^2+h_{C,S}^2),\quad k\geq1.
\]
If Assumption~\ref{ass:eigengap} further holds, then $\max_{1\le k\le K^*}\|\hat\phi_k-\phi_k\|_2=O_p(h_{\mu,S}^2+h_{C,S}^2)$.
\end{corollary}

With the bandwidth rates in \eqref{eq:scheme_bandwidths}, the phase transition
is determined by the effective frequencies. Increasing these
frequencies improves the mean and covariance error bounds until $a_{\mu,W}$ and $a_{C,W}$ reach
orders $n^{d/4}$ and $n_C^{d/2}$, respectively. 
At these thresholds, $h_{\mu,W}^2$ and
$h_{C,W}^2$ have orders $n^{-1/2}$ and $n_C^{-1/2}$, respectively.
Beyond these thresholds, increasing the number of observations per subject no longer improves the orders of the error bounds.

\begin{remark}\label{rmk3}
With unequal observation counts, the arithmetic averages used by OBS and the harmonic averages used by SUBJ can behave differently. For example, suppose a fixed proportion of subjects have $N_i=2$ and the remaining subjects have $N_i=b_n\to\infty$. The OBS frequencies $a_{\mu,O}$ and $a_{C,O}$ are then of order $b_n$ and $b_n^2$, respectively, whereas the SUBJ frequencies $a_{\mu,S}$ and $a_{C,S}$ remain bounded. Therefore, OBS can enter the dense regime while SUBJ remains sparse.
\end{remark}

\subsection{Uniform results}\label{sec:theory_rates}

Let $\|\cdot\|_\infty$ denote the supremum norm on the relevant domain.

\begin{assumption}\label{ass:uniform_moments}
For some $\beta\ge2$, it holds that
\[
    \E\!\left[\sup_{s\in\mathcal M}|X_i(s)|^{2\beta}\right]<\infty,
    \quad \E(|\epsilon_{ij}|^{2\beta})<\infty.
\]
\end{assumption}

The case $\beta=2$ requires only the fourth moments already assumed
in Assumption~\ref{ass:process}. Define
\begin{align*}
 b_\mu
 &:= (\log n)^{1-1/(2\beta)}
 \left[\sum_{i=1}^n(N_iw_i)^{2\beta}
       \min\left\{1,\frac{\log n}{N_i}\right\}^{2\beta-1}
 \right]^{1/(2\beta)},\\
 b_C
 &:= (\log n)^{1-1/\beta}
 \left[\sum_{i\in\mathcal I_C}((N_i)_2v_i)^{\beta}
       \min\left\{1,\frac{\log n}{N_i}\right\}^{2\beta}
 \right]^{1/\beta}\\
 &\;\quad +(\log n)^{1-1/(2\beta)}
 \left[\sum_{i\in\mathcal I_C}((N_i)_2v_i)^{2\beta}
       \min\left\{1,\frac{\log n}{N_i}\right\}^{4\beta-2}
 \right]^{1/(2\beta)}
\end{align*}
and write
\begin{align*}
 r_{\mu,\infty}
 &:=h_\mu^2+\sqrt{\log n}\cdot\rho_{\mu,2}
   +h_\mu^{-d/(2\beta)}\left(\frac{\log n}{n}\right)^{1-1/(2\beta)}+b_\mu h_\mu^{-d},\\
 r_{C,\infty}
 &:=\left(r_{\mu,\infty}+h_C^2+\sqrt{\log n}\cdot\rho_{C,2}
   +h_C^{-d/\beta}\left(\frac{\log n}{n_C}\right)^{1-1/(2\beta)}
       +b_C h_C^{-2d}\right)\cdot(1+r_{\mu,\infty}).
\end{align*}

\begin{theorem}\label{thm:mean_rate}
Under Assumptions~\ref{ass:sampling}--\ref{ass:weights} and
\ref{ass:uniform_moments} and the mean-side conditions in
Assumption~\ref{ass:bandwidth}, it holds that
\[
    \|\hat\mu-\mu\|_\infty=O_p(r_{\mu,\infty}).
\]
\end{theorem}

\begin{theorem}\label{thm:cov_rate}
Under Assumptions~\ref{ass:sampling}--\ref{ass:bandwidth} and
\ref{ass:uniform_moments}, it holds that
\[
    \|\hat C-C\|_\infty=O_p(r_{C,\infty}).
\]
\end{theorem}

For $d=1$ and $n_C=n$, the terms
$h_\mu^2+\sqrt{\log n}\cdot\rho_{\mu,2}$ and
$h_\mu^2+h_C^2+\sqrt{\log n}\cdot(\rho_{\mu,2}+\rho_{C,2})$
in Theorems~\ref{thm:mean_rate} and~\ref{thm:cov_rate}
match the respective uniform convergence rates of \citet{ZhangWang2016}. 
Their rates hold under additional bandwidth conditions, whereas ours include the additional error terms without imposing those conditions.

\begin{theorem}\label{thm:uniform_operator_eigen}
Under the condition $r_{C,2}\to0$ and Assumptions~\ref{ass:sampling}--\ref{ass:uniform_moments},
$\phi_k$ and $\hat\phi_k$ admit continuous
representatives for $1\le k\le K^*$, the latter with probability tending to one. Choosing these representatives, it holds that
\[
    \max_{1\le k\le K^*}\|\hat\phi_k-\phi_k\|_\infty
       =O_p(r_{C,\infty}).
\]
\end{theorem}

Now, we specify the rates for the OBS and SUBJ schemes. For $W\in\{O,S\}$, let $b_{\mu,W}$ and $b_{C,W}$ denote $b_\mu$
and $b_C$ evaluated under scheme $W$. 
We consider bandwidths satisfying
\begin{align}\label{eq:uniform_scheme_bandwidths}
\begin{split}
 h_{\mu,W}
 &\asymp\max\left\{
       \left(\frac{\log n}{n}\right)^{1/4},
       \left(\frac{\log n}{na_{\mu,W}}\right)^{1/(d+4)},
       \left(\frac{\log n}{n}\right)^{(2\beta-1)/(4\beta+d)},
       b_{\mu,W}^{1/(d+2)}\right\},\\
 h_{C,W}
 &\asymp\max\left\{
       \left(\frac{\log n}{n_C}\right)^{1/4},
       \left(\frac{\log n}{n_Ca_{C,W}}\right)^{1/(2d+4)},
       \left(\frac{\log n}{n_C}\right)^{(2\beta-1)/(4\beta+2d)},
       b_{C,W}^{1/(2d+2)}\right\}.
\end{split}
\end{align}

\begin{corollary}\label{cor:uniform_schemes}
Assume that $n_C/\log n\to\infty$. Under \eqref{eq:uniform_scheme_bandwidths} and Assumptions~\ref{ass:sampling}--\ref{ass:weights} and \ref{ass:uniform_moments} with scheme $W\in\{O,S\}$, it holds that
\[
 \|\hat\mu-\mu\|_\infty=O_p(h_{\mu,W}^2),\quad
 \|\hat C-C\|_\infty=O_p(h_{\mu,W}^2+h_{C,W}^2).
\]
If Assumption~\ref{ass:eigengap} further holds, then continuous representatives of $\phi_k$ and $\hat\phi_k$ satisfy
$$\max_{1\leq k\le K^*}\|\hat\phi_k-\phi_k\|_\infty=O_p(h_{\mu,W}^2+h_{C,W}^2).$$
\end{corollary}

\begin{remark}\label{specific rates}
For bounded observation counts and $n_C\asymp n$, both schemes give
\begin{align*}
 \|\hat\mu-\mu\|_\infty
 =O_p\!\left(\left(\frac{\log n}{n}\right)^{2/(d+4)}\right),\quad
 \|\hat C-C\|_\infty
 =O_p\!\left(\left(\frac{\log n}{n}\right)^{1/(d+2)}\right),
\end{align*}
provided that $\beta\ge\max\{2,1+d/4\}$ and $\beta\ge d+2$, respectively. For $n_C\asymp n$ and $\beta>\max\{2,1+d/2\}$,
the uniform bounds exhibit a phase transition similar to that
of the $L^2$ bounds.
For OBS, and for SUBJ when $\max_i N_i/\min_i N_i=O(1)$,
once $a_{\mu,W}$ and $a_{C,W}$ reach orders
$(n/\log n)^{d/4}$ and $(n/\log n)^{d/2}$, respectively,
the bandwidths in
\eqref{eq:uniform_scheme_bandwidths} have order
$((\log n)/n)^{1/4}$, and
\[
\|\hat\mu-\mu\|_\infty=O_p\!\left(\sqrt{\frac{\log n}{n}}\right),\quad\|\hat C-C\|_\infty=O_p\!\left(\sqrt{\frac{\log n}{n}}\right).
\]
Since the first terms in \eqref{eq:uniform_scheme_bandwidths} bound $h_{\mu,W}$ and $h_{C,W}$ from below by constant multiples of $((\log n)/n)^{1/4}$, increasing the observation counts beyond these thresholds does not improve the orders of the bounds in Corollary~\ref{cor:uniform_schemes}.
\end{remark}

\section{Simulation studies}
\label{sec:simulation}

We examine the performance of the proposed method when $\mathcal M=\mathbb{S}^1=\{(\cos(\theta),\sin(\theta))\tran:0\leq\theta<2\pi\}$ and when $\mathcal M=\mathbb{S}^2=\{(\cos(\varphi)\sin(\vartheta),\sin(\varphi)\sin(\vartheta),\cos(\vartheta))\tran:0\leq\varphi<2\pi,0\leq\vartheta\leq\pi\}$. In both settings, we generate observations according to
\begin{align*}
Y_{ij}=X_i(s_{ij})+\epsilon_{ij},
\quad
\epsilon_{ij}\sim N(0,0.5^2).
\end{align*}
The observation locations $s_{ij}$ are sampled from the uniform distribution on the corresponding manifold. We consider sample sizes $n\in\{100,400\}$ and use $R=100$ Monte Carlo replications.

We compare the proposed method with a naive coordinate-based method. When $\mathcal M=\mathbb{S}^1$, the naive method treats $[0,2\pi)$ as a Euclidean interval and uses $\theta$ as the observation coordinate. When $\mathcal M=\mathbb{S}^2$, it treats $[0,2\pi)\times[0,\pi)$ as a Euclidean rectangle and uses $(\varphi,\vartheta)\tran$ as the observation coordinate vector. In both cases, the naive method applies standard Euclidean smoothing based on the Euclidean distance in the chosen coordinates rather than the geodesic distance. 

For both methods, we use the Epanechnikov-type kernel
\begin{align*}
    K(r)
    :=c_d\cdot(1-r^2)\cdot\mathbf 1\{0\le r\le1\},
\end{align*}
where $c_d$ is a positive constant chosen to satisfy the normalization
condition \eqref{kernel normalization}, and $\mathbf 1$ denotes the indicator function. The mean and covariance bandwidths are selected sequentially by 5-fold cross-validation on the common grid $\mathcal H_\mu=\mathcal H_C=\{0.1,0.2,\ldots,3\}$. Specifically, let $\mathcal V_1,\ldots,\mathcal V_5$
be a partition of the subjects into five folds. We first select
the mean bandwidth by 5-fold cross-validation and denote the selected
value by $\hat h_\mu$. For each fold $q$ and candidate covariance
bandwidth $h_C$, let $\hat\mu^{(-q)}$ and $\hat C_{h_C}^{(-q)}$
denote the mean and covariance estimators computed without subjects in $\mathcal V_q$ and with the mean bandwidth fixed
at $\hat h_\mu$. Define the validation residuals
\begin{align*}
    R_{ij}^{(-q)}
    :=
    Y_{ij}-\hat\mu^{(-q)}(s_{ij}),
    \quad i\in\mathcal V_q.
\end{align*}
We then select the covariance bandwidth that minimizes
\begin{align*}
    \sum_{q=1}^{5}
    \sum_{i\in\mathcal I_C\cap\mathcal V_q}
    v_i
    \sum_{1\le j\ne k\le N_i}
    \left(
        R_{ij}^{(-q)}R_{ik}^{(-q)}
        -
        \hat C_{h_C}^{(-q)}(s_{ij},s_{ik})
    \right)^{2}.
\end{align*}

We evaluate the mean integrated squared errors of $\mu$ and $\phi_k$ with respect to $V$: 
\begin{align*}
    \mathrm{MISE}_\mu
    &:=
    \frac{1}{R}\sum_{r=1}^R\int_{\mathcal M}
    (\hat\mu^{[r]}(u)-\mu(u))^2\,\mathrm{d}V(u),\\
    \mathrm{MISE}_{\phi_k}
    &:=
    \frac{1}{R}\sum_{r=1}^R\int_{\mathcal M}
    (\hat\phi_k^{[r]}(u)-\phi_k(u))^2\,\mathrm{d}V(u).
\end{align*}
Here, $\hat\mu^{[r]}$ and $\hat\phi_k^{[r]}$ are the estimated mean function and eigenfunctions obtained in the $r$th Monte Carlo replication. The Riemannian volume measure $V$ equals the arc-length measure when $\mathcal M=\mathbb{S}^1$ and equals the surface-area measure when $\mathcal M=\mathbb{S}^2$.

\subsection{Case of $\mathbb{S}^1$}
\label{subsec:sim_circle}
We generate $X_i$ from
\begin{align*}
    X_i(s)
    =
    \mu(s)+\sum_{k=1}^{2}\xi_{ik}\phi_k(s),
\end{align*}
where
\begin{align*}
    \mu(s)=
    2\sin(\theta)+0.5\cos(2\theta),\quad
    \phi_1(s)=
    \frac{1}{\sqrt{\pi}}\sin(\theta),\quad\phi_2(s) =
    \frac{1}{\sqrt{\pi}}\cos(\theta)
\end{align*}
for $s=(\cos(\theta),\sin(\theta))\tran$, and $\xi_{ik}$ are independently generated from
\begin{align*}
    \xi_{i1}\sim N(0,4),
    \quad
    \xi_{i2}\sim N(0,2).
\end{align*}
To examine the sparse-to-dense transition, we consider the four sampling designs:
\begin{align*}
    \text{Sparse:}\
    &N_i\sim\mathrm{Unif}\{2,4,6,8\},\\
    \text{Dense-Boundary:}\
    &N_i\sim\mathrm{Unif}\left\{5\left\lceil n^{1/4}\right\rceil,\ldots,10\left\lceil n^{1/4}\right\rceil\right\},\\
    \text{Mixed-I:}\
    &N_i\sim
    \begin{cases}
    \mathrm{Unif}\{2,4,6,8\}\,&\text{with probability 0.5},\\
    \mathrm{Unif}\left\{5\left\lceil n^{1/4}\right\rceil,\ldots,10\left\lceil n^{1/4}\right\rceil\right\}\,&\text{with probability 0.5},
    \end{cases}\\
    \text{Mixed-II:}\
    &N_i\sim
    \begin{cases}
    \mathrm{Unif}\{2,4,6,8\}\,&\text{with probability}\,0.95,\\
    \mathrm{Unif}\left\{5\left\lceil n^{1/4}\right\rceil,\ldots,10\left\lceil n^{1/4}\right\rceil\right\}\,&\text{with probability}\,0.05,
    \end{cases}
\end{align*}
where $\mathrm{Unif}\{5\lceil n^{1/4}\rceil,\ldots,10\lceil n^{1/4}\rceil\}$ denotes the discrete uniform distribution with minimum $5\lceil n^{1/4}\rceil$ and maximum $10\lceil n^{1/4}\rceil$. Note that the Dense-Boundary design corresponds to the phase transition threshold $\bar{N}\asymp n^{1/4}$ for $d=1$.

\begin{table}[!t]
\centering
\caption{
MISE for the mean and eigenfunctions defined on $\mathbb{S}^1$. OBS and SUBJ denote the proposed estimator under the OBS and SUBJ schemes, respectively.
}
\label{tab:s1_main}
\resizebox{\textwidth}{!}{%
\begin{tabular}{llcccccccccc}
\toprule
Regime & $n$ & $\E(\bar{N})$
& \multicolumn{3}{c}{MISE of $\hat\mu$}
& \multicolumn{3}{c}{MISE of $\hat\phi_1$}
& \multicolumn{3}{c}{MISE of $\hat\phi_2$} \\
\cmidrule(lr){4-6}
\cmidrule(lr){7-9}
\cmidrule(lr){10-12}
&&&
OBS & SUBJ & Naive
& OBS & SUBJ & Naive
& OBS & SUBJ & Naive \\
\midrule
\multirow{2}{*}{Sparse}
& 100 & 5.0
& 0.239 & 0.259 & 0.273
& 0.090 & 0.085 & 0.095
& 0.117 & 0.158 & 0.110 \\
& 400 & 5.0
& 0.063 & 0.073 & 0.078
& 0.021 & 0.025 & 0.028
& 0.029 & 0.043 & 0.032 \\
\midrule
Dense-
& 100 & 30.0
& 0.097 & 0.096 & 0.109
& 0.037 & 0.036 & 0.040
& 0.039 & 0.039 & 0.039 \\
Boundary
& 400 & 37.5
& 0.024 & 0.023 & 0.028
& 0.008 & 0.007 & 0.009
& 0.009 & 0.008 & 0.009 \\
\midrule
\multirow{2}{*}{Mixed-I}
& 100 & 17.5
& 0.147 & 0.172 & 0.170
& 0.085 & 0.061 & 0.083
& 0.093 & 0.097 & 0.089 \\
& 400 & 21.25
& 0.042 & 0.054 & 0.049
& 0.017 & 0.014 & 0.019
& 0.019 & 0.024 & 0.019 \\
\midrule
\multirow{2}{*}{Mixed-II}
& 100 & 6.25
& 0.249 & 0.256 & 0.284
& 0.228 & 0.100 & 0.224
& 0.246 & 0.168 & 0.218 \\
& 400 & 6.625
& 0.076 & 0.071 & 0.088
& 0.128 & 0.025 & 0.121
& 0.134 & 0.041 & 0.122 \\
\bottomrule
\end{tabular}%
}
\end{table}

Table~\ref{tab:s1_main} shows that MISE decreases with $n$ across all methods and sampling designs, with the proposed estimator under OBS mostly outperforming the naive estimator.
Consistent with Remark~\ref{rmk1}, OBS mostly outperforms SUBJ under the Sparse design, whereas SUBJ offers modest gains under the Dense-Boundary design.
Under the Mixed-I design, OBS achieves lower MISE than SUBJ for the mean function, consistent with the distinction between arithmetic and harmonic effective frequencies discussed in Remark~\ref{rmk3}.
Under the Mixed-II design, SUBJ substantially outperforms OBS, illustrating the benefit of limiting the disproportionate influence of a small group of heavily sampled subjects, as discussed in Remark~\ref{rmk2}.

\subsection{Case of $\mathbb{S}^2$}
\label{subsec:sim_sphere}

We generate $X_i$ from
\begin{align*}
    X_i(s)
    =
    \mu(s)+\sum_{k=1}^{3}\xi_{ik}\phi_k(s),
\end{align*}
where
\begin{align*}
    \mu(s)
    &=
    1.5\cos(\vartheta)+0.75\cos(\varphi)\sin(\vartheta)\sin(\varphi)\sin(\vartheta),\\
    \phi_1(s)
    &=
    \sqrt{\frac{3}{4\pi}}\cos(\vartheta),
    \quad
    \phi_2(s)
    =
    \sqrt{\frac{3}{4\pi}}\sin(\varphi)\sin(\vartheta),
    \quad
    \phi_3(s)
    =
    \sqrt{\frac{3}{4\pi}}\cos(\varphi)\sin(\vartheta)
\end{align*}
for $s=(\cos(\varphi)\sin(\vartheta),\sin(\varphi)\sin(\vartheta),\cos(\vartheta))\tran$, and $\xi_{ik}$ are independently generated from
\begin{align*}
    \xi_{i1}\sim N(0,4),
    \quad
    \xi_{i2}\sim N(0,2),
    \quad
    \xi_{i3}\sim N(0,1).
\end{align*}
We again consider the four sampling designs:
\begin{align*}
    \text{Sparse:}\
    &N_i\sim\mathrm{Unif}\{12,14,16,18\},\\
    \text{Dense-Boundary:}\
    &N_i\sim\mathrm{Unif}\left\{5\left\lceil n^{1/2}\right\rceil,\ldots,10\left\lceil n^{1/2}\right\rceil\right\},\\
    \text{Mixed-I:}\
    &N_i\sim
    \begin{cases}
    \mathrm{Unif}\{12,14,16,18\}\,&\text{with probability 0.5},\\
    \mathrm{Unif}\left\{5\left\lceil n^{1/2}\right\rceil,\ldots,10\left\lceil n^{1/2}\right\rceil\right\}\,&\text{with probability 0.5}.
    \end{cases}\\
    \text{Mixed-II:}\
    &N_i\sim
    \begin{cases}
    \mathrm{Unif}\{12,14,16,18\}\,&\text{with probability}\,0.95,\\
    \mathrm{Unif}\left\{5\left\lceil n^{1/2}\right\rceil,\ldots,10\left\lceil n^{1/2}\right\rceil\right\}\,&\text{with probability}\,0.05,
    \end{cases}
\end{align*}

\begin{table}[!t]
\centering
\caption{
MISE for the mean and eigenfunctions defined on $\mathbb{S}^2$. OBS and SUBJ denote the proposed estimator under the OBS and SUBJ schemes, respectively.
}
\label{tab:s2_main}
\resizebox{\textwidth}{!}{%
\begin{tabular}{llccccccccccccc}
\toprule
Regime & $n$ & $\E(\bar{N})$
& \multicolumn{3}{c}{MISE of $\hat\mu$}
& \multicolumn{3}{c}{MISE of $\hat\phi_1$}
& \multicolumn{3}{c}{MISE of $\hat\phi_2$}
& \multicolumn{3}{c}{MISE of $\hat\phi_3$} \\
\cmidrule(lr){4-6}
\cmidrule(lr){7-9}
\cmidrule(lr){10-12}
\cmidrule(lr){13-15}
&&&
OBS & SUBJ & Naive
& OBS & SUBJ & Naive
& OBS & SUBJ & Naive
& OBS & SUBJ & Naive \\
\midrule
\multirow{2}{*}{Sparse}
& 100 & 15.0
& 0.240 & 0.239 & 0.363
& 0.066 & 0.065 & 0.187
& 0.144 & 0.136 & 0.273
& 0.175 & 0.170 & 0.231 \\
& 400 & 15.0
& 0.075 & 0.076 & 0.118
& 0.019 & 0.018 & 0.036
& 0.034 & 0.032 & 0.064
& 0.045 & 0.044 & 0.075 \\
\midrule
Dense-
& 100 & 75.0
& 0.119 & 0.118 & 0.160
& 0.037 & 0.037 & 0.074
& 0.075 & 0.057 & 0.120
& 0.067 & 0.048 & 0.089 \\
Boundary
& 400 & 150.0
& 0.029 & 0.029 & 0.040
& 0.009 & 0.009 & 0.012
& 0.015 & 0.014 & 0.021
& 0.014 & 0.014 & 0.020 \\
\midrule
\multirow{2}{*}{Mixed-I}
& 100 & 45.0
& 0.155 & 0.173 & 0.206
& 0.120 & 0.062 & 0.188
& 0.164 & 0.102 & 0.238
& 0.094 & 0.112 & 0.125 \\
& 400 & 82.5
& 0.045 & 0.056 & 0.061
& 0.015 & 0.012 & 0.021
& 0.026 & 0.024 & 0.037
& 0.022 & 0.032 & 0.030 \\
\midrule
\multirow{2}{*}{Mixed-II}
& 100 & 18.0
& 0.236 & 0.227 & 0.331
& 0.190 & 0.060 & 0.391
& 0.342 & 0.131 & 0.532
& 0.272 & 0.171 & 0.341 \\
& 400 & 21.75
& 0.109 & 0.078 & 0.151
& 0.144 & 0.017 & 0.274
& 0.255 & 0.030 & 0.382
& 0.170 & 0.042 & 0.197 \\
\bottomrule
\end{tabular}%
}
\end{table}

Table~\ref{tab:s2_main} shows that the proposed estimators substantially outperform the naive estimator.
OBS and SUBJ yield very similar MISE under the Sparse design, whereas the gains from SUBJ under the Dense-Boundary design are most evident for the second and third eigenfunctions at $n=100$.
As in the case of $\mathbb{S}^1$, OBS achieves lower MISE than SUBJ for the mean function under the Mixed-I design while SUBJ substantially outperforms OBS under the Mixed-II design.
These results highlight the importance of choosing an appropriate weighting scheme when observation counts vary across subjects.

\section{Real-data applications}\label{sec:realdata}

\subsection{OCT data analysis}\label{subsec:realdata_oct}

We analyze high-resolution optical coherence tomography (OCT) data on retinal nerve fiber layer (RNFL) thickness available at \url{https://journals.plos.org/plosone/article?id=10.1371/journal.pone.0292915} and analyzed by \citet{AliEtAl2023}. RNFL thickness is recorded at $256$ equally spaced angular locations along a peripapillary scan circle. Hence, the indexing domain is $\mathbb{S}^1$. Because measurements from both eyes are available for some participants, we restrict the analysis to left-eye profiles, thereby ensuring that each retained functional observation corresponds to a distinct participant. The resulting sample comprises $1992$ profiles from normal Asian Indians eyes and $123$ profiles from glaucomatous Asian Indians eyes.

Since all retained profiles are observed on the same grid of $256$ angular locations, the OBS and SUBJ schemes coincide. We use the Epanechnikov-type kernel and select $h_\mu$ and $h_C$ using the same sequential 5-fold cross-validation procedure described in the simulation studies, with the grid $\mathcal H_\mu=\mathcal H_C=\{0.01,0.02,\ldots,1\}$. 

\begin{figure}[!t]
\centering
\includegraphics[width=0.45\textwidth]{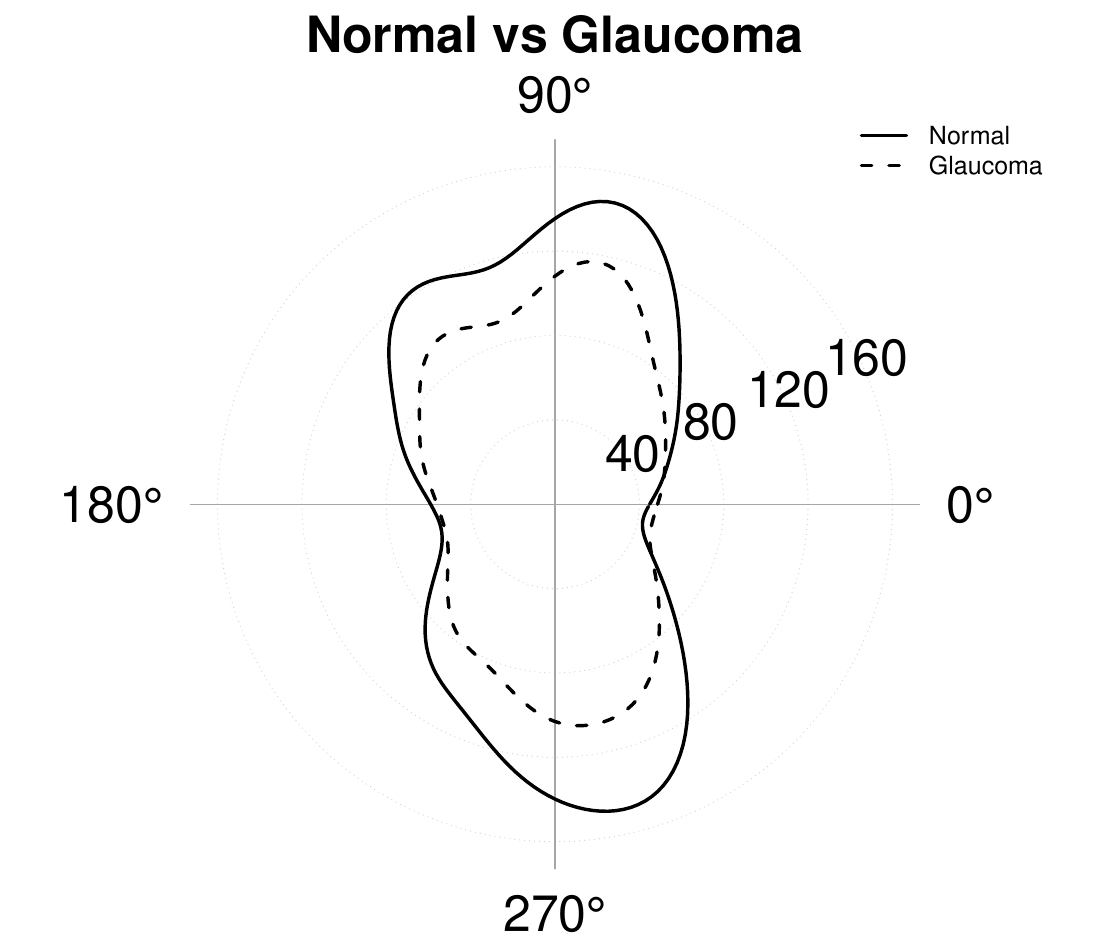}
\caption{
Estimated mean RNFL thickness profiles (in $\mu\mathrm{m}$) for the normal and glaucomatous groups.
}
\label{fig:oct_mean}
\end{figure}

Figure~\ref{fig:oct_mean} compares the estimated mean RNFL profiles for the normal and glaucomatous groups. Both estimated profiles exhibit the characteristic two-peak pattern. The estimated mean RNFL profile for the glaucomatous group is generally lower than that for the normal group, consistent with the well-established association between glaucoma and RNFL thinning. The largest pointwise gap between the two estimated mean profiles is approximately $42.8 \mu\mathrm{m}$ and occurs near $285^\circ$.

\begin{table}[!t]
\centering
\caption{
Estimated cumulative fraction of variance explained by the leading FPCs for the normal and glaucomatous groups.
}
\label{tab:oct_fve}
\begin{tabular}{lccccc}
\toprule
Group & $K=1$ & $K=2$ & $K=3$ & $K=5$ & $K=10$ \\
\midrule
Normal & 0.289 & 0.479 & 0.613 & 0.745 & 0.914 \\
Glaucoma & 0.540 & 0.692 & 0.778 & 0.873 & 0.966 \\
\bottomrule
\end{tabular}
\end{table}

Table~\ref{tab:oct_fve} reports the cumulative fraction of variance explained (FVE) by the leading FPCs. The reported FVE values indicate that the leading FPCs account for a larger proportion of the estimated total variance in the glaucomatous group than in the normal group. 

\begin{figure}[!t]
\centering
\includegraphics[width=0.475\textwidth]{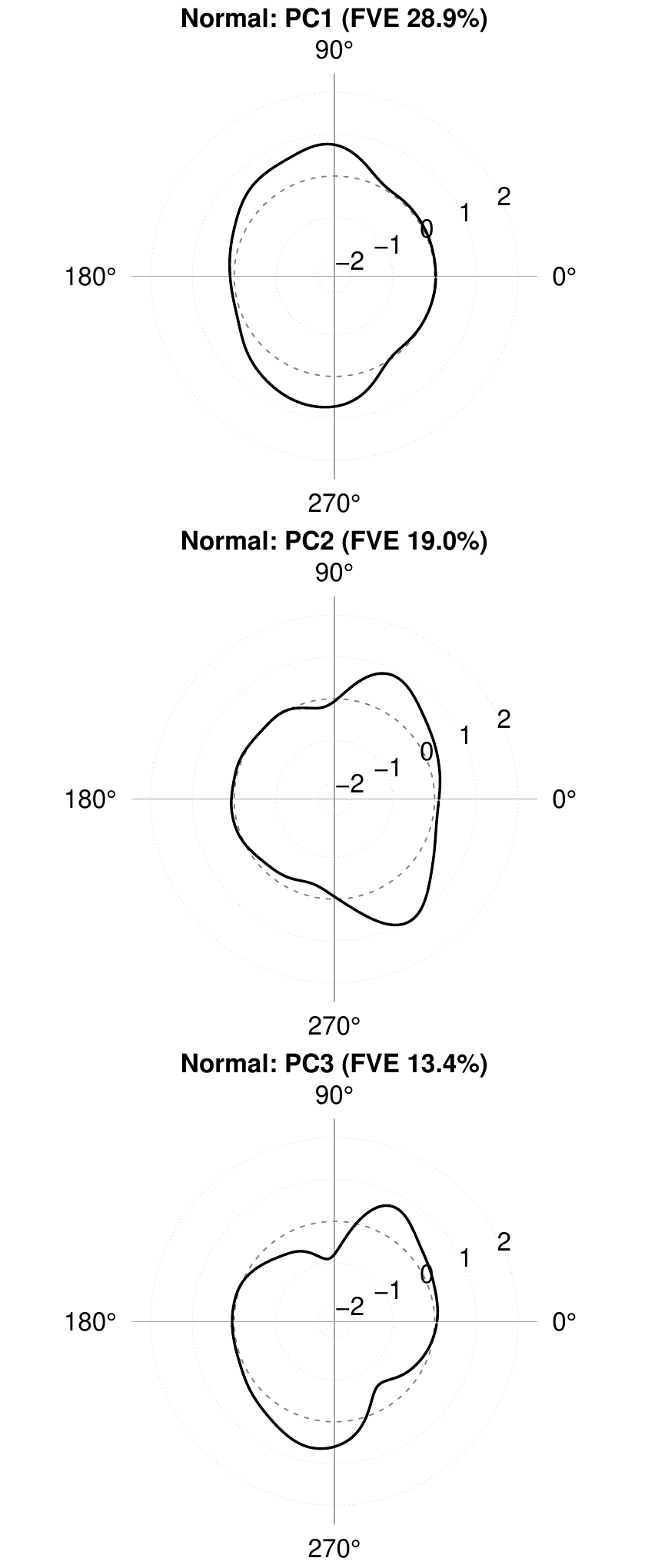}
\hfill
\includegraphics[width=0.475\textwidth]{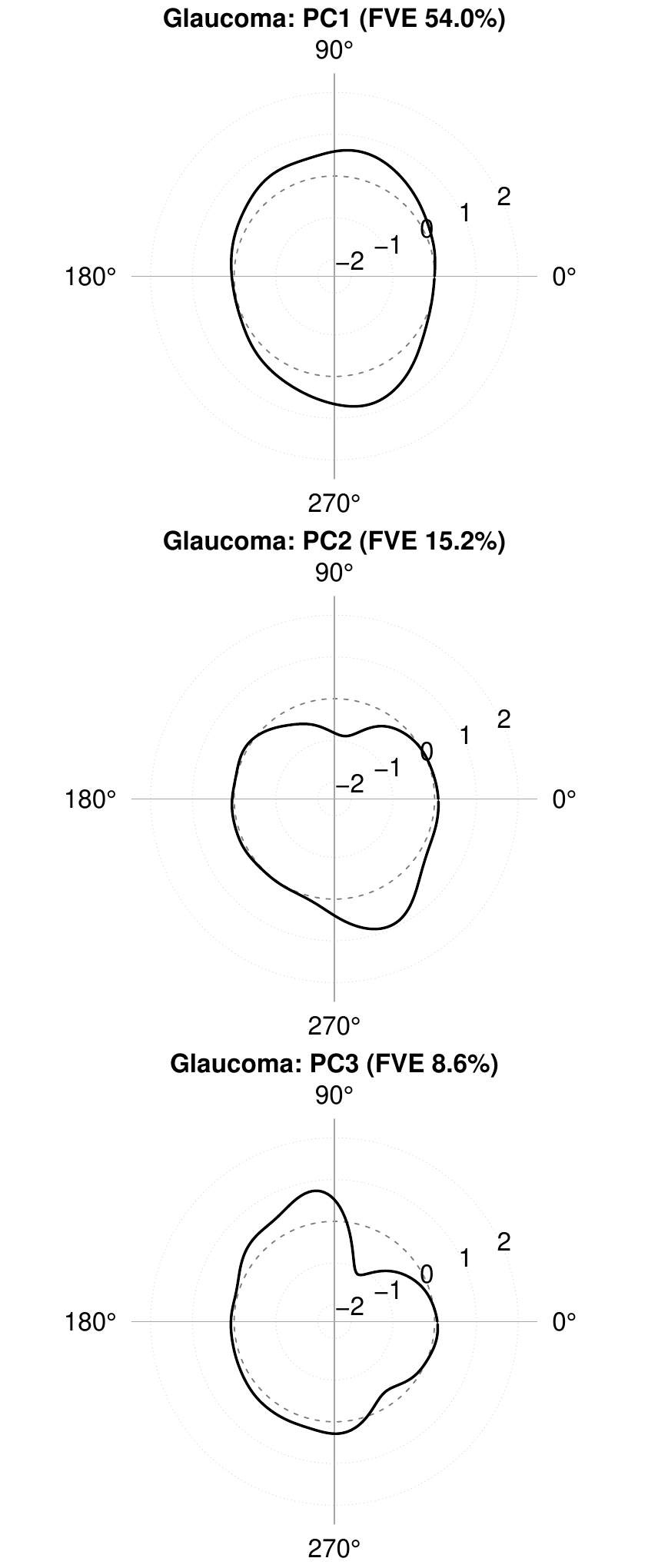}
\caption{
First three estimated eigenfunctions for the normal group (left) and the glaucomatous group (right).
}
\label{fig:oct_eigenfunctions}
\end{figure}

Figure~\ref{fig:oct_eigenfunctions} shows the first three $\hat{\phi}_k$ for both groups. In both groups, the first eigenfunction describes a broad mode of subject-to-subject variation involving the two peak regions, whereas the second and third eigenfunctions capture additional angular contrasts around these regions. At smaller angular positions, the latter eigenfunctions exhibit particularly pronounced differences in shape between the two groups.

\subsection{SONICOM data analysis}\label{subsec:realdata_sonicom}

We analyze head-related transfer functions (HRTFs) from the SONICOM dataset, available at \url{https://ecosystem.sonicom.eu/databases/3}. An HRTF characterizes how sound arriving from a given direction is filtered by a listener's head and outer ears. The detailed measurement protocol is described in \citet{EngelEtAl2023}. The indexing variable is therefore the sound-source direction, which can be represented as a point on $\mathbb{S}^2$. For each subject and sound-source direction, we extract a scalar HRTF value by computing the band-averaged log-magnitude response over the 3--8~kHz band. After this preprocessing, the data comprise observations from $399$ subjects, each with HRTF measurements at the same set of $793$ sound-source directions. We apply the proposed method to these data with $\mathcal M=\mathbb{S}^2$, using the Epanechnikov-type kernel and the common bandwidth grid $\mathcal H_\mu=\mathcal H_C=\{0.2,0.4,\ldots,1\}$.

\begin{figure}[!t]
\centering
\includegraphics[width=0.9\textwidth]{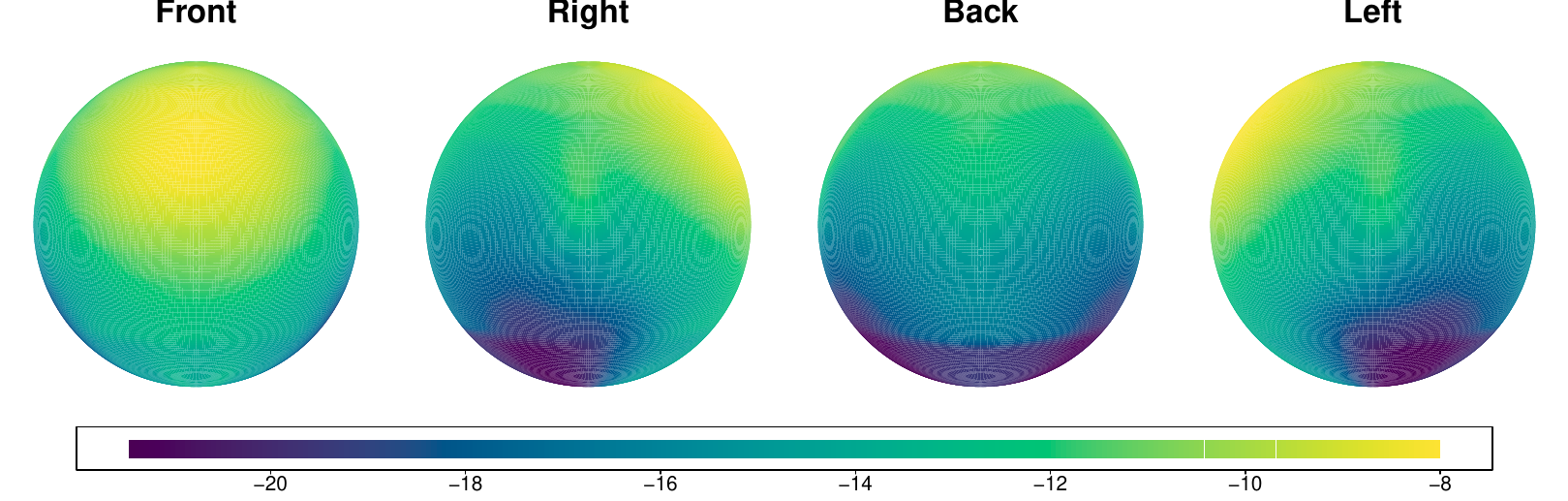}
\caption{Estimated mean HRTF for the SONICOM data, displayed on the sphere from four viewing directions.}
\label{fig:realdata_mean}
\end{figure}

Figure~\ref{fig:realdata_mean} displays the estimated mean HRTF. The estimated surface is shown from four viewing directions. It exhibits stronger attenuation at lower elevations and relatively larger values in upper and lateral sound-source directions.

\begin{figure}[!t]
\centering
\includegraphics[width=\textwidth]{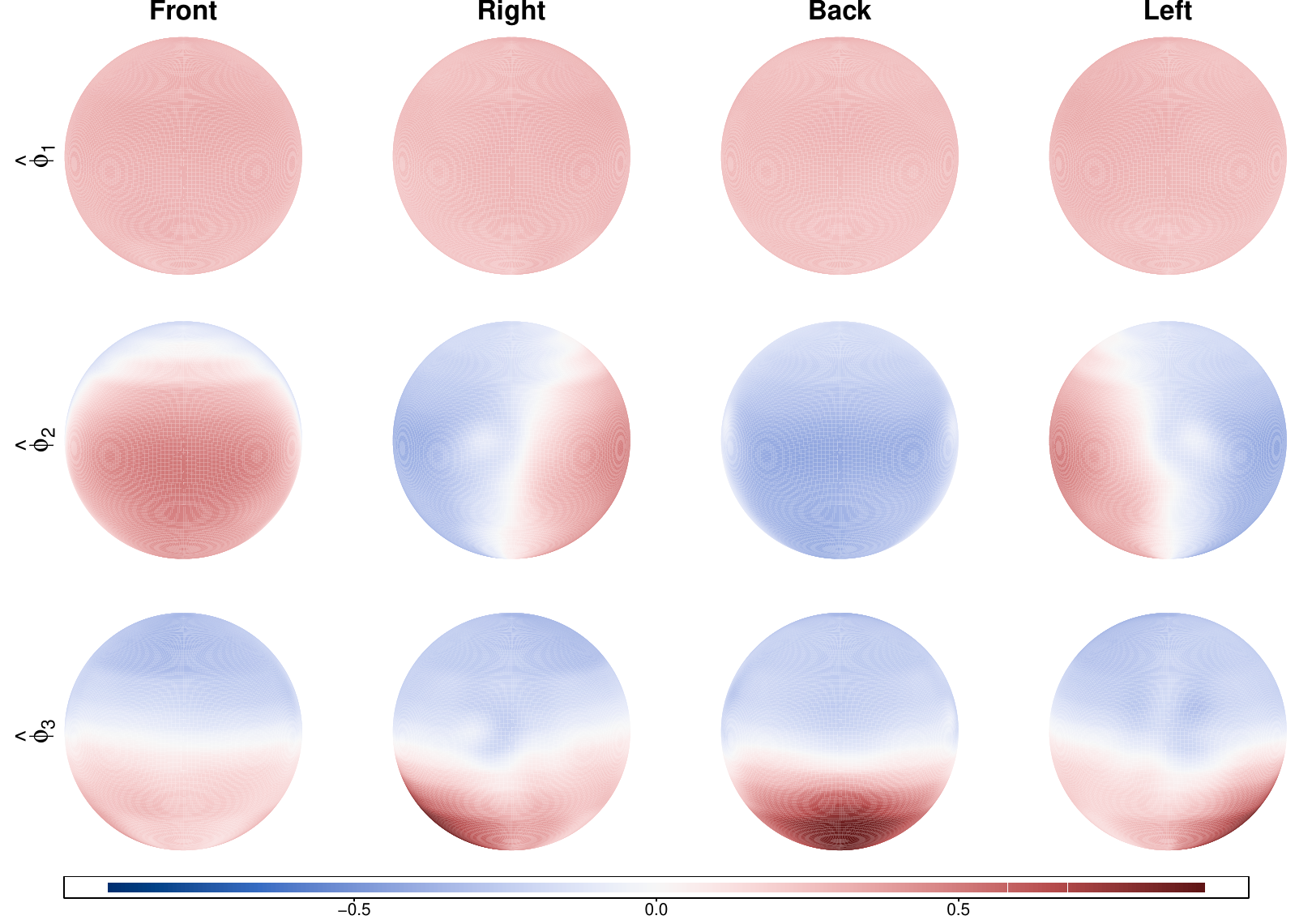}
\caption{
First three estimated eigenfunctions for the SONICOM data, displayed on the sphere from four viewing directions.
}
\label{fig:realdata_eigenfunctions}
\end{figure}

Figure~\ref{fig:realdata_eigenfunctions} shows the first three estimated eigenfunctions. The corresponding FPCs explain $75.9\%$, $9.7\%$ and $5.2\%$ of the estimated total variance, respectively. The first eigenfunction represents the dominant mode of variation and is nearly constant over the sphere. The second and third eigenfunctions primarily capture a lateral directional contrast and an elevation-related contrast, respectively. These leading eigenfunctions provide a compact summary of the principal modes of variation in the HRTFs.

\section*{Acknowledgements}

Chang Jun Im was supported by the National Research Foundation of Korea grant funded by the Korea government (MSIT) (No.~RS-2025-00515381). Jeong Min Jeon was supported by the National Research Foundation of Korea grant funded by the Korea government (MSIT) (No.~RS-2023-00211910).

\newpage

\appendix

\centerline{\Large Supplementary Material to}
\smallskip
\centerline{\large Functional Principal Component Analysis for
Manifold-Indexed Data}
\bigskip
\normalsize
\centerline{Chang Jun Im and Jeong Min Jeon}
\bigskip
\centerline{Seoul National University, South Korea}
\bigskip

This supplementary material contains all technical proofs to the main results. Throughout this material, we assume the kernel conditions given in Section \ref{sec:estimation method}. Let $(\mathrm{const.})$ denote a generic positive constant whose value may vary from line to line and is independent of $n$. For $r>0$, define
\begin{align*}
\omega_d&:=\operatorname{vol}_s(B_s(0_s,1)),\\
\overline B_s(0_s,r)&:=\{z\in T_s\mathcal M:\|z\|_s\leq r\},\\
\overline B_{\mathcal M}(s,r)&:=\{u\in\mathcal M:d_{\mathcal M}(s,u)\leq r\},\\
\mathcal D(r)&:=\{(s,u)\in\mathcal M\times\mathcal M:d_{\mathcal M}(s,u)<r\},\\
\overline{\mathcal D}(r)&:=\{(s,u)\in\mathcal M\times\mathcal M:d_{\mathcal M}(s,u)\le r\}.
\end{align*}

\section{Properties of volume density and kernel}

\begin{lemma}\label{lem:theta_bounds}
There exist constants $0<m_\theta<M_\theta<\infty$ such that
\begin{align}\label{eq:thetasbound}
    m_\theta\le \theta_s(u)\le M_\theta,
    \quad u\in\mathcal M,\quad s\in \overline B_{\mathcal M}\left(u,\frac{\iota_{\mathcal M}}{2}\right).
\end{align}
Additionally, there exists a constant $0\le L_\theta<\infty$ such that
\begin{align}\label{eq:thetasinversebound}
    \big|\theta_s(u)^{-1}-\theta_{s'}(u)^{-1}\big|
    \le L_\theta d_{\mathcal M}(s,s'),\quad
    u\in\mathcal M,\quad s,s'\in \overline B_{\mathcal M}\left(u,\frac{\iota_{\mathcal M}}{2}\right).
\end{align}
\end{lemma}

\begin{proof}
Since the map $(s,u)\mapsto\theta_s(u)$ is smooth and strictly positive on $\mathcal D(\iota_{\mathcal M})$, and $\overline{\mathcal D}(\iota_{\mathcal M}/2)\subset\mathcal D(\iota_{\mathcal M})$ is compact, \eqref{eq:thetasbound} holds by the extreme value theorem.

To prove \eqref{eq:thetasinversebound}, note that $(s,u)\mapsto\theta_s(u)^{-1}$ is also smooth on $\mathcal D(\iota_{\mathcal M})$. Since $\overline{\mathcal D}(3\iota_{\mathcal M}/4)$ is compact, the extreme value theorem implies that there exists a constant $0\le \widetilde L_\theta<\infty$ such that
\begin{align}\label{eq:thetainversegradient}
    \sup_{(s,u)\in\overline{\mathcal D}\left(3\iota_{\mathcal M}/4\right)}
    \left\|\nabla_s\left(\theta_s(u)^{-1}\right)\right\|_s
    \le \widetilde L_\theta,
\end{align}
where $\nabla_s$ denotes the Riemannian gradient with respect to the first argument. Fix $u\in\mathcal M$ and $s,s'\in\overline B_{\mathcal M}(u,\iota_{\mathcal M}/2)$. We first consider the case where $d_{\mathcal M}(s,s')\le\iota_{\mathcal M}/4$. Let $\gamma:[0,\ell]\to\mathcal M$ be a unit-speed minimizing geodesic from $s$ to $s'$, where $\ell:=d_{\mathcal M}(s,s')$. Then,
\begin{align*}
    d_{\mathcal M}(\gamma(x),u)
    \le d_{\mathcal M}(\gamma(x),s)+d_{\mathcal M}(s,u)
    \le \ell+\frac{\iota_{\mathcal M}}{2}
    \le \frac{3\iota_{\mathcal M}}{4}, \quad x\in[0,\ell].
\end{align*}
By \eqref{eq:thetainversegradient} and the fundamental theorem of calculus along $\gamma$,
\begin{align}\label{eq:thetainverseboundcase1}
    \big|\theta_s(u)^{-1}-\theta_{s'}(u)^{-1}\big|
    &\le \int_0^\ell
    \left\| \left.\nabla_s\left(\theta_s(u)^{-1}\right)\right|_{s=\gamma(x)} \right\|_{\gamma(x)}
    \|\dot\gamma(x)\|_{\gamma(x)}\,\mathrm{d}x
    \le \widetilde L_\theta d_{\mathcal M}(s,s').
\end{align}
We now consider the case where $d_{\mathcal M}(s,s')>\iota_{\mathcal M}/4$. By \eqref{eq:thetasbound},
\begin{align}\label{eq:thetainverseboundcase2}
    \big|\theta_s(u)^{-1}-\theta_{s'}(u)^{-1}\big|
    \le \frac{2}{m_\theta}
    \le \frac{8}{m_\theta\iota_{\mathcal M}}d_{\mathcal M}(s,s').
\end{align}
Combining \eqref{eq:thetainverseboundcase1} and \eqref{eq:thetainverseboundcase2}, \eqref{eq:thetasinversebound} follows with $L_\theta:=\max\{\widetilde L_\theta,8/(m_\theta\iota_{\mathcal M})\}$.
\end{proof}

For the next lemma, let $C^2(\mathcal M)$ denote the space of twice continuously differentiable functions on $\mathcal{M}$. For $\psi\in C^2(\mathcal M)$, let $\nabla\psi(s)$ and $\nabla^2\psi(s)$ denote its Riemannian gradient and Hessian at $s\in\mathcal M$, respectively. We may view $\nabla^2\psi(s)$ as a symmetric bilinear form on $T_s\mathcal M\times T_s\mathcal M$, with operator norm
\begin{align*}
    \|\nabla^2\psi(s)\|_{\mathrm{op}}
    :=\sup_{\substack{z,w\in T_s\mathcal M\\ \|z\|_s=\|w\|_s=1}}
    |\nabla^2\psi(s)(z,w)|.
\end{align*}
Note that $B_\psi:=\sup_{s\in\mathcal M}\|\nabla^2\psi(s)\|_{\mathrm{op}}<\infty$.

\begin{lemma}\label{lem:taylor}
For each $\psi\in C^2(\mathcal M)$ and
$0<h<\iota_{\mathcal M}/2$, there exists a remainder
$R_{\psi,h}(s,z)$ such that
\begin{align}\label{eq:taylor_part_2}
\begin{split}
    \psi(\Exp_s(hz))
    &=\psi(s)+h\langle\nabla\psi(s),z\rangle_s
    +R_{\psi,h}(s,z),\\
    |R_{\psi,h}(s,z)|
    &\le\frac{B_\psi}{2}h^2\|z\|_s^2,
    \quad s\in\mathcal M,\quad z\in\overline B_s(0_s,1).
\end{split}
\end{align}
\end{lemma}

\begin{proof}
Fix $s\in\mathcal M$ and $z\in\overline B_s(0_s,1)$, and define
\begin{align*}
    \gamma(x):=\Exp_s(xhz),\quad x\in[0,1].
\end{align*}
Then, $\gamma$ is a geodesic satisfying
$\dot\gamma(0)=hz$ and
$\|\dot\gamma(x)\|_{\gamma(x)}=h\|z\|_s$.
For $\eta(x):=\psi(\gamma(x))$, we obtain
\begin{align*}
    \eta'(0)&=h\langle\nabla\psi(s),z\rangle_s,\\
    \eta''(x)&=\nabla^2\psi(\gamma(x))
    (\dot\gamma(x),\dot\gamma(x)),\quad x\in(0,1).
\end{align*}
Taylor's formula with integral remainder gives the expansion
in \eqref{eq:taylor_part_2} with
\begin{align*}
    R_{\psi,h}(s,z)
    :=\int_0^1(1-x)\nabla^2\psi(\gamma(x))
    (\dot\gamma(x),\dot\gamma(x))\,\mathrm{d}x.
\end{align*}
Then,
\begin{align*}
    |R_{\psi,h}(s,z)|
    \le B_\psi h^2\|z\|_s^2\int_0^1(1-x)\,\mathrm{d}x
    =\frac{B_\psi}{2}h^2\|z\|_s^2,
\end{align*}
which proves the lemma.
\end{proof}

For the next lemma, let $C^2(\mathcal M\times\mathcal{M})$ denote the space of all twice continuously differentiable functions on $\mathcal M\times\mathcal{M}$.

\begin{lemma}\label{lem:intrinsic_bias}
For $\psi\in C^2(\mathcal M)$ and $0<h<\iota_{\mathcal M}/2$,
\begin{align}\label{eq:intrinsic_bias_single}
    \sup_{s\in\mathcal M} \left| \int_{\mathcal M} \mathcal L_{s,h}(u)\psi(u)\,\mathrm{d}V(u) - \psi(s) \right| \le (\mathrm{const.})h^2.
\end{align}
Similarly, for $\Psi\in C^2(\mathcal M\times\mathcal M)$ and $0<h<\iota_{\mathcal M}/2$,
\begin{align}\label{eq:intrinsic_bias_product}
    \sup_{(s,t)\in\mathcal M\times\mathcal M} \left| \int_{\mathcal M\times\mathcal M} \mathcal L_{s,h}(u)\mathcal L_{t,h}(v)\,\Psi(u,v) \,\mathrm{d}V(u)\,\mathrm{d}V(v) - \Psi(s,t) \right| \le (\mathrm{const.})h^2.
\end{align}
\end{lemma}

\begin{proof}
By the change of variables $u=\Exp_s(hz)$ and
\eqref{kernel normalization},
\begin{align}\label{eq:kernel_normalization}
\begin{split}
    \int_{\mathcal M}\mathcal L_{s,h}(u)\,\mathrm{d}V(u)
    &=1,\\
    \int_{B_s(0_s,1)}
    \langle\nabla\psi(s),z\rangle_sK(\|z\|_s)
    \,\mathrm{d}\operatorname{vol}_s(z)
    &=0.
\end{split}
\end{align}
By Lemma~\ref{lem:taylor} and \eqref{eq:kernel_normalization},
\begin{align}\label{eq:kernel_second_order}
\begin{split}
    \int_{\mathcal M}\mathcal L_{s,h}(u)\psi(u)
    \,\mathrm{d}V(u)-\psi(s)=\int_{B_s(0_s,1)}
    K(\|z\|_s)R_{\psi,h}(s,z)
    \,\mathrm{d}\operatorname{vol}_s(z).
\end{split}
\end{align}
By \Cref{lem:taylor}, \eqref{eq:kernel_second_order}
and the nonnegativity of $K$,
\begin{align}\label{eq:taylor_upper_bound_part1}
    \left|
    \int_{\mathcal M}\mathcal L_{s,h}(u)\psi(u)
    \,\mathrm{d}V(u)-\psi(s)
    \right|
    \le\frac{B_\psi}{2}h^2
    \int_{\mathbb R^d}\|\bm z\|^2K(\|\bm z\|)
    \,\mathrm{d}\bm z.
\end{align}
Taking the supremum over $s\in\mathcal M$ proves
\eqref{eq:intrinsic_bias_single}.

We now prove \eqref{eq:intrinsic_bias_product}.
Since $\Psi\in C^2(\mathcal M\times\mathcal M)$ and
$\mathcal M$ is compact, the Riemannian Hessians of the
functions $u\mapsto\Psi(u,v)$ and $v\mapsto\Psi(u,v)$
are uniformly bounded over
$(u,v)\in\mathcal M\times\mathcal M$.
Hence, \eqref{eq:taylor_upper_bound_part1} applies to these
functions with a common constant.
By \eqref{eq:kernel_normalization}, Fubini's theorem
and the nonnegativity of $K$,
\begin{align*}
    &\left|
    \int_{\mathcal M\times\mathcal M}
    \mathcal L_{s,h}(u)\mathcal L_{t,h}(v)\Psi(u,v)
    \,\mathrm{d}V(u)\,\mathrm{d}V(v)-\Psi(s,t)
    \right|\\
    &\le
    \int_{\mathcal M}\mathcal L_{t,h}(v)
    \left|
    \int_{\mathcal M}\mathcal L_{s,h}(u)\Psi(u,v)
    \,\mathrm{d}V(u)-\Psi(s,v)
    \right|\,\mathrm{d}V(v)\\
    &\quad+
    \left|
    \int_{\mathcal M}\mathcal L_{t,h}(v)\Psi(s,v)
    \,\mathrm{d}V(v)-\Psi(s,t)
    \right|\\
    &\le(\mathrm{const.})h^2.
\end{align*}
Taking the supremum over
$(s,t)\in\mathcal M\times\mathcal M$ proves
\eqref{eq:intrinsic_bias_product}.
\end{proof}

\begin{lemma}\label{supp:kernel}
For $0<h<\iota_{\mathcal M}/4$,
\begin{align} \label{eq:supp_kernel_bound}
    \sup_{s,u\in\mathcal M}|\mathcal L_{s,h}(u)|&\le (\mathrm{const.})h^{-d}.
\end{align}
Additionally, for $0<h<\iota_{\mathcal M}/4$,
\begin{align} \label{eq:supp_kernel_lipschitz}
    \sup_{u\in\mathcal M}|\mathcal L_{s,h}(u)-\mathcal L_{s',h}(u)|&\le (\mathrm{const.})h^{-d-1}d_{\mathcal M}(s,s'), \quad s,s'\in\mathcal M.
\end{align}
\end{lemma}

\begin{proof}
\eqref{eq:supp_kernel_bound} follows from Lemma~\ref{lem:theta_bounds} and the boundedness of $K$.
To prove \eqref{eq:supp_kernel_lipschitz}, we consider two cases depending on the value of $d_{\mathcal M}(s,s')$. Suppose that $d_{\mathcal M}(s,s')\le h$. 
If both $\mathcal L_{s,h}(u)$ and $\mathcal L_{s',h}(u)$ are zero, then \eqref{eq:supp_kernel_lipschitz} is immediate. If either of them, say $\mathcal L_{s,h}(u)$, is positive, then
\begin{align}\label{distance_inequality}
    d_{\mathcal M}(s,u)\le h<\frac{\iota_{\mathcal M}}{4},\quad
    d_{\mathcal M}(s',u)\le2h<\frac{\iota_{\mathcal M}}{2}.
\end{align}
By Lemma~\ref{lem:theta_bounds}, \eqref{distance_inequality} and the Lipschitz continuity of $K$,
\begin{align*}
    &h^d|\mathcal L_{s,h}(u)-\mathcal L_{s',h}(u)|\\
    &=\left|
    \theta_s(u)^{-1}
    K\left(\frac{d_{\mathcal M}(s,u)}{h}\right)
    -\theta_{s'}(u)^{-1}
    K\left(\frac{d_{\mathcal M}(s',u)}{h}\right)
    \right|\\
    &=\left|
    \theta_s(u)^{-1}
    \left\{
    K\left(\frac{d_{\mathcal M}(s,u)}{h}\right)
    -K\left(\frac{d_{\mathcal M}(s',u)}{h}\right)
    \right\}
    +\left\{\theta_s(u)^{-1}-\theta_{s'}(u)^{-1}\right\}
    K\left(\frac{d_{\mathcal M}(s',u)}{h}\right)
    \right|\\
    &\le
    \theta_s(u)^{-1}
    \left|
    K\left(\frac{d_{\mathcal M}(s,u)}{h}\right)
    -K\left(\frac{d_{\mathcal M}(s',u)}{h}\right)
    \right|
    +\left|\theta_s(u)^{-1}-\theta_{s'}(u)^{-1}\right|
    \left|K\left(\frac{d_{\mathcal M}(s',u)}{h}\right)\right|\\
    &\le (\mathrm{const.}) h^{-1}
    |d_{\mathcal M}(s,u)-d_{\mathcal M}(s',u)|
    +(\mathrm{const.}) d_{\mathcal M}(s,s')\\
    &\le
    \left((\mathrm{const.}) h^{-1}+(\mathrm{const.})\right)
    d_{\mathcal M}(s,s')\\
    &\le(\mathrm{const.})h^{-1}d_{\mathcal M}(s,s').
\end{align*}
This gives \eqref{eq:supp_kernel_lipschitz}. 
If $d_{\mathcal M}(s,s')>h$, then \eqref{eq:supp_kernel_bound} gives
\begin{align*}
    |\mathcal L_{s,h}(u)-\mathcal L_{s',h}(u)|\le2(\mathrm{const.})h^{-d}\le2(\mathrm{const.})h^{-d-1}d_{\mathcal M}(s,s').
\end{align*}
This also gives \eqref{eq:supp_kernel_lipschitz}.
\end{proof}

\section{Proofs for \texorpdfstring{$L^2$}{L2} results}

\subsection{Lemmas for \texorpdfstring{$L^2$}{L2} results}

To state lemmas, define
\begin{align*}
    \hat f_\mu(s):=\sum_{i=1}^n Z_{\mu,i}(s),\quad
    \hat f_C(s,t):=\sum_{i\in\mathcal I_C}Z_{C,i}(s,t),
\end{align*}
where
\begin{align*}
    Z_{\mu,i}(s):=w_i\sum_{j=1}^{N_i}\mathcal L_{s,h_\mu}(s_{ij}),\quad
    Z_{C,i}(s,t):=v_i\sum_{1\le j\ne k\le N_i}
    \mathcal L_{s,h_C}(s_{ij})\mathcal L_{t,h_C}(s_{ik}).
\end{align*}
Also, define
\begin{align*}
    \mathcal E_\mu:=\left\{\inf_{s\in\mathcal M}\hat f_\mu(s)\ge\frac{c_f}{2}\right\},\quad
    \mathcal E_C:=\left\{\inf_{(s,t)\in\mathcal M\times\mathcal M}\hat f_C(s,t)\ge\frac{c_f^2}{2}\right\}.
\end{align*}
For $0<h<\iota_{\mathcal M}$ and any bounded measurable functions $g:\mathcal M\to\mathbb R$ and $G:\mathcal M\times\mathcal M\to\mathbb R$, define
\begin{align*}
    (S_hg)(s)&:=\int_{\mathcal M}\mathcal L_{s,h}(u)g(u)\,\mathrm{d}V(u),\\
    (R_hG)(s,t)&:=\int_{\mathcal M\times\mathcal M}\mathcal L_{s,h}(u)\mathcal L_{t,h}(v)G(u,v)\,\mathrm{d}V(u)\,\mathrm{d}V(v).
\end{align*}
Under Assumption~\ref{ass:sampling}, \eqref{eq:weight_mean_normalization} and \eqref{eq:weight_covariance_normalization} give
\begin{align}
\begin{split}
    \E\hat f_\mu(s)
    &=\sum_{i=1}^nw_i\sum_{j=1}^{N_i}\E\{\mathcal L_{s,h_\mu}(s_{ij})\}\\
    &=\sum_{i=1}^nw_i\sum_{j=1}^{N_i}
    \int_{\mathcal M}\mathcal L_{s,h_\mu}(u)f(u)\,\mathrm{d}V(u)\\
    &=\left(\sum_{i=1}^nN_iw_i\right)
    \int_{\mathcal M}\mathcal L_{s,h_\mu}(u)f(u)\,\mathrm{d}V(u)\\
    &=(S_{h_\mu}f)(s),\\
    \E\hat f_C(s,t)
    &=\sum_{i\in\mathcal I_C}v_i\sum_{1\le j\ne k\le N_i}
    \E\{\mathcal L_{s,h_C}(s_{ij})\mathcal L_{t,h_C}(s_{ik})\}\\
    &=\sum_{i\in\mathcal I_C}v_i\sum_{1\le j\ne k\le N_i}
    \int_{\mathcal M\times\mathcal M}
    \mathcal L_{s,h_C}(u)\mathcal L_{t,h_C}(v)
    f(u)f(v)\,\mathrm{d}V(u)\,\mathrm{d}V(v)\\
    &=\left\{\sum_{i\in\mathcal I_C}(N_i)_2v_i\right\}
    \int_{\mathcal M\times\mathcal M}
    \mathcal L_{s,h_C}(u)\mathcal L_{t,h_C}(v)
    f(u)f(v)\,\mathrm{d}V(u)\,\mathrm{d}V(v)\\
    &=R_{h_C}(f\otimes f)(s,t).
\end{split}
\label{eq:denominator_expectation_representations}
\end{align}
where $(f\otimes f)(s,t):=f(s)f(t)$.

\begin{lemma}\label{supp:denominators}
Under Assumptions~\ref{ass:sampling} and~\ref{ass:weights} and the mean-side conditions in Assumption~\ref{ass:bandwidth},
\begin{align*}
    \P(\mathcal E_\mu)\to1.
\end{align*}
Under Assumptions~\ref{ass:sampling}, \ref{ass:weights} and~\ref{ass:bandwidth},
\begin{align*}
    \P(\mathcal E_\mu\cap\mathcal E_C)\to1.
\end{align*}
\end{lemma}


\begin{proof}
By Assumption \ref{ass:bandwidth}, we may assume that $h\in\{h_\mu,h_C\}$ is sufficiently small throughout this proof.
Assumption~\ref{ass:sampling}, Lemma~\ref{lem:theta_bounds}, \eqref{kernel normalization} and the Cauchy-Schwarz inequality give
\begin{align}\label{eq:kernel_design_moments}
\begin{split}
    \E\{\mathcal L_{s,h}(s_{ij})\}
    &=\int_{\mathcal M}\mathcal L_{s,h}(u)f(u)\,\mathrm{d}V(u)
    \le C_f,\\
    \E\{\mathcal L_{s,h}(s_{ij})^2\}
    &\le(\mathrm{const.})h^{-d}
    \E\{\mathcal L_{s,h}(s_{ij})\}
    \le(\mathrm{const.})h^{-d},\\
    \E\{\mathcal L_{s,h}(s_{ij})\mathcal L_{t,h}(s_{ij})\}
    &\le\left[
    \E\{\mathcal L_{s,h}(s_{ij})^2\}
    \E\{\mathcal L_{t,h}(s_{ij})^2\}
    \right]^{1/2}\le(\mathrm{const.})h^{-d}.
\end{split}
\end{align}
By Assumption~\ref{ass:weights} and \eqref{eq:kernel_design_moments},
\begin{align}\label{eq:denominator_mean_bounds}
    \E Z_{\mu,i}(s)\le\frac{(\mathrm{const.})}{n}, \quad
    \E Z_{C,i}(s,t)\le\frac{(\mathrm{const.})}{n_C},
\end{align}
uniformly in $s\in\mathcal M$ and $(s,t)\in\mathcal M\times\mathcal M$, respectively. Since $Z_{\mu,i}(s)$ and $Z_{C,i}(s,t)$ are nonnegative, \eqref{eq:denominator_mean_bounds} gives
\begin{align}\label{eq:denominator_mean_bounds_minusz}
    \E Z_{\mu,i}(s)-Z_{\mu,i}(s)\le\frac{(\mathrm{const.})}{n}, \quad \E Z_{C,i}(s,t)-Z_{C,i}(s,t)\le\frac{(\mathrm{const.})}{n_C}.
\end{align}
Expanding the second moments and grouping terms by the number
of distinct observation indices, and using
Assumption~\ref{ass:sampling},
\eqref{eq:kernel_design_moments} and the independence of
the observation locations, we obtain
\begin{align}\label{eq:denominator_subject_second_moments}
\begin{split}
    \E\{Z_{\mu,i}(s)^2\}
    &\le(\mathrm{const.})w_i^2
    \{N_ih_\mu^{-d}+(N_i)_2\},\\
    \E\{Z_{C,i}(s,t)^2\}
    &\le(\mathrm{const.})v_i^2
    \{(N_i)_2h_C^{-2d}+(N_i)_3h_C^{-d}+(N_i)_4\}.
\end{split}
\end{align}
By \eqref{eq:denominator_subject_second_moments},
\begin{align}\label{eq:denominator_variance_bounds}
\begin{split}
    \sum_{i=1}^n\operatorname{Var}\{Z_{\mu,i}(s)\}
    &\le\sum_{i=1}^n\E\{Z_{\mu,i}(s)^2\}
    \le(\mathrm{const.})\rho_{\mu,2}^2,\\
    \sum_{i\in\mathcal I_C}\operatorname{Var}\{Z_{C,i}(s,t)\}
    &\le\sum_{i\in\mathcal I_C}\E\{Z_{C,i}(s,t)^2\}
    \le(\mathrm{const.})\rho_{C,2}^2.
\end{split}
\end{align}
Note that
\begin{align}
\begin{split}
    \rho_{\mu,2}^2
    &\ge V_{\mu,1}+V_{\mu,2}
    =\sum_{i=1}^n(N_iw_i)^2
    \ge\frac{1}{n},\\
    \rho_{C,2}^2
    &\ge\frac{1}{4}(2V_{C,2}+4V_{C,3}+V_{C,4})
    =\frac{1}{4}\sum_{i\in\mathcal I_C}\{(N_i)_2v_i\}^2
    \ge\frac{1}{4n_C}.
\end{split}
\label{eq:variance_scale_lower_bounds}
\end{align}
Here we used
\begin{align*}
    N_i^2=N_i+(N_i)_2, \quad \{(N_i)_2\}^2=2(N_i)_2+4(N_i)_3+(N_i)_4.
\end{align*}
By \eqref{eq:denominator_mean_bounds_minusz}, \eqref{eq:denominator_variance_bounds}, \eqref{eq:variance_scale_lower_bounds} and the one-sided Bernstein inequality,
\begin{align}
\begin{split}
    \P\left\{\sum_{i=1}^n\{\E Z_{\mu,i}(s)-Z_{\mu,i}(s)\}>\frac{c_f}{4}\right\}
    &\le\exp\left\{-\frac{(c_f/4)^2}{(\mathrm{const.})(\rho_{\mu,2}^2+n^{-1})}\right\}\\
    &\le\exp\left\{-\frac{(\mathrm{const.})}{\rho_{\mu,2}^2}\right\},\\
    \P\left\{\sum_{i\in\mathcal I_C}\{\E Z_{C,i}(s,t)-Z_{C,i}(s,t)\}>\frac{c_f^2}{4}\right\}
    &\le\exp\left\{-\frac{(c_f^2/4)^2}{(\mathrm{const.})(\rho_{C,2}^2+n_C^{-1})}\right\}\\
    &\le\exp\left\{-\frac{(\mathrm{const.})}{\rho_{C,2}^2}\right\}.
\end{split}
\label{eq:denominator_lower_tail}
\end{align}

We now construct a finite net. There exists a constant $A>0$ such that $h_\mu^{-1},h_C^{-1}\le n^A$ for all sufficiently large $n$. Define
\begin{align*}
    \delta_n:=n^{-\{A(2d+1)+2\}}.
\end{align*}
Assumption~\ref{ass:bandwidth} gives
\begin{align}\label{eq:denominator_superpolynomial_tail}
    \exp\left\{-\frac{(\mathrm{const.})}{\rho_{\mu,2}^2}\right\}=o\left(\delta_n^{d}\right),\quad \exp\left\{-\frac{(\mathrm{const.})}{\rho_{C,2}^2}\right\}=o\left(\delta_n^{2d}\right).
\end{align}
Let $\mathcal G_n$ be a maximal subset of $\mathcal M$ such that $B_{\mathcal M}\left(s',\delta_n/2\right)$ for $s'\in\mathcal G_n$ are pairwise disjoint. By \eqref{eq:tangent_volume_scaling}, \eqref{eq:change_of_variables} and \Cref{lem:theta_bounds}, for all sufficiently large $n$,
\begin{align}
\begin{split}
    V(B_{\mathcal M}\left(s,\delta_n/2\right))
    &=\int_{B_s(0_s,\delta_n/2)}
    \theta_s(\Exp_s(z))\,\mathrm{d}\operatorname{vol}_s(z)\\
    &\ge m_\theta
    \operatorname{vol}_s(B_s(0_s,\delta_n/2))\\
    &=m_\theta\,\omega_d\,(\delta_n/2)^d.
\end{split}
\label{eq:denominator_ball_volume}
\end{align}
By \eqref{eq:denominator_ball_volume}, for all sufficiently large $n$,
\begin{align}
\begin{split}
    |\mathcal G_n|
    &\le \frac{V(\mathcal M)}{\min_{s'\in\mathcal G_n}V\left(B_{\mathcal M}\left(s',\delta_n/2\right)\right)}\le
    \frac{2^dV(\mathcal M)}{m_\theta\,\omega_d}\delta_n^{-d}
    \le(\mathrm{const.})\delta_n^{-d},\\
    |\mathcal G_n\times\mathcal G_n|
    &=|\mathcal G_n|^2
    \le
    \left\{\frac{2^dV(\mathcal M)}{m_\theta\,\omega_d}\right\}^2
    \delta_n^{-2d}
    \le(\mathrm{const.})\delta_n^{-2d}.
\end{split}
\label{eq:denominator_net_size}
\end{align}
By Assumption~\ref{ass:sampling}, \eqref{kernel normalization} and \eqref{eq:denominator_expectation_representations},
\begin{align}
\begin{split}
    \E\hat f_\mu(s)&=(S_{h_\mu}f)(s)\ge c_f,\\
    \E\hat f_C(s,t)&=R_{h_C}(f\otimes f)(s,t)\ge c_f^2.
\end{split}
\label{eq:denominator_expectations}
\end{align}
By \eqref{eq:denominator_lower_tail}, \eqref{eq:denominator_superpolynomial_tail}, \eqref{eq:denominator_net_size} and \eqref{eq:denominator_expectations},
\begin{align}
\begin{split}
    \P\left\{\min_{s\in\mathcal G_n}\hat f_\mu(s)<\frac{3c_f}{4}\right\}
    &\le|\mathcal G_n|\exp\left\{-(\mathrm{const.})/\rho_{\mu,2}^2\right\}=o(1),\\
    \P\left\{\min_{(s,t)\in\mathcal G_n\times\mathcal G_n}\hat f_C(s,t)<\frac{3c_f^2}{4}\right\}
    &\le|\mathcal G_n\times\mathcal G_n|\exp\left\{-(\mathrm{const.})/\rho_{C,2}^2\right\}=o(1).
\end{split}
\label{eq:denominator_net_bound}
\end{align}

We now extend the bounds from the finite net to the full domain by interpolation. Note that $\mathcal G_n\times\mathcal G_n$ is a $2\delta_n$-net of $\mathcal M\times\mathcal M$ under the product metric
\begin{align}\label{eq:product_manifold_metric}
    d_{\mathcal M\times\mathcal M}\left((s,t),(s',t')\right):=d_{\mathcal M}(s,s')+d_{\mathcal M}(t,t'),\quad (s,t),(s',t')\in\mathcal M\times\mathcal M.
\end{align}
Hence, for given $(s,t)\in\mathcal M\times\mathcal M$, there exists $(s', t')\in\mathcal G_n\times\mathcal G_n$ such that 
\begin{align} \label{eq:epsilon_net_ineq}
    d_{\mathcal M}(s,s') \le \delta_n,\quad d_{\mathcal M}(t,t') \le \delta_n,\quad d_{\mathcal M\times\mathcal M}\left((s,t),(s',t')\right)\le 2\delta_n.
\end{align}
By Assumption~\ref{ass:bandwidth}, Lemma~\ref{supp:kernel}, \eqref{eq:weight_mean_normalization} and \eqref{eq:weight_covariance_normalization}, for all sufficiently large $n$,
\begin{align}
\begin{split}
    |\hat f_\mu(s)-\hat f_\mu(s')|
    &=
    \left|
    \sum_{i=1}^n w_i\sum_{j=1}^{N_i}
    \left\{
    \mathcal L_{s,h_\mu}(s_{ij})
    -\mathcal L_{s',h_\mu}(s_{ij})
    \right\}
    \right|\\
    &\le
    \sum_{i=1}^n w_i\sum_{j=1}^{N_i}
    \left|
    \mathcal L_{s,h_\mu}(s_{ij})
    -\mathcal L_{s',h_\mu}(s_{ij})
    \right|\\
    &\le
    (\mathrm{const.})h_\mu^{-d-1}d_{\mathcal M}(s,s')
    \sum_{i=1}^nN_iw_i\\
    &\le(\mathrm{const.})h_\mu^{-d-1}\delta_n,
\end{split}
\label{eq:denominator_lipschitz_mean}
\end{align}
and
\begin{align}
\begin{split}
    &|\hat f_C(s,t)-\hat f_C(s',t')|\\
    &=
    \left|
    \sum_{i\in\mathcal I_C}v_i
    \sum_{1\le j\ne k\le N_i}
    \left\{
    \mathcal L_{s,h_C}(s_{ij})\mathcal L_{t,h_C}(s_{ik})
    -\mathcal L_{s',h_C}(s_{ij})\mathcal L_{t',h_C}(s_{ik})
    \right\}
    \right|\\
    &\le
    \sum_{i\in\mathcal I_C}v_i
    \sum_{1\le j\ne k\le N_i}
    \left|
    \mathcal L_{s,h_C}(s_{ij})
    -\mathcal L_{s',h_C}(s_{ij})
    \right|
    \mathcal L_{t,h_C}(s_{ik})\\
    &\quad+
    \sum_{i\in\mathcal I_C}v_i
    \sum_{1\le j\ne k\le N_i}
    \mathcal L_{s',h_C}(s_{ij})
    \left|
    \mathcal L_{t,h_C}(s_{ik})
    -\mathcal L_{t',h_C}(s_{ik})
    \right|\\
    &\le
    (\mathrm{const.})h_C^{-2d-1}
    \left\{
    d_{\mathcal M}(s,s')+d_{\mathcal M}(t,t')
    \right\}
    \sum_{i\in\mathcal I_C}(N_i)_2v_i\\
    &\le
    (\mathrm{const.})h_C^{-2d-1}\delta_n,
\end{split}
\label{eq:denominator_lipschitz_cov}
\end{align}
whenever $(s,t)\in\mathcal M\times\mathcal M$ and $(s',t')\in\mathcal G_n\times\mathcal G_n$ satisfy \eqref{eq:epsilon_net_ineq}. Also, $h_\mu^{-1},h_C^{-1}\le n^A$ implies
\begin{align}
\begin{split}
    h_\mu^{-d-1}\delta_n
    \le(\mathrm{const.})n^{-Ad-2}
    =o(1),\quad
    h_C^{-2d-1}\delta_n
    \le(\mathrm{const.})n^{-2}
    =o(1).
\end{split}
\label{eq:denominator_interpolation}
\end{align}
By \eqref{eq:denominator_lipschitz_mean}, \eqref{eq:denominator_lipschitz_cov} and \eqref{eq:denominator_interpolation}, for all sufficiently large $n$,
\begin{align}
\begin{split}
    \inf_{s\in\mathcal M}\hat f_\mu(s)
    \ge\min_{s'\in\mathcal G_n}\hat f_\mu(s')-\frac{c_f}{4},\quad
    \inf_{(s,t)\in\mathcal M\times\mathcal M}\hat f_C(s,t)
    \ge\min_{(s',t')\in\mathcal G_n\times\mathcal G_n}\hat f_C(s',t')-\frac{c_f^2}{4}.
\end{split}
\label{eq:denominator_net_interpolation}
\end{align}
Therefore, by \eqref{eq:denominator_net_bound} and \eqref{eq:denominator_net_interpolation},
\begin{align*}
    \P(\mathcal E_\mu\cap\mathcal E_C)
    &\ge1-\P\left\{\min_{s'\in\mathcal G_n}\hat f_\mu(s')<\frac{3c_f}{4}\right\}-\P\left\{\min_{(s',t')\in\mathcal G_n\times\mathcal G_n}\hat f_C(s',t')<\frac{3c_f^2}{4}\right\}
    \to1.
\end{align*}
This proves the second assertion in \Cref{supp:denominators}.

The first assertion in \Cref{supp:denominators} requires only the mean-side conditions in Assumption~\ref{ass:bandwidth}. Indeed, the above argument applies with $A>0$ chosen so that $h_\mu^{-1}\le n^A$ for all sufficiently large $n$, and only the bounds concerning $\hat f_\mu$ are used. This gives
\begin{align*}
    \P(\mathcal E_\mu)
    \ge1-\P\left\{\min_{s'\in\mathcal G_n}\hat f_\mu(s')<\frac{3c_f}{4}\right\}
    \to1.
\end{align*}
\end{proof}

For the next lemma, define
\begin{align}\label{eq:cov_subject_contribution}
\begin{split}
    W_{\mu,i}(s)&:=w_i\sum_{j=1}^{N_i}
    \mathcal L_{s,h_\mu}(s_{ij})Y_{ij},\quad i=1,\ldots,n,\\
    W_{C,i}(s,t)&:=v_i\sum_{1\le j\ne k\le N_i}
    \mathcal L_{s,h_C}(s_{ij})\mathcal L_{t,h_C}(s_{ik})E_{ij}E_{ik},\quad i\in\mathcal I_C,    
\end{split}
\end{align}
where
\begin{align*}
    E_{ij}:=Y_{ij}-\mu(s_{ij}).
\end{align*}
Also, define
\begin{align*}
    \hat g_\mu(s)&:=\sum_{i=1}^n W_{\mu,i}(s)
    =\sum_{i=1}^n w_i\sum_{j=1}^{N_i}\mathcal L_{s,h_\mu}(s_{ij})Y_{ij},\\
    \widetilde g_C(s,t)&:=\sum_{i\in\mathcal I_C}W_{C,i}(s,t)
    =\sum_{i\in\mathcal I_C}v_i\sum_{1\le j\ne k\le N_i}
    \mathcal L_{s,h_C}(s_{ij})\mathcal L_{t,h_C}(s_{ik})E_{ij}E_{ik}.
\end{align*}
Let
\begin{align*}
    \mathcal S:=\sigma(\{s_{ij}:i=1,\ldots,n,\ j=1,\ldots,N_i\})
\end{align*}
denote the sigma-field generated by the observation locations.

\begin{lemma}\label{supp:second_moments}
Under Assumptions~\ref{ass:sampling} and \ref{ass:process} and for $0<h_\mu,h_C<\iota_{\mathcal M}/4$,
\begin{align*}
    \E\|\hat g_\mu-\E\hat g_\mu\|_2^2
    +\E\|\hat f_\mu-\E\hat f_\mu\|_2^2
    &\le(\mathrm{const.})\rho_{\mu,2}^2,\\
    \E\|\widetilde g_C-\E\widetilde g_C\|_2^2
    +\E\|\hat f_C-\E\hat f_C\|_2^2
    &\le(\mathrm{const.})\rho_{C,2}^2.    
\end{align*}
\end{lemma}

\begin{proof}
Assumptions~\ref{ass:sampling} and~\ref{ass:process}
imply that, conditional on $\mathcal S$, the second moments of $Y_{ij}$ and
of $E_{ij}E_{ik}$ for $j\ne k$  are uniformly bounded.
By the Cauchy-Schwarz inequality,
\begin{align*}
    \E\{W_{\mu,i}(s)^2\mid\mathcal S\}
    &\le(\mathrm{const.})Z_{\mu,i}(s)^2,\\
    \E\{W_{C,i}(s,t)^2\mid\mathcal S\}
    &\le(\mathrm{const.})Z_{C,i}(s,t)^2.
\end{align*}
By \eqref{eq:denominator_subject_second_moments}
and the tower property,
\begin{align}
    \E\{W_{\mu,i}(s)^2\}+\E\{Z_{\mu,i}(s)^2\}
    \le(\mathrm{const.})w_i^2
    \{N_ih_\mu^{-d}+(N_i)_2\}.
    \label{eq:mean_subject_second_moment_bound}
\end{align}
Similarly,
\begin{align}
\begin{split}
    \E\{W_{C,i}(s,t)^2\}+\E\{Z_{C,i}(s,t)^2\}\le(\mathrm{const.})v_i^2
    \{(N_i)_2h_C^{-2d}+(N_i)_3h_C^{-d}+(N_i)_4\}.
\end{split}
\label{eq:cov_subject_second_moment_bound}
\end{align}
Summing \eqref{eq:mean_subject_second_moment_bound} and
\eqref{eq:cov_subject_second_moment_bound} over subjects
and integrating, and using independence across subjects
and Tonelli's theorem, we get the desired results.
\end{proof}

For the next lemma, define
\begin{align*}
    \Delta(s):=\hat\mu(s)-\mu(s),\quad\Delta_{ij}:=\Delta(s_{ij}).
\end{align*}

\begin{lemma}\label{supp:insample}
Under the assumptions of Theorem~\ref{thm:cov_l2_rate},
\begin{align*}
    \mathbf 1_{\mathcal E_\mu}\sum_{i\in\mathcal I_C}v_i
    \sum_{1\le j\ne k\le N_i}
    \left(E_{ij}^2\Delta_{ik}^2+E_{ik}^2\Delta_{ij}^2+\Delta_{ij}^2\Delta_{ik}^2\right)
    =O_p(r_{\mu C,2}^2).
\end{align*}
\end{lemma}

\begin{proof}
By Lemma~\ref{supp:denominators}, we may simply assume that the event $\mathcal E_\mu$ holds throughout this proof.
Define
\begin{align*}
    B_i(s)&:=w_i\sum_{j=1}^{N_i}\mathcal L_{s,h_\mu}(s_{ij})
    \{\mu(s_{ij})-\mu(s)\}, \\
    D_i(s)&:=w_i\sum_{j=1}^{N_i}\mathcal L_{s,h_\mu}(s_{ij})E_{ij}.
\end{align*}
Then,
\begin{align}
    \Delta(s)
    =\frac{\sum_{i=1}^nB_i(s)+\sum_{i=1}^nD_i(s)}
    {\hat f_\mu(s)}.
    \label{eq:insample_mean_decomposition}
\end{align}
By Assumptions~\ref{ass:sampling} and~\ref{ass:process} and Jensen's inequality,
\begin{align}
\begin{split}
    \E\left[\left.\left\{\frac{D_i(s)}{\hat f_\mu(s)}\right\}^4\right|\mathcal S\right]
    &=\E\left[
    \left\{
    \frac{w_i}{\hat f_\mu(s)}
    \sum_{j=1}^{N_i}\mathcal L_{s,h_\mu}(s_{ij})E_{ij}
    \right\}^4
    \,\middle|\,\mathcal S
    \right]\\
    &\le\left\{\frac{Z_{\mu,i}(s)}{\hat f_\mu(s)}\right\}^3
    \sum_{j=1}^{N_i}\frac{w_i\mathcal L_{s,h_\mu}(s_{ij})}{\hat f_\mu(s)}
    \E(E_{ij}^4\mid\mathcal S)\\
    &\le(\mathrm{const.})\left\{\frac{Z_{\mu,i}(s)}{\hat f_\mu(s)}\right\}^4.
\end{split}
\label{eq:insample_subject_fourth}
\end{align}
Conditional on $\mathcal S$, $D_1(s),\ldots,D_n(s)$ are independent and centered. By \eqref{eq:insample_mean_decomposition}, \eqref{eq:insample_subject_fourth} and the Cauchy-Schwarz inequality,
\begin{align}
\begin{split}
    &\E\{\Delta(s)^4\mid\mathcal S\}\\
    &=\E\left[\left.\left\{\frac{\sum_{i=1}^nB_i(s)+\sum_{i=1}^nD_i(s)}
    {\hat f_\mu(s)}\right\}^4\right|\mathcal S\right]\\
    &\le8\left\{\frac{\sum_{i=1}^nB_i(s)}{\hat f_\mu(s)}\right\}^4
    +8\E\left[\left.\left\{\frac{\sum_{i=1}^nD_i(s)}{\hat f_\mu(s)}\right\}^4\right|\mathcal S\right]\\
    &=8\left\{\frac{\sum_{i=1}^nB_i(s)}{\hat f_\mu(s)}\right\}^4
    +8\sum_{i=1}^n\E\left[\left.\left\{\frac{D_i(s)}{\hat f_\mu(s)}\right\}^4\right|\mathcal S\right]\\
    &\quad+48\sum_{1\le i<\ell\le n}
    \E\left[\left.\left\{\frac{D_i(s)}{\hat f_\mu(s)}\right\}^2\right|\mathcal S\right]
    \E\left[\left.\left\{\frac{D_\ell(s)}{\hat f_\mu(s)}\right\}^2\right|\mathcal S\right]\\
    &\le(\mathrm{const.})\left[
    \left\{\frac{\sum_{i=1}^nB_i(s)}{\hat f_\mu(s)}\right\}^4
    +\sum_{i=1}^n\left\{\frac{Z_{\mu,i}(s)}{\hat f_\mu(s)}\right\}^4
    +\sum_{1\le i<\ell\le n}\left\{\frac{Z_{\mu,i}(s)Z_{\mu,\ell}(s)}{\hat f_\mu(s)^2}\right\}^2
    \right]\\
    &\le(\mathrm{const.})Q(s)^2,
\end{split}
\label{eq:insample_conditional_fourth}
\end{align}
where
\begin{align}\label{eq:insample_Q_definition}
    Q(s):=\frac{\left\{\sum_{i=1}^nB_i(s)\right\}^2+\sum_{i=1}^nZ_{\mu,i}(s)^2}
        {\hat f_\mu(s)^2}.
\end{align}
Note that
\begin{align} \label{eq:insample_Z_bound}
    Z_{\mu,i}(s)\ge0,\quad
    |B_i(s)|\le2\|\mu\|_\infty Z_{\mu,i}(s).
\end{align}
Hence,
\begin{align}\label{eq:insample_Q_bound}
    Q(s)
    \le
    \frac{\left\{\sum_{i=1}^nB_i(s)\right\}^2+\sum_{i=1}^nZ_{\mu,i}(s)^2}
    {\left\{\sum_{i=1}^nZ_{\mu,i}(s)\right\}^2}
    \le4\|\mu\|_\infty^2+1.
\end{align}
By Assumptions~\ref{ass:sampling} and~\ref{ass:process}, the conditional fourth moments of $E_{ij}$ given $\mathcal S$ are uniformly bounded. By \eqref{eq:insample_conditional_fourth}, \eqref{eq:insample_Q_bound} and the Cauchy-Schwarz inequality,
\begin{align}
\begin{split}
    \E(E_{ij}^2\Delta_{ik}^2\mid\mathcal S)
    &\le(\mathrm{const.})Q(s_{ik}),\\
    \E(\Delta_{ij}^2\Delta_{ik}^2\mid\mathcal S)
    &\le(\mathrm{const.})Q(s_{ij})Q(s_{ik})\\
    &\le(\mathrm{const.})Q(s_{ik}), \quad i\in\mathcal I_C,\quad 1\le j\ne k\le N_i.
\end{split}
\label{eq:insample_conditional_products}
\end{align}

We now bound the weighted average of $Q$ evaluated at the observation locations. By \eqref{eq:insample_Q_definition} and \eqref{eq:insample_Z_bound},
\begin{align*}
    Q(s)
    &\le\frac{(\mathrm{const.})}{\hat f_\mu(s)^2}
    \left\{
    \left(\sum_{\ell\ne i}B_\ell(s)\right)^2
    +B_i(s)^2
    +\sum_{\ell\ne i}Z_{\mu,\ell}(s)^2
    +Z_{\mu,i}(s)^2
    \right\}\\
    &\le\frac{4(\mathrm{const.})}{c_f^2}
    \left\{
    \left(\sum_{\ell\ne i}B_\ell(s)\right)^2
    + 4\|\mu\|_\infty^2 Z_{\mu,i}(s)^2
    +\sum_{\ell\ne i}Z_{\mu,\ell}(s)^2
    +Z_{\mu,i}(s)^2
    \right\}\\
    &\le(\mathrm{const.})\left\{
    \left(\sum_{\ell\ne i}B_\ell(s)\right)^2
    +\sum_{\ell\ne i}Z_{\mu,\ell}(s)^2
    +Z_{\mu,i}(s)^2
    \right\}.
\end{align*}
This implies
\begin{align}\label{eq:insample_Q_split}
    Q(s_{ik})\le(\mathrm{const.})\left\{
    \left(\sum_{\ell\ne i}B_\ell(s_{ik})\right)^2
    +\sum_{\ell\ne i}Z_{\mu,\ell}(s_{ik})^2
    +Z_{\mu,i}(s_{ik})^2
    \right\}.
\end{align}
We first get an upper bound for the expectation of $(\sum_{\ell\ne i}B_\ell(s_{ik}))^2+\sum_{\ell\ne i}Z_{\mu,\ell}(s_{ik})^2$ in \eqref{eq:insample_Q_split}. Assumptions~\ref{ass:sampling} and~\ref{ass:process}, Lemma~\ref{lem:intrinsic_bias} and \eqref{eq:weight_mean_normalization} give
\begin{align}\label{eq:insample_B_mean}
    \left|\E\sum_{\ell\ne i}B_\ell(s)\right|
    \le(\mathrm{const.})h_\mu^2.
\end{align}
Also, by \eqref{eq:kernel_design_moments} and the independence of the observation locations,
\begin{align}\label{eq:insample_B_variance}
    \operatorname{Var}\left(\sum_{\ell\ne i}B_\ell(s)\right)
    =\sum_{\ell\ne i}\operatorname{Var}\{B_\ell(s)\}
    \le(\mathrm{const.})\frac{V_{\mu,1}}{h_\mu^d}.
\end{align}
Hence, by \eqref{eq:insample_B_mean} and \eqref{eq:insample_B_variance},
\begin{align}\label{eq:insample_B_second_moment}
    \E\left(\sum_{\ell\ne i}B_\ell(s)\right)^2
    =\operatorname{Var}\left(\sum_{\ell\ne i}B_\ell(s)\right)
    +\left\{\E\sum_{\ell\ne i}B_\ell(s)\right\}^2\le(\mathrm{const.})\left(\frac{V_{\mu,1}}{h_\mu^d}+h_\mu^4\right).
\end{align}
Moreover, by \eqref{eq:kernel_design_moments} and the independence of the observation locations,
\begin{align}\label{eq:insample_Z_second_moment}
    \sum_{\ell\ne i}\E Z_{\mu,\ell}(s)^2
    \le(\mathrm{const.})\sum_{\ell\ne i}w_\ell^2
    \left\{N_\ell h_\mu^{-d}+(N_\ell)_2\right\}\le(\mathrm{const.})\left(\frac{V_{\mu,1}}{h_\mu^d}+V_{\mu,2}\right).
\end{align}
Therefore, by \eqref{eq:insample_B_second_moment} and \eqref{eq:insample_Z_second_moment},
\begin{align}\label{eq:insample_other_subject_bounds}
\begin{split}
    \E\left(\sum_{\ell\ne i}B_\ell(s)\right)^2
    +\sum_{\ell\ne i}\E Z_{\mu,\ell}(s)^2
    &\le(\mathrm{const.})\left(h_\mu^4+\frac{V_{\mu,1}}{h_\mu^d}+V_{\mu,2}\right)\\
    &\le(\mathrm{const.})(h_\mu^4+\rho_{\mu,2}^2).
\end{split}
\end{align} 
Since the bound in \eqref{eq:insample_other_subject_bounds} is independent of $s\in\mathcal M$, and $s_{ik}$ is independent of the observation locations from subjects $\ell\ne i$, evaluating the bound at $s=s_{ik}$ gives
\begin{align}\label{eq:insample_other_subject_random_location}
    \E\left\{
    \left(\sum_{\ell\ne i}B_\ell(s_{ik})\right)^2
    +\sum_{\ell\ne i}Z_{\mu,\ell}(s_{ik})^2
    \right\}
    \le(\mathrm{const.})\left(h_\mu^4+\rho_{\mu,2}^2\right).
\end{align}
We now obtain an upper bound for the expectation of $Z_{\mu,i}(s_{ik})^2$ in \eqref{eq:insample_Q_split}. Note that
\begin{align*}
    Z_{\mu,i}(s_{ik}) = w_i\left\{\mathcal L_{s_{ik},h_\mu}(s_{ik})+\sum_{j\ne k}\mathcal L_{s_{ik},h_\mu}(s_{ij})\right\}
    =w_i\left\{h_\mu^{-d}K(0)+\sum_{j\ne k}\mathcal L_{s_{ik},h_\mu}(s_{ij})\right\}.
\end{align*}
Conditional on $s_{ik}$, the locations $\{s_{ij}:1\le j\le N_i,\ j\ne k\}$ remain mutually independent. Hence, by \eqref{eq:kernel_design_moments},
\begin{align*}
    &\E\{Z_{\mu,i}(s_{ik})^2\mid s_{ik}\}\\
    &=w_i^2\left[
    h_\mu^{-2d}K(0)^2
    +2h_\mu^{-d}K(0)\sum_{j\ne k}\E\{\mathcal L_{s_{ik},h_\mu}(s_{ij})\mid s_{ik}\}\right.\\
    &\quad\left.
    +\sum_{j\ne k}\E\{\mathcal L_{s_{ik},h_\mu}(s_{ij})^2\mid s_{ik}\}
    +\sum_{\substack{j\ne j'\\j,j'\ne k}}
    \E\{\mathcal L_{s_{ik},h_\mu}(s_{ij})\mathcal L_{s_{ik},h_\mu}(s_{ij'})\mid s_{ik}\}
    \right]\\
    &\le(\mathrm{const.})w_i^2
    \left\{
    h_\mu^{-2d}
    +(N_i-1)h_\mu^{-d}
    +(N_i-1)^2
    \right\}.
\end{align*}
Taking the expectation gives
\begin{align}\label{eq:insample_self_subject_unconditional}
    \E Z_{\mu,i}(s_{ik})^2
    \le(\mathrm{const.})w_i^2
    \left\{
    h_\mu^{-2d}
    +(N_i-1)h_\mu^{-d}
    +(N_i-1)^2
    \right\}.
\end{align}
Since $(N_i-1)^2\le(N_i)_2$ for $i\in\mathcal I_C$,
\begin{align*}
    w_i^2\left\{(N_i-1)h_\mu^{-d}+(N_i-1)^2\right\}
    \le\frac{V_{\mu,1}}{h_\mu^d}+V_{\mu,2}
    =\rho_{\mu,2}^2.
\end{align*}
Thus, by \eqref{eq:insample_Q_split},
\eqref{eq:insample_other_subject_random_location} and
\eqref{eq:insample_self_subject_unconditional},
\begin{align*}
    \E\{\mathbf 1_{\mathcal E_\mu}Q(s_{ik})\}
    \le(\mathrm{const.})\left\{
    h_\mu^4+\rho_{\mu,2}^2+w_i^2h_\mu^{-2d}
    \right\}.
\end{align*}
By \eqref{eq:weight_covariance_normalization} and
the definition of $\Gamma_{\mu C}$,
\begin{align}\label{eq:insample_design_average}
\begin{split}
    &\E\left[
    \mathbf 1_{\mathcal E_\mu}
    \sum_{i\in\mathcal I_C}(N_i-1)v_i
    \sum_{k=1}^{N_i}Q(s_{ik})
    \right]\\
    &\le(\mathrm{const.})\left\{
    h_\mu^4+\rho_{\mu,2}^2
    +\frac{\sum_{i\in\mathcal I_C}(N_i)_2v_iw_i^2}{h_\mu^{2d}}
    \right\}\\
    &\le(\mathrm{const.})r_{\mu C,2}^2.
\end{split}
\end{align}
By \eqref{eq:insample_conditional_products}, \eqref{eq:insample_design_average} and the tower property,
\begin{align}
\begin{split}
    &\E\left[
    \mathbf 1_{\mathcal E_\mu}
    \sum_{i\in\mathcal I_C}v_i\sum_{1\le j\ne k\le N_i}
    \left\{E_{ij}^2\Delta_{ik}^2+E_{ik}^2\Delta_{ij}^2+\Delta_{ij}^2\Delta_{ik}^2\right\}
    \right]\\
    &=\E\left[\E\left\{\left.
    \mathbf 1_{\mathcal E_\mu}
    \sum_{i\in\mathcal I_C}v_i\sum_{1\le j\ne k\le N_i}
    \left(E_{ij}^2\Delta_{ik}^2+E_{ik}^2\Delta_{ij}^2+\Delta_{ij}^2\Delta_{ik}^2\right)
    \right|\mathcal S\right\}\right]\\
    &\le(\mathrm{const.})\E\left[
    \mathbf 1_{\mathcal E_\mu}
    \sum_{i\in\mathcal I_C}v_i\sum_{1\le j\ne k\le N_i}
    \{Q(s_{ik})+Q(s_{ij})\}
    \right]\\
    &\le(\mathrm{const.})\E\left[
    \mathbf 1_{\mathcal E_\mu}
    \sum_{i\in\mathcal I_C}(N_i-1)v_i
    \sum_{k=1}^{N_i}Q(s_{ik})
    \right]\\
    &\le(\mathrm{const.})r_{\mu C,2}^2,
\end{split}
\label{eq:insample_lemma_expectation}
\end{align}
where we used
\begin{align*}
    \sum_{1\le j\ne k\le N_i}\{Q(s_{ik})+Q(s_{ij})\}
    =2(N_i-1)\sum_{k=1}^{N_i}Q(s_{ik}).
\end{align*} 
Since the random variable inside the expectation on the left-hand side
of \eqref{eq:insample_lemma_expectation} is nonnegative,
Markov's inequality proves the lemma.
\end{proof}

\subsection{Proof of Theorem~\ref{thm:mean_l2_rate}}

By Lemma~\ref{supp:denominators}, we may simply assume that the event $\mathcal E_\mu$ holds throughout this proof.
By Assumptions~\ref{ass:sampling} and~\ref{ass:process}, Lemma~\ref{lem:intrinsic_bias} and \eqref{eq:weight_mean_normalization},
\begin{align}\label{eq:mean_numerator_bias}
    \|\E\hat g_\mu-\mu\E\hat f_\mu\|_\infty
    =\|S_{h_\mu}(\mu f)-\mu S_{h_\mu}f\|_\infty
    =O(h_\mu^2).
\end{align}
By \Cref{supp:second_moments} and Markov's inequality,
\begin{align}\label{eq:mean_stochastic_l2}
    \|\hat g_\mu-\E\hat g_\mu\|_2+\|\hat f_\mu-\E\hat f_\mu\|_2
    =O_p(\rho_{\mu,2}).
\end{align}
Note that
\begin{align}\label{eq:mean_error_decomposition}
    \hat g_\mu-\mu\hat f_\mu
    =(\hat g_\mu-\E\hat g_\mu)-\mu(\hat f_\mu-\E\hat f_\mu)
    +(\E\hat g_\mu-\mu\E\hat f_\mu).
\end{align}
By \eqref{eq:mean_numerator_bias}, \eqref{eq:mean_stochastic_l2} and \eqref{eq:mean_error_decomposition},
\begin{align*}
    \|\hat\mu-\mu\|_2
    \le\frac{2}{c_f}\|\hat g_\mu-\mu\hat f_\mu\|_2
    =O_p(h_\mu^2+\rho_{\mu,2}).
\end{align*}

\subsection{Proof of Theorem~\ref{thm:cov_l2_rate}}

By Lemma~\ref{supp:denominators}, we may simply assume that the event $\mathcal E_\mu\cap\mathcal E_C$ holds throughout this proof.
Define
\begin{align*}
    \widetilde C(s,t):=\dfrac{\widetilde g_C(s,t)}{\hat f_C(s,t)}.
\end{align*}
We first obtain an upper bound for $\|\widetilde C-C\|_2$. By Assumptions~\ref{ass:sampling} and~\ref{ass:process},
\begin{align} \label{eq:conditional_covariance}
    \E(E_{ij}E_{ik}\mid\mathcal S)
    =C(s_{ij},s_{ik}),
    \quad i\in\mathcal I_C,\quad 1\le j\ne k\le N_i.
\end{align}
By Assumption~\ref{ass:sampling}, \eqref{eq:weight_covariance_normalization}, \eqref{eq:conditional_covariance} and the tower property, and arguing as in \eqref{eq:denominator_expectation_representations},
\begin{align}\label{eq:oracle_covariance_expectations}
\begin{split}
    \E\widetilde g_C(s,t)
    &=\E\left[\E\{\widetilde g_C(s,t)\mid\mathcal S\}\right]\\
    &=\E\left[
    \sum_{i\in\mathcal I_C}v_i\sum_{1\le j\ne k\le N_i}
    \mathcal L_{s,h_C}(s_{ij})\mathcal L_{t,h_C}(s_{ik})
    \E(E_{ij}E_{ik}\mid\mathcal S)
    \right]\\
    &=\E\left[
    \sum_{i\in\mathcal I_C}v_i\sum_{1\le j\ne k\le N_i}
    \mathcal L_{s,h_C}(s_{ij})\mathcal L_{t,h_C}(s_{ik})
    C(s_{ij},s_{ik})
    \right]\\
    &=R_{h_C}\{C(f\otimes f)\}(s,t).
\end{split}
\end{align}
By Assumptions~\ref{ass:sampling} and~\ref{ass:process}, Lemma~\ref{lem:intrinsic_bias}, \eqref{eq:denominator_expectation_representations} and \eqref{eq:oracle_covariance_expectations},
\begin{align}\label{eq:oracle_covariance_bias}
    \|\E\widetilde g_C-C\E\hat f_C\|_\infty
    =\|R_{h_C}\{C(f\otimes f)\}-CR_{h_C}(f\otimes f)\|_\infty
    =O(h_C^2).
\end{align}
By Lemma~\ref{supp:second_moments} and Markov's inequality,
\begin{align}\label{eq:oracle_covariance_stochastic_l2}
    \|\widetilde g_C-\E\widetilde g_C\|_2+\|\hat f_C-\E\hat f_C\|_2
    =O_p(\rho_{C,2}).
\end{align}
Note that
\begin{align}\label{eq:oracle_covariance_decomposition}
    \widetilde g_C-C\hat f_C
    =(\widetilde g_C-\E\widetilde g_C)-C(\hat f_C-\E\hat f_C)
    +(\E\widetilde g_C-C\E\hat f_C).
\end{align}
By \eqref{eq:oracle_covariance_bias}, \eqref{eq:oracle_covariance_stochastic_l2} and \eqref{eq:oracle_covariance_decomposition},
\begin{align}\label{eq:oracle_covariance_l2}
    \|\widetilde C-C\|_2
    \le\frac{2}{c_f^2}\|\widetilde g_C-C\hat f_C\|_2
    =O_p(h_C^2+\rho_{C,2}).
\end{align}

We now obtain an upper bound for $\|\hat C-\widetilde C\|_2$. By the definitions of $E_{ij}$, $\Delta_{ij}$ and $R_{ij}$,
\begin{align*}
    E_{ij}-\Delta_{ij}=Y_{ij}-\mu(s_{ij})-\{\hat\mu(s_{ij})-\mu(s_{ij})\}=R_{ij},\quad i\in\mathcal I_C.
\end{align*}
Hence,
\begin{align}\label{eq:residual_product_difference}
    R_{ij}R_{ik}-E_{ij}E_{ik}=-E_{ij}\Delta_{ik}-E_{ik}\Delta_{ij}+\Delta_{ij}\Delta_{ik},\quad i\in\mathcal I_C,\quad 1\le j\ne k\le N_i.
\end{align}
By \eqref{eq:residual_product_difference} and the weighted Cauchy-Schwarz inequality,
\begin{align}
\begin{split}
    &|\hat C(s,t)-\widetilde C(s,t)|^2\\
    &=\frac{1}{\hat f_C(s,t)^2}
    \left|
    \sum_{i\in\mathcal I_C}v_i\sum_{1\le j\ne k\le N_i}
    \mathcal L_{s,h_C}(s_{ij})\mathcal L_{t,h_C}(s_{ik})
    \{R_{ij}R_{ik}-E_{ij}E_{ik}\}
    \right|^2\\
    &\le\frac{1}{\hat f_C(s,t)^2}
    \left\{
    \sum_{i\in\mathcal I_C}v_i\sum_{1\le j\ne k\le N_i}
    \mathcal L_{s,h_C}(s_{ij})\mathcal L_{t,h_C}(s_{ik})
    \right\}\\
    &\quad\times
    \left[
    \sum_{i\in\mathcal I_C}v_i\sum_{1\le j\ne k\le N_i}
    \mathcal L_{s,h_C}(s_{ij})\mathcal L_{t,h_C}(s_{ik})
    \{R_{ij}R_{ik}-E_{ij}E_{ik}\}^2
    \right]\\
    &=\frac{1}{\hat f_C(s,t)}
    \sum_{i\in\mathcal I_C}v_i\sum_{1\le j\ne k\le N_i}
    \mathcal L_{s,h_C}(s_{ij})\mathcal L_{t,h_C}(s_{ik})
    \{-E_{ij}\Delta_{ik}-E_{ik}\Delta_{ij}+\Delta_{ij}\Delta_{ik}\}^2\\
    &\le\frac{3}{\hat f_C(s,t)}
    \sum_{i\in\mathcal I_C}v_i\sum_{1\le j\ne k\le N_i}
    \mathcal L_{s,h_C}(s_{ij})\mathcal L_{t,h_C}(s_{ik})
    \{E_{ij}^2\Delta_{ik}^2+E_{ik}^2\Delta_{ij}^2+\Delta_{ij}^2\Delta_{ik}^2\}.
\end{split}
\label{eq:covariance_residualization_pointwise}
\end{align}
The change of variables $s=\Exp_u(z)$, \eqref{eq:tangent_volume_scaling}, \eqref{eq:change_of_variables}, \eqref{kernel normalization} and \Cref{lem:theta_bounds} give
\begin{align}\label{eq:kernel_center_integral}
\begin{split}
    \int_{\mathcal M}\mathcal L_{s,h}(u)\,\mathrm{d}V(s)
    &=\int_{B_u(0_u,h)}
    \frac{\theta_u(\Exp_u(z))}{h^d\,\theta_{\Exp_u(z)}(u)}
    K\left(\frac{\|z\|_u}{h}\right)\mathrm{d}\operatorname{vol}_u(z)\\
    &\le\frac{M_\theta}{m_\theta}
    \int_{B_u(0_u,1)}K(\|z\|_u)\,\mathrm{d}\operatorname{vol}_u(z)\\
    &=\frac{M_\theta}{m_\theta}.
\end{split}
\end{align}
By \eqref{eq:covariance_residualization_pointwise} and \eqref{eq:kernel_center_integral},
\begin{align}
\begin{split}
    \|\hat C-\widetilde C\|_2^2
    &\le(\mathrm{const.})
    \sum_{i\in\mathcal I_C}v_i\sum_{1\le j\ne k\le N_i}
    \{E_{ij}^2\Delta_{ik}^2+E_{ik}^2\Delta_{ij}^2+\Delta_{ij}^2\Delta_{ik}^2\}\\
    &\quad\times\int_{\mathcal M\times\mathcal M}
    \mathcal L_{s,h_C}(s_{ij})\mathcal L_{t,h_C}(s_{ik})
    \,\mathrm{d}V(s)\,\mathrm{d}V(t)\\
    &\le(\mathrm{const.})
    \sum_{i\in\mathcal I_C}v_i\sum_{1\le j\ne k\le N_i}
    \{E_{ij}^2\Delta_{ik}^2+E_{ik}^2\Delta_{ij}^2+\Delta_{ij}^2\Delta_{ik}^2\}.
\end{split}
\label{eq:covariance_residualization_integrated}
\end{align}
Hence, by Lemmas~\ref{supp:denominators} and \ref{supp:insample} and \eqref{eq:covariance_residualization_integrated},
\begin{align}\label{eq:covariance_residualization_l2}
    \|\hat C-\widetilde C\|_2=O_p(r_{\mu C,2}).
\end{align}
Therefore, by \eqref{eq:oracle_covariance_l2}, \eqref{eq:covariance_residualization_l2} and the triangle inequality,
\begin{align*}
    \|\hat C-C\|_2
    \le\|\hat C-\widetilde C\|_2+\|\widetilde C-C\|_2
    =O_p(r_{\mu C,2}+h_C^2+\rho_{C,2})
    =O_p(r_{C,2}).
\end{align*}

\subsection{Proof of Theorem~\ref{thm:eigen_rates}}

By Lemma~\ref{supp:denominators}, we may simply assume that the event $\mathcal E_\mu\cap\mathcal E_C$ holds throughout this proof. Let $\|\cdot\|_{\mathrm{op}}$ and $\|\cdot\|_{\mathrm{HS}}$ denote the operator norm and the Hilbert-Schmidt norm for operators, respectively.
The proof of Theorem~\ref{thm:cov_l2_rate} and
$C\in L^2(\mathcal M\times\mathcal M)$ imply
$\hat C\in L^2(\mathcal M\times\mathcal M)$.
Hence, $\hat{\mathcal C}-\mathcal C$ is Hilbert-Schmidt,
and Theorem~\ref{thm:cov_l2_rate} gives
\begin{align}\label{eq:covariance_operator_error}
    \|\hat{\mathcal C}-\mathcal C\|_{\mathrm{op}}
    \le\|\hat{\mathcal C}-\mathcal C\|_{\mathrm{HS}}
    =\|\hat C-C\|_2
    =O_p(r_{C,2}).
\end{align}
By Theorem~4.2.8 in \citet{HsingEubank2015} and the discussion following it,
\begin{align}\label{eq:eigenvalue_perturbation_bound}
    |\hat\lambda_k-\lambda_k|
    \le\|\hat{\mathcal C}-\mathcal C\|_{\mathrm{op}},
    \quad k\ge1.
\end{align}
By \eqref{eq:covariance_operator_error} and \eqref{eq:eigenvalue_perturbation_bound},
\begin{align*}
    |\hat\lambda_k-\lambda_k|=O_p(r_{C,2}),\quad k\ge1.
\end{align*}

Suppose that Assumption~\ref{ass:eigengap} holds. Define
\begin{align*}
    \lambda_0:=\infty,\quad
    \alpha_k:=\min\{\lambda_{k-1}-\lambda_k,\lambda_k-\lambda_{k+1}\}>0,
    \quad k=1,\ldots,K^*.
\end{align*}
By Lemma 2.3 in \citet{HorvathKokoszka2012} and using the sign convention in Section~2.4,
\begin{align}\label{eq:spectral_gap_eigenfunction_bound}
    \|\hat\phi_k-\phi_k\|_2
    \le(\mathrm{const.})\alpha_k^{-1}
    \|\hat{\mathcal C}-\mathcal C\|_{\mathrm{op}},
    \quad k=1,\ldots,K^*.
\end{align}
This gives
\begin{align*}
    \max_{1\le k\le K^*}\|\hat\phi_k-\phi_k\|_2
    =O_p(r_{C,2}).
\end{align*}

\subsection{Proofs of Corollaries~\ref{cor:obs_rates} and~\ref{cor:subj_rates}}

For either scheme $W\in\{O,S\}$, direct substitution gives
\begin{align}\label{eq:scheme_variance_exact}
    V_{\mu,1}=\frac{1}{na_{\mu,W}},\quad
    V_{C,2}=\frac{1}{n_Ca_{C,W}}.
\end{align}
By Assumption~\ref{ass:weights}, \eqref{eq:weight_mean_normalization} and \eqref{eq:weight_covariance_normalization},
\begin{align} \label{eq:cor_variance_bound_1}
\begin{split}
    V_{\mu,2}
    &\le\sum_{i=1}^n(N_iw_i)^2
    \le\left\{\max_{1\le i\le n}(N_iw_i)\right\}\sum_{i=1}^nN_iw_i
    \le\frac{(\mathrm{const.})}{n},\\
    V_{C,4}
    &\le\sum_{i\in\mathcal I_C}\{(N_i)_2v_i\}^2
    \le\left[\max_{i\in\mathcal I_C}\{(N_i)_2v_i\}\right]
    \sum_{i\in\mathcal I_C}(N_i)_2v_i
    \le\frac{(\mathrm{const.})}{n_C}.
\end{split} 
\end{align}
Since $(N_i)_3\le(N_i)_2^{3/2}$ for $i\in\mathcal I_C$, \eqref{eq:scheme_variance_exact}, \eqref{eq:cor_variance_bound_1} and the Cauchy-Schwarz inequality give
\begin{align} \label{eq:cor_variance_bound_2}
    V_{C,3}\le\sum_{i\in\mathcal I_C}(N_i)_2^{3/2}v_i^2\le\left[V_{C,2}\sum_{i\in\mathcal I_C}\{(N_i)_2v_i\}^2\right]^{1/2}\le\frac{(\mathrm{const.})}{n_C\sqrt{a_{C,W}}}.
\end{align}
Hence, \eqref{eq:cor_variance_bound_1} and \eqref{eq:cor_variance_bound_2} give
\begin{align*}
    \rho_{\mu,2}^2
    \le(\mathrm{const.})\left\{\frac{1}{na_{\mu,W}h_\mu^d}+\frac{1}{n}\right\},\quad
    \rho_{C,2}^2
    \le(\mathrm{const.})\left\{
    \frac{1}{n_Ca_{C,W}h_C^{2d}}
    +\frac{1}{n_C\sqrt{a_{C,W}}h_C^d}
    +\frac{1}{n_C}\right\}.
\end{align*}
Since
\begin{align*}
    \frac{1}{n_C\sqrt{a_{C,W}}h_C^d}
    =\left\{\frac{1}{n_Ca_{C,W}h_C^{2d}}\cdot\frac{1}{n_C}\right\}^{1/2}
    \le\frac{1}{2}\left\{\frac{1}{n_Ca_{C,W}h_C^{2d}}+\frac{1}{n_C}\right\},
\end{align*}
we obtain
\begin{align}\label{supp:scheme_scales}
    \rho_{\mu,2}^2
    \le(\mathrm{const.})\left\{\frac{1}{na_{\mu,W}h_\mu^d}+\frac{1}{n}\right\},\quad
    \rho_{C,2}^2
    \le(\mathrm{const.})\left\{\frac{1}{n_Ca_{C,W}h_C^{2d}}+\frac{1}{n_C}\right\}.
\end{align}
Take $h_\mu=h_{\mu,W}$ and $h_C=h_{C,W}$ as in \eqref{eq:scheme_bandwidths}. Then,
\begin{align} \label{supp:scheme_bandwidth}
\begin{split}
    \frac{1}{na_{\mu,W}h_{\mu,W}^d}
    &\le(\mathrm{const.})h_{\mu,W}^4,\quad
    \frac{1}{n}\le(\mathrm{const.})h_{\mu,W}^4,\\
    \frac{1}{n_Ca_{C,W}h_{C,W}^{2d}}
    &\le(\mathrm{const.})h_{C,W}^4,\quad
    \frac{1}{n_C}\le(\mathrm{const.})h_{C,W}^4.
\end{split}
\end{align}
By \eqref{supp:scheme_scales} and \eqref{supp:scheme_bandwidth},
\begin{align}\label{eq:scheme_rho_rates}
    \rho_{\mu,2}=O(h_{\mu,W}^2),\quad
    \rho_{C,2}=O(h_{C,W}^2).
\end{align}
We verify Assumption~\ref{ass:bandwidth}. Since $a_{\mu,W}\ge1$, $a_{C,W}\ge2$ and $n_C\to\infty$ as $n\to\infty$, \eqref{eq:scheme_bandwidths} gives
\begin{align}\label{supp:bandwidth_condition_1}
    h_{\mu,W}\to0,\quad h_{C,W}\to0.
\end{align}
Moreover, by \eqref{eq:scheme_bandwidths},
\begin{align*}
    h_{\mu,W}\ge(\mathrm{const.})n^{-1/4},\quad
    h_{C,W}\ge(\mathrm{const.})n_C^{-1/4}.
\end{align*}
Since $n_C\le n$,
\begin{align}\label{supp:bandwidth_condition_2}
    \log(1/h_{\mu,W})=O(\log n),\quad
    \log(1/h_{C,W})=O(\log n).
\end{align}
By \eqref{eq:scheme_bandwidths} and \eqref{supp:scheme_scales} and using $a_{\mu,W}\ge1$ and $a_{C,W}\ge2$,
\begin{align}\label{supp:bandwidth_condition_3}
\begin{split}
    (\log n)\rho_{\mu,2}^2
    &\le(\mathrm{const.})(\log n)\left\{n^{-4/(d+4)}+n^{-1}\right\}\to0,\\
    (\log n)\rho_{C,2}^2
    &\le(\mathrm{const.})(\log n)\left\{n_C^{-2/(d+2)}+n_C^{-1}\right\}\to0.
\end{split}
\end{align}
The second limit follows from \eqref{eq:sufficient}. Now, \eqref{supp:bandwidth_condition_1}, \eqref{supp:bandwidth_condition_2} and \eqref{supp:bandwidth_condition_3} verify Assumption~\ref{ass:bandwidth} for both schemes. Therefore, by Theorem~\ref{thm:mean_l2_rate} and \eqref{eq:scheme_rho_rates},
\begin{align*}
    \|\hat\mu-\mu\|_2=O_p(h_{\mu,W}^2).
\end{align*}

Since $(N_i)_2v_i\le1$ for $i\in\mathcal I_C$,
\begin{align}\label{eq:gamma_second_term_bound}
    \frac{\sum_{i\in\mathcal I_C}N_i(N_i-1)^2v_iw_i^2}{h_\mu^d}
    \le\frac{\sum_{i\in\mathcal I_C}(N_i-1)w_i^2}{h_\mu^d}
    \le\frac{V_{\mu,1}}{h_\mu^d}
    \le\rho_{\mu,2}^2.
\end{align}
For OBS,
\begin{align*}
    w_i&=\frac{1}{na_{\mu,O}},\quad i=1,\ldots,n,\\
    v_i&=\frac{1}{n_Ca_{C,O}},\quad i\in\mathcal I_C.
\end{align*}
By \eqref{eq:mean_cov_leverage},
\eqref{supp:scheme_bandwidth}, \eqref{eq:scheme_rho_rates} and \eqref{eq:gamma_second_term_bound},
\begin{align}\label{eq:obs_gamma_rate_scale}
    \Gamma_{\mu C}
    \le\frac{1}{n^2a_{\mu,O}^2h_{\mu,O}^{2d}}
    +\rho_{\mu,2}^2
    \le(\mathrm{const.})h_{\mu,O}^4.
\end{align}
By \eqref{eq:scheme_rho_rates} and \eqref{eq:obs_gamma_rate_scale},
\begin{align*}
    r_{\mu,2}&=h_{\mu,O}^2+\rho_{\mu,2}=O(h_{\mu,O}^2),\\
    r_{\mu C,2}&=r_{\mu,2}+\Gamma_{\mu C}^{1/2}=O(h_{\mu,O}^2),\\
    r_{C,2}&=r_{\mu C,2}+h_{C,O}^2+\rho_{C,2}=O(h_{\mu,O}^2+h_{C,O}^2).
\end{align*}

For SUBJ,
\begin{align*}
    w_i&=\frac{1}{nN_i},\quad i=1,\ldots,n,\\
    v_i&=\frac{1}{n_C(N_i)_2},\quad i\in\mathcal I_C.
\end{align*}
By the definition of $a_{\mu,S}$,
\begin{align}\label{supp:subj_scheme_sample_size}
    \sum_{i\in\mathcal I_C}N_i^{-2}
    \le\sum_{i=1}^nN_i^{-1}
    =\frac{n}{a_{\mu,S}}.
\end{align}
By \eqref{eq:mean_cov_leverage},
\eqref{supp:scheme_bandwidth}, \eqref{eq:scheme_rho_rates}, \eqref{eq:gamma_second_term_bound} and \eqref{supp:subj_scheme_sample_size},
and using the additional condition in
Corollary~\ref{cor:subj_rates},
\begin{align}\label{eq:subj_gamma_rate_scale}
\begin{split}
    \Gamma_{\mu C}
    \le
    \frac{\sum_{i\in\mathcal I_C}N_i^{-2}}
    {n^2n_Ch_{\mu,S}^{2d}}
    +\rho_{\mu,2}^2
    \le
    \frac{1}{na_{\mu,S}h_{\mu,S}^d}
    \frac{1}{n_Ch_{\mu,S}^d}
    +\rho_{\mu,2}^2
    \le(\mathrm{const.})h_{\mu,S}^4.
\end{split}
\end{align} 
By \eqref{eq:scheme_rho_rates} and \eqref{eq:subj_gamma_rate_scale},
\begin{align*}
    r_{\mu,2}&=h_{\mu,S}^2+\rho_{\mu,2}=O(h_{\mu,S}^2),\\
    r_{\mu C,2}&=r_{\mu,2}+\Gamma_{\mu C}^{1/2}=O(h_{\mu,S}^2),\\
    r_{C,2}&=r_{\mu C,2}+h_{C,S}^2+\rho_{C,2}=O(h_{\mu,S}^2+h_{C,S}^2).
\end{align*}

\section{Proofs for uniform results}

\subsection{Lemmas for uniform results}

\begin{lemma}\label{supp:finite_net}
Let $(\mathcal T,d_{\mathcal T})$ be a compact metric space and let $\mathcal I_n\subset\{1,\ldots,n\}$ be an index set. Also, let $\{F_i:i\in\mathcal I_n\}$ be independent random continuous functions from $\mathcal T$ to $\mathbb R$, and let $\mathcal N_n\subset\mathcal T$ be a deterministic set satisfying $|\mathcal N_n|\le n^a$ for some constant $a>0$. For every $q\ge2$ such that $\sum_{i\in\mathcal I_n}\E\|F_i\|_\infty^q<\infty$,
\begin{align*}
    &\max_{z\in\mathcal N_n}\left|\sum_{i\in\mathcal I_n}\{F_i(z)-\E F_i(z)\}\right|\\
    &=O_p\left[
    \left\{(\log n)\sup_{u\in\mathcal T}\sum_{i\in\mathcal I_n}\E F_i(u)^2\right\}^{1/2}
    +(\log n)^{1-1/q}
    \left\{\sum_{i\in\mathcal I_n}\E\|F_i\|_\infty^q\right\}^{1/q}
    \right].
\end{align*}
\end{lemma}

\begin{proof}
Define
\begin{align*}
    \tau_n&:=\left\{\frac{1}{\log n}\sum_{i\in\mathcal I_n}\E\|F_i\|_\infty^q\right\}^{1/q},\\
    F_{i,\tau_n}&:=F_i\mathbf 1\{\|F_i\|_\infty\le\tau_n\},\quad i\in\mathcal I_n.
\end{align*}
Since $F_i-F_{i,\tau_n}=F_i\mathbf 1\{\|F_i\|_\infty>\tau_n\}$, the triangle inequality gives
\begin{align*}
    \E\left\|
    \sum_{i\in\mathcal I_n}
    \{F_i-F_{i,\tau_n}-\E(F_i-F_{i,\tau_n})\}
    \right\|_\infty
    &\le
    \sum_{i\in\mathcal I_n}
    \E\|F_i-F_{i,\tau_n}\|_\infty
    +\sum_{i\in\mathcal I_n}
    \|\E(F_i-F_{i,\tau_n})\|_\infty\\
    &\le
    2\sum_{i\in\mathcal I_n}
    \E\left[
    \|F_i\|_\infty
    \mathbf 1\{\|F_i\|_\infty>\tau_n\}
    \right]\\
    &\le
    2\sum_{i\in\mathcal I_n}
    \E\left[
    \|F_i\|_\infty^q\tau_n^{1-q}
    \mathbf 1\{\|F_i\|_\infty>\tau_n\}
    \right]\\
    &\le 2\tau_n^{1-q}
    \sum_{i\in\mathcal I_n}\E\|F_i\|_\infty^q\\
    &=2\tau_n\log n.
\end{align*}
By Markov's inequality,
\begin{align}\label{eq:finite_net_truncation_bound}
\begin{split}
    \left\|
    \sum_{i\in\mathcal I_n}
    \{F_i-F_{i,\tau_n}-\E(F_i-F_{i,\tau_n})\}
    \right\|_\infty
    &=O_p(\tau_n\log n)\\
    &=O_p\left[(\log n)^{1-1/q}
    \left\{\sum_{i\in\mathcal I_n}\E\|F_i\|_\infty^q\right\}^{1/q}\right].
\end{split}
\end{align}
Fix $z\in\mathcal N_n$. Note that the random variables $\{F_{i,\tau_n}(z)-\E F_{i,\tau_n}(z):i\in\mathcal I_n\}$ are independent and centered. Moreover,
\begin{align*}
    |F_{i,\tau_n}(z)-\E F_{i,\tau_n}(z)|
    &\le |F_{i,\tau_n}(z)|+\E|F_{i,\tau_n}(z)|\le2\tau_n,\\
    \sum_{i\in\mathcal I_n}\operatorname{Var}\{F_{i,\tau_n}(z)\}
    &\le\sum_{i\in\mathcal I_n}\E\{F_{i,\tau_n}(z)^2\}\le\sum_{i\in\mathcal I_n}\E F_i(z)^2.
\end{align*}
Hence, Bernstein's inequality gives
\begin{align*}
    \P\left[
    \left|
    \sum_{i\in\mathcal I_n}
    \{F_{i,\tau_n}(z)-\E F_{i,\tau_n}(z)\}
    \right|
    >
    \left\{2x\sum_{i\in\mathcal I_n}\E F_i(z)^2\right\}^{1/2}
    +\frac{4}{3}\tau_n x
    \right]\le2e^{-x}, \quad x>0.
\end{align*}
Since $|\mathcal N_n|\le n^a$, the union bound gives
\begin{align*}
    &\P\left[
    \max_{z\in\mathcal N_n}
    \left|
    \sum_{i\in\mathcal I_n}\{F_{i,\tau_n}(z)-\E F_{i,\tau_n}(z)\}
    \right|
    >
    \left\{2x\sup_{u\in\mathcal T}\sum_{i\in\mathcal I_n}\E F_i(u)^2\right\}^{1/2}
    +\frac{4}{3}\tau_n x
    \right]\\
    &\le
    \sum_{z\in\mathcal N_n}
    \P\left[
    \left|
    \sum_{i\in\mathcal I_n}\{F_{i,\tau_n}(z)-\E F_{i,\tau_n}(z)\}
    \right|
    >
    \left\{2x\sup_{u\in\mathcal T}\sum_{i\in\mathcal I_n}\E F_i(u)^2\right\}^{1/2}
    +\frac{4}{3}\tau_n x
    \right]\\
    &\le
    \sum_{z\in\mathcal N_n}
    \P\left[
    \left|
    \sum_{i\in\mathcal I_n}\{F_{i,\tau_n}(z)-\E F_{i,\tau_n}(z)\}
    \right|
    >
    \left\{2x\sum_{i\in\mathcal I_n}\E F_i(z)^2\right\}^{1/2}
    +\frac{4}{3}\tau_n x
    \right]\\
    &\le2|\mathcal N_n|e^{-x}\\
    &\le2n^ae^{-x},\quad x>0.
\end{align*}
Taking $x=(a+2)\log n>0$ gives $2n^ae^{-x}=2n^{-2}\to0$, which implies
\begin{align}\label{eq:finite_net_truncated_bound}
\begin{split}
    &\max_{z\in\mathcal N_n}
    \left|
    \sum_{i\in\mathcal I_n}
    \{F_{i,\tau_n}(z)-\E F_{i,\tau_n}(z)\}
    \right|\\
    &=O_p\left[
    \left\{(\log n)\sup_{z\in\mathcal T}
    \sum_{i\in\mathcal I_n}\E F_i(z)^2\right\}^{1/2}
    +\tau_n\log n
    \right]\\
    &=O_p\left[
    \left\{(\log n)\sup_{u\in\mathcal T}\sum_{i\in\mathcal I_n}\E F_i(u)^2\right\}^{1/2}
    +(\log n)^{1-1/q}
    \left\{\sum_{i\in\mathcal I_n}\E\|F_i\|_\infty^q\right\}^{1/q}
    \right].
\end{split}
\end{align}
By \eqref{eq:finite_net_truncation_bound}, \eqref{eq:finite_net_truncated_bound} and the triangle inequality,
\begin{align*}
    &\max_{z\in\mathcal N_n}
    \left|
    \sum_{i\in\mathcal I_n}\{F_i(z)-\E F_i(z)\}
    \right|\\
    &=\max_{z\in\mathcal N_n}
    \left|
    \sum_{i\in\mathcal I_n}\{F_{i,\tau_n}(z)-\E F_{i,\tau_n}(z)\}
    +\sum_{i\in\mathcal I_n}
    \left[
    F_i(z)-F_{i,\tau_n}(z)
    -\E\{F_i(z)-F_{i,\tau_n}(z)\}
    \right]
    \right|\\
    &\le
    \max_{z\in\mathcal N_n}
    \left|
    \sum_{i\in\mathcal I_n}
    \{F_{i,\tau_n}(z)-\E F_{i,\tau_n}(z)\}
    \right|
    +
    \left\|
    \sum_{i\in\mathcal I_n}
    \{F_i-F_{i,\tau_n}-\E(F_i-F_{i,\tau_n})\}
    \right\|_\infty\\
    &=O_p\left[
    \left\{(\log n)\sup_{u\in\mathcal T}
    \sum_{i\in\mathcal I_n}\E F_i(u)^2\right\}^{1/2}
    +(\log n)^{1-1/q}
    \left\{\sum_{i\in\mathcal I_n}\E\|F_i\|_\infty^q\right\}^{1/q}
    \right],
\end{align*}
which proves the lemma.
\end{proof}

For the next lemma, define
\begin{align*}
    H_C(s,t):=\sum_{i\in\mathcal I_C}v_i\sum_{1\le j\ne k\le N_i}
    \mathcal L_{s,h_C}(s_{ij})\mathcal L_{t,h_C}(s_{ik})
    (|E_{ij}|+|E_{ik}|).
\end{align*}

\begin{lemma}\label{supp:uniform_sums}
Under the assumptions of Theorem~\ref{thm:mean_rate},
\begin{align}\label{eq:uniform_mean_sums}
\begin{split}
    &\|\hat g_\mu-\E\hat g_\mu\|_\infty+\|\hat f_\mu-\E\hat f_\mu\|_\infty\\
    &=O_p\left\{
    \sqrt{\log n}\,\rho_{\mu,2}
    +\left(\frac{\log n}{n}\right)^{1-1/(2\beta)}h_\mu^{-d/(2\beta)}
    +b_\mu h_\mu^{-d}
    \right\}.
\end{split}
\end{align}
Under the assumptions of Theorem~\ref{thm:cov_rate},
\begin{align}\label{eq:uniform_cov_sums}
\begin{split}
    &\|\widetilde g_C-\E\widetilde g_C\|_\infty+\|\hat f_C-\E\hat f_C\|_\infty\\
    &=O_p\left\{
    \sqrt{\log n}\,\rho_{C,2}
    +\left(\frac{\log n}{n_C}\right)^{1-1/(2\beta)}h_C^{-d/\beta}
    +b_Ch_C^{-2d}
    \right\},
\end{split}
\end{align}
and
\begin{align}\label{eq:positive_first_moment_pair_bound}
    \|H_C\|_\infty
    =O_p\left\{
    1+\sqrt{\log n}\,\rho_{C,2}
    +\left(\frac{\log n}{n_C}\right)^{1-1/(2\beta)}h_C^{-d/\beta}
    +b_Ch_C^{-2d}
    \right\}.
\end{align}
\end{lemma}

\begin{proof}
Define
\begin{align*}
    M_i(h):=\sup_{s\in\mathcal M}\sum_{j=1}^{N_i}
    \mathbf 1\{d_{\mathcal M}(s,s_{ij})\le h\},\quad h\in\{h_\mu,h_C\}.
\end{align*}
By Assumption \ref{ass:bandwidth}, we may assume that $h\in\{h_\mu,h_C\}$ is sufficiently small throughout this proof. Let $\mathcal J_h$ be a maximal
$h$-separated subset of $\mathcal M$.
By maximality and the argument leading to
\eqref{eq:denominator_ball_volume} and
\eqref{eq:denominator_net_size},
$\mathcal J_h$ is an $h$-net with
$|\mathcal J_h|\le(\mathrm{const.})h^{-d}$.
For every $s\in\mathcal M$, there exists $s'\in\mathcal J_h$
such that
$\overline B_{\mathcal M}(s,h)
\subset\overline B_{\mathcal M}(s',2h)$.
Hence,
\begin{align*}
    M_i(h)
    \le\max_{s'\in\mathcal J_h}
    \sum_{j=1}^{N_i}
    \mathbf 1\{s_{ij}\in\overline B_{\mathcal M}(s',2h)\}.
\end{align*}
By Assumption~\ref{ass:sampling},
\eqref{eq:change_of_variables} and \Cref{lem:theta_bounds},
\begin{align*}
    \P\{s_{ij}\in\overline B_{\mathcal M}(s',2h)\}
    \le(\mathrm{const.})h^d.
\end{align*}
Since the observation locations are independent within subjects,
Bernstein's inequality and the union bound give, for some constant
$\Lambda>0$ independent of $i$, $n$ and $h$,
\begin{align}\label{eq:m_i_bound}
    \P\{M_i(h)>\Lambda(N_ih^d+x)\}
    \le(\mathrm{const.})h^{-d}e^{-x},
    \quad x>0.
\end{align}
For every $q\ge1$, by Assumption~\ref{ass:bandwidth}, \eqref{eq:m_i_bound} and
the tail-integral formula,
\begin{align*}
    \E\{M_i(h)^q\}
    \le(\mathrm{const.})\{N_ih^d+\log(1/h)\}^q\le(\mathrm{const.})(N_ih^d+\log n)^q.
\end{align*}
Since $M_i(h)\le N_i$,
\begin{align}\label{eq:occupancy_moment_envelope}
    \E\{M_i(h)^q\}
    \le(\mathrm{const.})\min\{N_i,N_ih^d+\log n\}^q.
\end{align}

We now obtain envelope moments. For the mean side, by
Assumptions~\ref{ass:sampling} and~\ref{ass:uniform_moments},
\begin{align}\label{eq:uniform_response_moment}
    \E(|Y_{ij}|^{2\beta}\mid\mathcal S)
    \le(\mathrm{const.}).
\end{align}
By Lemma~\ref{supp:kernel},
\eqref{eq:uniform_response_moment}, H\"older's inequality
and $M_i(h_\mu)\le N_i$,
\begin{align*}
\begin{split}
    \E\left[
    \|W_{\mu,i}\|_\infty^{2\beta}
    +\|Z_{\mu,i}\|_\infty^{2\beta}
    \,\middle|\,\mathcal S
    \right]\le(\mathrm{const.})w_i^{2\beta}h_\mu^{-2\beta d}
    N_iM_i(h_\mu)^{2\beta-1}.
\end{split}
\end{align*}
Taking expectations, applying
\eqref{eq:occupancy_moment_envelope}
and summing over subjects gives
\begin{align}\label{eq:uniform_mean_envelope_moments}
\begin{split}
    \sum_{i=1}^n
    \left\{
    \E\|W_{\mu,i}\|_\infty^{2\beta}
    +\E\|Z_{\mu,i}\|_\infty^{2\beta}
    \right\}\le(\mathrm{const.})h_\mu^{-2\beta d}
    \sum_{i=1}^nw_i^{2\beta}N_i
    \min\{N_i,N_ih_\mu^d+\log n\}^{2\beta-1}.
\end{split}
\end{align}
For the covariance side, define
\begin{align*}
    U_i(s):=X_i(s)-\mu(s),\quad i\in\mathcal I_C.
\end{align*}
Note that
\begin{align*}
    E_{ij}E_{ik}
    &=U_i(s_{ij})U_i(s_{ik})
    +U_i(s_{ij})\epsilon_{ik}+\epsilon_{ij}U_i(s_{ik})
    +\epsilon_{ij}\epsilon_{ik},
    \quad 1\le j\ne k\le N_i.
\end{align*}
We decompose $W_{C,i}$ defined in \eqref{eq:cov_subject_contribution} as
\begin{align}\label{eq:cov_subject_decomposition}
    W_{C,i}
    =W_{C,i}^{UU}
    +W_{C,i}^{U\epsilon}
    +W_{C,i}^{\epsilon U}
    +W_{C,i}^{\epsilon\epsilon},
    \quad i\in\mathcal I_C,
\end{align} 
where
\begin{align*}
    W_{C,i}^{UU}(s,t)
    &:=v_i\sum_{1\le j\ne k\le N_i}
    \mathcal L_{s,h_C}(s_{ij})\mathcal L_{t,h_C}(s_{ik})
    U_i(s_{ij})U_i(s_{ik}),\\
    W_{C,i}^{U\epsilon}(s,t)
    &:=v_i\sum_{1\le j\ne k\le N_i}
    \mathcal L_{s,h_C}(s_{ij})\mathcal L_{t,h_C}(s_{ik})
    U_i(s_{ij})\epsilon_{ik},\\
    W_{C,i}^{\epsilon U}(s,t)
    &:=v_i\sum_{1\le j\ne k\le N_i}
    \mathcal L_{s,h_C}(s_{ij})\mathcal L_{t,h_C}(s_{ik})
    \epsilon_{ij}U_i(s_{ik}),\\
    W_{C,i}^{\epsilon\epsilon}(s,t)
    &:=v_i\sum_{1\le j\ne k\le N_i}
    \mathcal L_{s,h_C}(s_{ij})\mathcal L_{t,h_C}(s_{ik})
    \epsilon_{ij}\epsilon_{ik}.
\end{align*}
By Lemma~\ref{supp:kernel},
\begin{align*}
    \|W_{C,i}^{UU}\|_\infty
    &\le\|U_i\|_\infty^2\|Z_{C,i}\|_\infty,\\
    \|Z_{C,i}\|_\infty
    &\le(\mathrm{const.})v_ih_C^{-2d}M_i(h_C)^2.
\end{align*}
Since $U_i$ is independent of $\mathcal S$ and
$\E\|U_i\|_\infty^{2\beta}<\infty$,
applying
\eqref{eq:occupancy_moment_envelope}
and summing over subjects gives
\begin{align}\label{eq:uniform_cov_path_occupancy_bound}
\begin{split}
    \sum_{i\in\mathcal I_C}
    \left\{
    \E\|W_{C,i}^{UU}\|_\infty^\beta
    +\E\|Z_{C,i}\|_\infty^\beta
    \right\}
    \le(\mathrm{const.})h_C^{-2\beta d}
    \sum_{i\in\mathcal I_C}v_i^\beta
    \min\{N_i,N_ih_C^d+\log n\}^{2\beta}.
\end{split}
\end{align}
Similarly to \eqref{eq:uniform_response_moment}, Assumptions~\ref{ass:sampling} and~\ref{ass:uniform_moments} give
\begin{align}
\begin{split}
    &\E\left[
    |U_i(s_{ij})\epsilon_{ik}|^{2\beta}
    +|\epsilon_{ij}U_i(s_{ik})|^{2\beta}
    +|\epsilon_{ij}\epsilon_{ik}|^{2\beta}
    +(|E_{ij}|+|E_{ik}|)^{2\beta}
    \,\middle|\,\mathcal S
    \right]\\
    &\le(\mathrm{const.}),\quad 1\le j\ne k\le N_i.
\end{split}
\label{eq:uniform_cov_noise_multiplier_moments}
\end{align}
Define
\begin{align*}
    H_{C,i}(s,t)
    :=v_i\sum_{1\le j\ne k\le N_i}
    \mathcal L_{s,h_C}(s_{ij})\mathcal L_{t,h_C}(s_{ik})
    (|E_{ij}|+|E_{ik}|),\quad i\in\mathcal I_C.
\end{align*}
Note that the number of ordered pairs contributing to the kernel product
is at most $M_i(h_C)^2$.
Hence, by Lemma~\ref{supp:kernel}, H\"older's inequality
and \eqref{eq:uniform_cov_noise_multiplier_moments},
\begin{align*}
\begin{split}
    &\E\left[
    \|W_{C,i}^{U\epsilon}\|_\infty^{2\beta}
    +\|W_{C,i}^{\epsilon U}\|_\infty^{2\beta}
    +\|W_{C,i}^{\epsilon\epsilon}\|_\infty^{2\beta}
    +\|H_{C,i}\|_\infty^{2\beta}
    \,\middle|\,\mathcal S
    \right]\\
    &\le(\mathrm{const.})v_i^{2\beta}h_C^{-4\beta d}
    (N_i)_2M_i(h_C)^{4\beta-2}.
\end{split}
\end{align*}
Taking expectations, applying
\eqref{eq:occupancy_moment_envelope}
and summing over subjects gives
\begin{align}\label{eq:uniform_cov_noise_envelope_moments}
\begin{split}
    &\sum_{i\in\mathcal I_C}
    \left\{
    \E\|W_{C,i}^{U\epsilon}\|_\infty^{2\beta}
    +\E\|W_{C,i}^{\epsilon U}\|_\infty^{2\beta}
    +\E\|W_{C,i}^{\epsilon\epsilon}\|_\infty^{2\beta}
    +\E\|H_{C,i}\|_\infty^{2\beta}
    \right\}\\
    &\le(\mathrm{const.})h_C^{-4\beta d}
    \sum_{i\in\mathcal I_C}v_i^{2\beta}(N_i)_2
    \min\{N_i,N_ih_C^d+\log n\}^{4\beta-2}.
\end{split}
\end{align}
Assumptions~\ref{ass:sampling}
and~\ref{ass:uniform_moments} and the same ordered-pair calculation used to derive
\eqref{eq:cov_subject_second_moment_bound} give
\begin{align}
\begin{split}
    &\sum_{i\in\mathcal I_C}
    \left[
    \E\{W_{C,i}^{UU}(s,t)^2\}
    +\E\{W_{C,i}^{U\epsilon}(s,t)^2\}
    +\E\{W_{C,i}^{\epsilon U}(s,t)^2\}
    \right.\\
    &\left.
    +\E\{W_{C,i}^{\epsilon\epsilon}(s,t)^2\}
    +\E\{Z_{C,i}(s,t)^2\}
    +\E\{H_{C,i}(s,t)^2\}
    \right]\\
    &\le(\mathrm{const.})\rho_{C,2}^2
\end{split}
\label{eq:uniform_cov_subject_second_moments}
\end{align}
uniformly over
$(s,t)\in\mathcal M\times\mathcal M$.

We now simplify the envelope moments. For every $h>0$,
\begin{align}\label{eq:uniform_occupancy_split}
    \min\{N_i,N_ih^d+\log n\}
    \le N_ih^d
    +N_i\min\left\{1,\frac{\log n}{N_i}\right\}.
\end{align}
Fix $q>1$. By Assumption~\ref{ass:weights}, \eqref{eq:weight_mean_normalization} and \eqref{eq:weight_covariance_normalization}, and arguing as in \eqref{eq:cor_variance_bound_1},
\begin{align}
\begin{split}
    \sum_{i=1}^n(N_iw_i)^q
    &\le
    \left\{\max_{1\le i\le n}(N_iw_i)\right\}^{q-1}
    \sum_{i=1}^nN_iw_i
    \le\frac{(\mathrm{const.})}{n^{q-1}},\\
    \sum_{i\in\mathcal I_C}\{(N_i)_2v_i\}^q
    &\le
    \left[\max_{i\in\mathcal I_C}\{(N_i)_2v_i\}\right]^{q-1}
    \sum_{i\in\mathcal I_C}(N_i)_2v_i
    \le\frac{(\mathrm{const.})}{n_C^{q-1}}.
\end{split}
\label{eq:uniform_weight_power_bounds}
\end{align}
In applying Lemma~\ref{supp:finite_net} below,
we simplify the envelope terms using
\eqref{eq:uniform_occupancy_split},
\eqref{eq:uniform_weight_power_bounds} and the inequalities
$(a+b)^p\le2^{p-1}(a^p+b^p)$
for $a,b\ge0$ and $p\ge1$,
and $N_i^2\le2(N_i)_2$ for $i\in\mathcal I_C$.

We first establish the desired bounds on a finite net. Use the same constant $A$, sequence $\delta_n$ and set $\mathcal G_n$ defined in the proof of \Cref{supp:denominators}. By \eqref{eq:denominator_net_size}, for all sufficiently large $n$,
\begin{align}\label{eq:delta_n_polynomial_bound}
\begin{split}
    |\mathcal G_n|
    &\le(\mathrm{const.})\delta_n^{-d}
    \le(\mathrm{const.})n^{d\{A(2d+1)+2\}},\\
    |\mathcal G_n\times\mathcal G_n|
    &\le(\mathrm{const.})\delta_n^{-2d}
    \le(\mathrm{const.})n^{2d\{A(2d+1)+2\}}.
\end{split}
\end{align}
For the mean side, applying Lemma~\ref{supp:finite_net}
with $q=2\beta$ to $W_{\mu,i}$ and $Z_{\mu,i}$, and using
\eqref{eq:mean_subject_second_moment_bound},
\eqref{eq:uniform_mean_envelope_moments},
\eqref{eq:delta_n_polynomial_bound}
and the preceding envelope simplification, we obtain
\begin{align}\label{eq:uniform_mean_net_bound}
\begin{split}
    &\max_{s'\in\mathcal G_n}|\hat g_\mu(s')-\E\hat g_\mu(s')|
    +\max_{s'\in\mathcal G_n}|\hat f_\mu(s')-\E\hat f_\mu(s')|\\
    &=O_p\left\{
    \sqrt{\log n}\,\rho_{\mu,2}
    +\left(\frac{\log n}{n}\right)^{1-1/(2\beta)}
    h_\mu^{-d/(2\beta)}
    +b_\mu h_\mu^{-d}
    \right\}.
\end{split}
\end{align}
For the covariance side, Assumption~\ref{ass:bandwidth} and \eqref{eq:variance_scale_lower_bounds} imply, for all sufficiently large $n$,
\begin{align}\label{eq:uniform_path_variance_absorption}
    \left(\frac{\log n}{n_C}\right)^{1-1/\beta}
    \le\sqrt{\frac{\log n}{n_C}}
    \le2\sqrt{\log n}\,\rho_{C,2}.
\end{align}
We apply Lemma~\ref{supp:finite_net} with $q=\beta$ to
$W_{C,i}^{UU}$ and $Z_{C,i}$, and with $q=2\beta$ to
$W_{C,i}^{U\epsilon}$, $W_{C,i}^{\epsilon U}$,
$W_{C,i}^{\epsilon\epsilon}$ and $H_{C,i}$.
By \eqref{eq:cov_subject_decomposition},
\eqref{eq:uniform_cov_path_occupancy_bound},
\eqref{eq:uniform_cov_noise_envelope_moments},
\eqref{eq:uniform_cov_subject_second_moments},
\eqref{eq:delta_n_polynomial_bound}
and \eqref{eq:uniform_path_variance_absorption},
the preceding envelope simplification
and the triangle inequality give
\begin{align}\label{eq:uniform_cov_net_bound}
\begin{split}
    &\max_{(s',t')\in\mathcal G_n\times\mathcal G_n}
    |\widetilde g_C(s',t')-\E\widetilde g_C(s',t')|
    +\max_{(s',t')\in\mathcal G_n\times\mathcal G_n}
    |\hat f_C(s',t')-\E\hat f_C(s',t')|\\
    &+
    \max_{(s',t')\in\mathcal G_n\times\mathcal G_n}
    |H_C(s',t')-\E H_C(s',t')|\\
    &=O_p\left\{
    \sqrt{\log n}\,\rho_{C,2}
    +\left(\frac{\log n}{n_C}\right)^{1-1/(2\beta)}
    h_C^{-d/\beta}
    +b_Ch_C^{-2d}
    \right\}.
\end{split}
\end{align}

It remains to extend \eqref{eq:uniform_mean_net_bound}
and \eqref{eq:uniform_cov_net_bound} to the full domains.
By \eqref{eq:weight_mean_normalization},
\eqref{eq:weight_covariance_normalization}, \eqref{eq:uniform_response_moment},
\eqref{eq:uniform_cov_noise_multiplier_moments}
and the Cauchy-Schwarz inequality,
\begin{align}\label{eq:uniform_lipschitz_weight_bounds}
\begin{split}
    \E\left[
    \sum_{i=1}^nw_i\sum_{j=1}^{N_i}|Y_{ij}|
    \right]
    &\le(\mathrm{const.}),\\
    \E\left[
    \sum_{i\in\mathcal I_C}v_i
    \sum_{1\le j\ne k\le N_i}
    \{|E_{ij}E_{ik}|+|E_{ij}|+|E_{ik}|\}
    \right]
    &\le(\mathrm{const.}).
\end{split}
\end{align}
By Lemma~\ref{supp:kernel},
the increments of the mean-side kernels and the product kernels
are bounded by $(\mathrm{const.})h_\mu^{-d-1}d_{\mathcal M}(s,s')$
and $(\mathrm{const.})h_C^{-2d-1}
d_{\mathcal M\times\mathcal M}((s,t),(s',t'))$, respectively.
By the triangle inequality, the absolute increment after centering
is bounded by the absolute increment before centering plus its expectation.
Hence, by \eqref{eq:uniform_lipschitz_weight_bounds},
\begin{align}\label{eq:uniform_mean_interpolation_increment}
    \sup_{\substack{s,s'\in\mathcal M\\
                    d_{\mathcal M}(s,s')\le\delta_n}}
    |(F-\E F)(s)-(F-\E F)(s')|
    =O_p(h_\mu^{-d-1}\delta_n)
\end{align}
for $F\in\{\hat g_\mu,\hat f_\mu\}$, and
\begin{align}\label{eq:uniform_cov_interpolation_increment}
\begin{split}
    \sup_{\substack{
        (s,t),(s',t'):
        d_{\mathcal M\times\mathcal M}((s,t),(s',t'))
        \le2\delta_n}}
    |(F-\E F)(s,t)-(F-\E F)(s',t')|=O_p(h_C^{-2d-1}\delta_n)
\end{split}
\end{align}
for $F\in\{\widetilde g_C,\hat f_C,H_C\}$. Note that $\mathcal G_n\times\mathcal G_n$ is a
$2\delta_n$-net of $\mathcal M\times\mathcal M$ under
\eqref{eq:product_manifold_metric}.
Hence, by \eqref{eq:variance_scale_lower_bounds}
and \eqref{eq:denominator_interpolation},
\begin{align}\label{eq:uniform_interpolation_remainders}
\begin{split}
    h_\mu^{-d-1}\delta_n
    &=o(n^{-2})
    =o\left(\sqrt{\log n}\,\rho_{\mu,2}\right),\\
    h_C^{-2d-1}\delta_n
    &=O(n^{-2})
    =o\left(\sqrt{\log n}\,\rho_{C,2}\right).
\end{split}
\end{align}
Thus, \eqref{eq:uniform_mean_net_bound},
\eqref{eq:uniform_mean_interpolation_increment}
and \eqref{eq:uniform_interpolation_remainders}
prove \eqref{eq:uniform_mean_sums}. For the mean assertion alone, the finite-net construction above
applies with $A$ chosen using only the mean-side bandwidth condition.
Similarly, by \eqref{eq:uniform_cov_net_bound},
\eqref{eq:uniform_cov_interpolation_increment}
and \eqref{eq:uniform_interpolation_remainders},
\begin{align}\label{eq:uniform_cov_sums_proof}
\begin{split}
    &\|\widetilde g_C-\E\widetilde g_C\|_\infty
    +\|\hat f_C-\E\hat f_C\|_\infty
    +\|H_C-\E H_C\|_\infty\\
    &=O_p\left\{
    \sqrt{\log n}\,\rho_{C,2}
    +\left(\frac{\log n}{n_C}\right)^{1-1/(2\beta)}
    h_C^{-d/\beta}
    +b_Ch_C^{-2d}
    \right\},
\end{split}
\end{align}
which proves \eqref{eq:uniform_cov_sums}.

It remains to prove \eqref{eq:positive_first_moment_pair_bound}.
By \eqref{eq:uniform_cov_noise_multiplier_moments}
and conditional Jensen's inequality,
\begin{align*}
    \E\{H_C(s,t)\mid\mathcal S\}
    \le(\mathrm{const.})\hat f_C(s,t).
\end{align*}
Hence, by Assumption~\ref{ass:sampling}, \eqref{eq:kernel_normalization} and
\eqref{eq:denominator_expectation_representations},
\begin{align}\label{eq:uniform_H_expectation_bound}
    \|\E H_C\|_\infty
    \le(\mathrm{const.})\|\E\hat f_C\|_\infty
    \le(\mathrm{const.}).
\end{align}
Combining \eqref{eq:uniform_cov_sums_proof}
and \eqref{eq:uniform_H_expectation_bound}
with the triangle inequality proves
\eqref{eq:positive_first_moment_pair_bound}.
\end{proof}

\subsection{Proof of Theorem~\ref{thm:mean_rate}}

By Lemma~\ref{supp:denominators}, we may simply assume that the event $\mathcal E_\mu$ holds throughout this proof.
By \eqref{eq:mean_error_decomposition},
\begin{align*}
    \|\hat\mu-\mu\|_\infty
    \le\frac{2}{c_f}\left\{
    \|\hat g_\mu-\E\hat g_\mu\|_\infty
    +\|\mu\|_\infty\|\hat f_\mu-\E\hat f_\mu\|_\infty
    +\|\E\hat g_\mu-\mu\E\hat f_\mu\|_\infty
    \right\}.
\end{align*}
Now, \Cref{supp:uniform_sums} and \eqref{eq:mean_numerator_bias} give the theorem.

\subsection{Proof of Theorem~\ref{thm:cov_rate}}

By Lemma~\ref{supp:denominators}, we may simply assume that the event $\mathcal E_\mu\cap\mathcal E_C$ holds throughout this proof.
By \eqref{eq:oracle_covariance_bias}, \eqref{eq:oracle_covariance_decomposition} and \eqref{eq:uniform_cov_sums},
\begin{align}\label{eq:oracle_covariance_uniform}
\begin{split}
    \|\widetilde C-C\|_\infty
    &\le\frac{2}{c_f^2}\|\widetilde g_C-C\hat f_C\|_\infty\\
    &=O_p\left\{
    h_C^2+\sqrt{\log n}\,\rho_{C,2}
    +\left(\frac{\log n}{n_C}\right)^{1-1/(2\beta)}h_C^{-d/\beta}
    +b_Ch_C^{-2d}
    \right\}.
\end{split}
\end{align}
By \eqref{eq:residual_product_difference},
\begin{align*}
    |\hat C(s,t)-\widetilde C(s,t)|
    \le\frac{\|\hat\mu-\mu\|_\infty}{\hat f_C(s,t)}H_C(s,t)
    +\|\hat\mu-\mu\|_\infty^2.
\end{align*}
Hence, by Theorem~\ref{thm:mean_rate} and Lemma~\ref{supp:uniform_sums},
\begin{align}\label{eq:empirical_covariance_uniform}
\begin{split}
    &\|\hat C-\widetilde C\|_\infty\\
    &=O_p\left[
    r_{\mu,\infty}\left\{
    1+\sqrt{\log n}\,\rho_{C,2}
    +\left(\frac{\log n}{n_C}\right)^{1-1/(2\beta)}h_C^{-d/\beta}
    +b_Ch_C^{-2d}
    \right\}
    +r_{\mu,\infty}^2
    \right].
\end{split}
\end{align}
Therefore, by \eqref{eq:oracle_covariance_uniform} and \eqref{eq:empirical_covariance_uniform},
\begin{align*}
    \|\hat C-C\|_\infty
    \le\|\hat C-\widetilde C\|_\infty+\|\widetilde C-C\|_\infty
    =O_p(r_{C,\infty}).
\end{align*}

\subsection{Proof of Theorem~\ref{thm:uniform_operator_eigen}}

By Lemma~\ref{supp:denominators}, we may simply assume that the event $\mathcal E_\mu\cap\mathcal E_C$ holds throughout this proof.
By Theorem~\ref{thm:eigen_rates} and the condition $r_{C,2}\to0$, $\hat\lambda_k$ for $k=1,\ldots,K^*$ are bounded away from zero with probability tending to one.

Since $C$ and $\hat C$ are continuous on $\mathcal M\times\mathcal M$,
both $\mathcal C$ and $\hat{\mathcal C}$
map $L^2(\mathcal M)$ into the space of continuous functions on $\mathcal M$.
Hence, we may choose the continuous representatives
\begin{align*}
    \phi_k=\lambda_k^{-1}\mathcal C\phi_k,
    \quad
    \hat\phi_k=\hat\lambda_k^{-1}\hat{\mathcal C}\hat\phi_k,
    \quad k=1,\ldots,K^*.
\end{align*}
For these representatives, we obtain
\begin{align*}
    \hat\lambda_k(\hat\phi_k-\phi_k)
    =(\hat{\mathcal C}-\mathcal C)\hat\phi_k
    +\mathcal C(\hat\phi_k-\phi_k)
    +(\lambda_k-\hat\lambda_k)\phi_k.
\end{align*}
Since $\|\hat\phi_k\|_2=1$,
$\sup_{s\in\mathcal M}\|C(s,\cdot)\|_2<\infty$
and $\max_{1\le k\le K^*}\|\phi_k\|_\infty<\infty$,
the Cauchy-Schwarz inequality gives
\begin{align*}
\begin{split}
    \max_{1\le k\le K^*}\|\hat\phi_k-\phi_k\|_\infty&\le(\mathrm{const.})\left[
    \|\hat C-C\|_\infty
    +\max_{1\le k\le K^*}
    \left\{
    \|\hat\phi_k-\phi_k\|_2
    +|\hat\lambda_k-\lambda_k|
    \right\}
    \right]\\
    &\le(\mathrm{const.})\left\{
    \|\hat C-C\|_\infty
    +\|\hat{\mathcal C}-\mathcal C\|_{\mathrm{op}}
    \right\}\\
    &\le(\mathrm{const.})\|\hat C-C\|_\infty\\
    &=O_p(r_{C,\infty}),
\end{split}
\end{align*}
where the second inequality follows from
\eqref{eq:eigenvalue_perturbation_bound} and
\eqref{eq:spectral_gap_eigenfunction_bound},
the third inequality follows from that
$\|\hat{\mathcal C}-\mathcal C\|_{\mathrm{op}}
\le\|\hat C-C\|_2
\le V(\mathcal M)\|\hat C-C\|_\infty$,
and the last equality follows from Theorem~\ref{thm:cov_rate}.
This proves the theorem.

\subsection{Proof of Corollary~\ref{cor:uniform_schemes}}

By \eqref{eq:uniform_weight_power_bounds} and $\min\{1,\log n/N_i\}\le1$,
\begin{align}\label{eq:uniform_scheme_b_bounds}
\begin{split}
    b_{\mu,W}
    &\le(\mathrm{const.})
    \left(\frac{\log n}{n}\right)^{1-1/(2\beta)},\\
    b_{C,W}
    &\le(\mathrm{const.})\left\{
    \left(\frac{\log n}{n_C}\right)^{1-1/\beta}
    +\left(\frac{\log n}{n_C}\right)^{1-1/(2\beta)}
    \right\}.
\end{split}
\end{align}
Since $a_{\mu,W}\ge1$, $a_{C,W}\ge2$ and $n_C/\log n\to\infty$, \eqref{eq:uniform_scheme_bandwidths} and \eqref{eq:uniform_scheme_b_bounds} give
\begin{align}\label{eq:uniform_scheme_bandwidth_limits}
    h_{\mu,W}\to0,\quad h_{C,W}\to0.
\end{align}
By \eqref{eq:variance_scale_lower_bounds}, \eqref{eq:scheme_variance_exact} and \eqref{supp:scheme_scales},
\begin{align}\label{eq:uniform_scheme_variance_comparison}
    \rho_{\mu,2}^2
    \asymp\frac{1}{na_{\mu,W}h_\mu^d}+\frac{1}{n},\quad
    \rho_{C,2}^2
    \asymp\frac{1}{n_Ca_{C,W}h_C^{2d}}+\frac{1}{n_C}.
\end{align}
Take $h_\mu=h_{\mu,W}$ and $h_C=h_{C,W}$. The first two terms in each maximum in \eqref{eq:uniform_scheme_bandwidths} imply
\begin{align*}
\begin{split}
    \frac{\log n}{n}
    &\le(\mathrm{const.})h_{\mu,W}^4,\quad
    \frac{\log n}{na_{\mu,W}}
    \le(\mathrm{const.})h_{\mu,W}^{d+4},\\
    \frac{\log n}{n_C}
    &\le(\mathrm{const.})h_{C,W}^4,\quad
    \frac{\log n}{n_Ca_{C,W}}
    \le(\mathrm{const.})h_{C,W}^{2d+4}.
\end{split}
\end{align*}
Hence, by \eqref{eq:uniform_scheme_variance_comparison},
\begin{align}\label{eq:uniform_scheme_variance_rates}
    (\log n)\rho_{\mu,2}^2
    \le(\mathrm{const.})h_{\mu,W}^4,\quad
    (\log n)\rho_{C,2}^2
    \le(\mathrm{const.})h_{C,W}^4.
\end{align}
Moreover,
\begin{align*}
    h_{\mu,W}\ge(\mathrm{const.})\left(\frac{\log n}{n}\right)^{1/4},\quad
    h_{C,W}\ge(\mathrm{const.})\left(\frac{\log n}{n_C}\right)^{1/4} \ge (\mathrm{const.})\left(\frac{\log n}{n}\right)^{1/4},
\end{align*}
which implies
\begin{align}\label{eq:uniform_scheme_log_bandwidth}
    \log(1/h_{\mu,W})=O(\log n),\quad
    \log(1/h_{C,W})=O(\log n).
\end{align}
By \eqref{eq:uniform_scheme_bandwidth_limits}, \eqref{eq:uniform_scheme_variance_rates} and \eqref{eq:uniform_scheme_log_bandwidth}, Assumption~\ref{ass:bandwidth} holds.

The last two terms in each maximum in \eqref{eq:uniform_scheme_bandwidths} give
\begin{align}\label{eq:uniform_scheme_remainder_rates}
\begin{split}
    \left(\frac{\log n}{n}\right)^{1-1/(2\beta)}
    h_{\mu,W}^{-d/(2\beta)}
    +b_{\mu,W}h_{\mu,W}^{-d}
    &\le(\mathrm{const.})h_{\mu,W}^2,\\
    \left(\frac{\log n}{n_C}\right)^{1-1/(2\beta)}
    h_{C,W}^{-d/\beta}
    +b_{C,W}h_{C,W}^{-2d}
    &\le(\mathrm{const.})h_{C,W}^2.
\end{split}
\end{align}
By \eqref{eq:uniform_scheme_bandwidth_limits}, \eqref{eq:uniform_scheme_variance_rates} and \eqref{eq:uniform_scheme_remainder_rates},
\begin{align}\label{eq:uniform_scheme_rate_scales}
\begin{split}
    r_{\mu,\infty}
    &\le(\mathrm{const.})h_{\mu,W}^2,\\
    r_{C,\infty}
    &\le(\mathrm{const.})
    (h_{\mu,W}^2+h_{C,W}^2)(1+h_{\mu,W}^2)\le(\mathrm{const.})(h_{\mu,W}^2+h_{C,W}^2).
\end{split}
\end{align}
By Theorems~\ref{thm:mean_rate} and~\ref{thm:cov_rate} and \eqref{eq:uniform_scheme_rate_scales}, we obtain the desired mean and covariance rates.

For all sufficiently large $n$,
\begin{align*}
    N_i^{2\beta}
    \min\left\{1,\frac{\log n}{N_i}\right\}^{2\beta-1}
    \ge1.
\end{align*}
Hence, $\max_{1\le i\le n}w_i\le b_{\mu,W}$.
By \eqref{eq:weight_covariance_normalization}
and \eqref{eq:gamma_second_term_bound},
\begin{align*}
    \Gamma_{\mu C}^{1/2}
    \le h_{\mu,W}^{-d}\max_{1\le i\le n}w_i+\rho_{\mu,2}
    \le b_{\mu,W}h_{\mu,W}^{-d}+\rho_{\mu,2}.
\end{align*}
Therefore, by \eqref{eq:uniform_scheme_bandwidth_limits},
\eqref{eq:uniform_scheme_variance_rates}
and \eqref{eq:uniform_scheme_remainder_rates},
\begin{align*}
    r_{C,2}
    =O(h_{\mu,W}^2+h_{C,W}^2)\to0.
\end{align*}
If Assumption~\ref{ass:eigengap} further holds,
Theorem~\ref{thm:uniform_operator_eigen}
and \eqref{eq:uniform_scheme_rate_scales}
give the desired eigenfunction rate.


\begin{thebibliography}{99}

\bibitem[Ali et al.(2021)]{AliEtAl2021} Ali, M. H., Wainwright, B., Petersen, A., Jonnadula, G. B., Desai, M., Rao, H. L., Srinivas, M. B., Jammalamadaka, S. R., Senthil, S. and Pyne, S. (2021). Circular functional analysis of OCT data for precise identification of structural phenotypes in the eye. {\it Scientific Reports}, \textbf{11}, 23336.

\bibitem[Ali et al.(2023)]{AliEtAl2023}
Ali, M. H., Ray, M., Jammalamadaka, S. R., Senthil, S., Srinivas, M. B. and Pyne, S. (2023). Focused analysis of RNFL decay in glaucomatous eyes using circular statistics on high-resolution OCT data. {\it PLOS ONE}, \textbf{18}, e0292915.

\bibitem[Angelaki et al.(2020)]{AngelakiEtAl2020} Angelaki, D. E., Ng, J., Abrego, A. M., Cham, H. X., Asprodini, E. K., Dickman, J. D. and Laurens, J. (2020). A gravity-based three-dimensional compass in the mouse brain. {\it Nature Communications}, \textbf{11}, 1855.



\bibitem[Caponera et al.(2022)]{CaponeraEtAl2022}
Caponera, A., Fageot, J., Simeoni, M. and Panaretos, V. M. (2022).
Functional estimation of anisotropic covariance and autocovariance operators on the sphere.
{\it Electronic Journal of Statistics}, \textbf{16}, 5080--5148.

\bibitem[Chavel(2006)]{Chavel2006} Chavel, I. (2006). {\it Riemannian Geometry: A Modern Introduction} (2nd ed.). Cambridge University Press.

\bibitem[Cornea et al.(2017)]{CorneaEtAl2017}
Cornea, E., Zhu, H., Kim, P., Ibrahim, J. G. and
the Alzheimer's Disease Neuroimaging Initiative (2017).
Regression models on Riemannian symmetric spaces.
{\it Journal of the Royal Statistical Society Series B: Statistical Methodology},
\textbf{79}, 463--482.

\bibitem[Dai and M\"{u}ller(2018)]{DaiMuller2018} Dai, X. and M\"{u}ller, H.-G. (2018). Principal component analysis for functional data on Riemannian manifolds and spheres. {\it Annals of Statistics}, \textbf{46}, 3334--3361.

\bibitem[Di Marzio et al.(2014)]{Di MarzioEtAl2014} Di Marzio, M., Panzera, A. and Taylor, C. C. (2014). Nonparametric regression for spherical data. {\it Journal of the American Statistical Association}, \textbf{109}, 748--763.

\bibitem[do Carmo(1992)]{DoCarmo1992} do Carmo, M. P. (1992). {\it Riemannian Geometry}. Birkh\"auser.

\bibitem[Engel et al.(2023)]{EngelEtAl2023}
Engel, I., Daugintis, R., Vicente, T., Hogg, A. O. T., Pauwels, J., Tournier, A. J. and Picinali, L. (2023). The SONICOM HRTF dataset. {\it Journal of the Audio Engineering Society}, \textbf{71}, 241--253.

\bibitem[Falorsi et al.(2019)]{FalorsiEtAl2019}
Falorsi, L., de Haan, P., Davidson, T. R. and Forr\'e, P. (2019).
Reparameterizing distributions on Lie groups.
{\it Proceedings of Machine Learning Research}, \textbf{89}, 3244--3253.

\bibitem[Hall et al.(2006)]{HallEtAl2006} Hall, P., M\"{u}ller, H.-G. and Wang, J.-L. (2006). Properties of principal component methods for functional and longitudinal data analysis. {\it Annals of Statistics}, \textbf{34}, 1493--1517.

\bibitem[Hall and Hosseini-Nasab(2006)]{HallHosseiniNasab2006} Hall, P. and Hosseini-Nasab, M. (2006). On properties of functional principal components analysis. {\it Journal of the Royal Statistical Society Series B: Statistical Methodology}, \textbf{68}, 109--126.

\bibitem[Im et al.(2025)]{ImEtAl2025} Im, C. J., Jeon, J. M. and Park, B. U. (2025). Local Fr\'{e}chet regression with spherical predictors. {\it Electronic Journal of Statistics}, \textbf{19}, 5313--5367.

\bibitem[Jeon et al.(2021)]{JeonEtAl2021} Jeon, J. M., Park, B. U. and Van Keilegom, I. (2021). Additive regression for non-Euclidean responses and predictors. {\it Annals of Statistics}, \textbf{49}, 2611--2641.

\bibitem[Jeon et al.(2022)]{JeonEtAl2022} Jeon, J. M., Park, B. U. and Van Keilegom, I. (2022). Nonparametric regression on Lie groups with measurement errors. {\it Annals of Statistics}, \textbf{50}, 2973--3008.

\bibitem[Jeon and Van Bever(2025)]{JeonVan2025}
Jeon, J. M. and Van Bever, G. (2025).
Additive regression for Riemannian functional responses.
{\it Journal of Multivariate Analysis}, \textbf{210}, 105466.


\bibitem[Le and Kendall(1993)]{LeKendall1993} Le, H. and Kendall, D. G. (1993). The Riemannian structure of Euclidean shape spaces: A novel environment for statistics. {\it Annals of Statistics}, \textbf{21}, 1225--1271.

\bibitem[Li and Hsing(2010)]{LiHsing2010} Li, Y. and Hsing, T. (2010). Uniform convergence rates for nonparametric regression and principal component analysis in functional/longitudinal data. {\it Annals of Statistics}, \textbf{38}, 3321--3351.

\bibitem[Lila et al.(2016)]{LilaEtAl2016} Lila, E., Aston, J. A. D. and Sangalli, L. M. (2016). Smooth principal component analysis over two-dimensional manifolds with an application to neuroimaging. {\it Annals of Applied Statistics}, \textbf{10}(4), 1854--1879.

\bibitem[Lin and Yao(2019)]{LinYao2019} Lin, Z. and Yao, F. (2019). Intrinsic Riemannian functional data analysis. {\it Annals of Statistics}, \textbf{47}, 3533--3577.

\bibitem[Olsen et al.(2027)]{OlsenEtAl2027}
Olsen, N. L., Pini, A. and Vantini, S. (2027).
Local inference for functional data on manifold domains using permutation tests.
{\it Journal of Multivariate Analysis}, \textbf{217}, 105719.

\bibitem[Pelletier(2005)]{Pelletier2005} Pelletier, B. (2005). Kernel density estimation on Riemannian manifolds. {\it Statistics \& Probability Letters}, \textbf{73}, 297--304.

\bibitem[Petersen and M\"{u}ller(2019)]{PetersenMuller2019} Petersen, A. and M\"{u}ller, H.-G. (2019). Fr\'{e}chet regression for random objects with Euclidean predictors. {\it Annals of Statistics}, \textbf{47}, 691--719.




\bibitem[Petersen(2024)]{Petersen2024} Petersen, A. (2024). Mean and covariance estimation for discretely observed high-dimensional functional data: Rates of convergence and division of observational regimes. {\it Journal of Multivariate Analysis}, \textbf{204}, 105355.

\bibitem[Reimherr et al.(2023)]{ReimherrEtAl2023} Reimherr, M., Sriperumbudur, B. and Kang, H. B. (2023). Optimal function-on-scalar regression over complex domains. {\it Electronic Journal of Statistics}, \textbf{17}, 156--197.

\bibitem[Rice and Silverman(1991)]{RiceSilverman1991} Rice, J. A. and Silverman, B. W. (1991). Estimating the mean and covariance structure nonparametrically when the data are curves. {\it Journal of the Royal Statistical Society: Series B (Methodological)}, \textbf{53}, 233--243.

\bibitem[Shao et al.(2022)]{ShaoEtAl2022} Shao, L., Lin, Z. and Yao, F. (2022). Intrinsic Riemannian functional data analysis for sparse longitudinal observations. {\it Annals of Statistics}, \textbf{50}, 1696--1721.

\bibitem[St\"{o}cker et al.(2023)]{Stocker et al. (2023)} St\"{o}cker, A., Steyer, L. and Greven, S. (2023). Functional additive models on manifolds of planar shapes and forms. {\it Journal of Computational and Graphical Statistics}, \textbf{32}, 1600--1612.

\bibitem[Yao et al.(2005)]{YaoEtAl2005} Yao, F., M\"{u}ller, H.-G. and Wang, J.-L. (2005). Functional data analysis for sparse longitudinal data. {\it Journal of the American Statistical Association}, \textbf{100}, 577--590.

\bibitem[Yuan et al.(2012)]{YuanEtAl2012}
Yuan, Y., Zhu, H., Lin, W. and Marron, J. S. (2012).
Local polynomial regression for symmetric positive definite matrices.
{\it Journal of the Royal Statistical Society Series B: Statistical Methodology},
\textbf{74}, 697--719.

\bibitem[Zhang and Wang(2016)]{ZhangWang2016} Zhang, X. and Wang, J.-L. (2016). From sparse to dense functional data and beyond. {\it Annals of Statistics}, \textbf{44}, 2281--2321.

\end{thebibliography}

\begin{thebibliography}{100}

\bibitem[Horv\'{a}th and Kokoszka(2012)]{HorvathKokoszka2012} Horv\'{a}th, L. and Kokoszka, P. (2012). {\it Inference for Functional Data with Applications}. Springer New York.
\bibitem[Hsing and Eubank(2015)]{HsingEubank2015} Hsing, T. and Eubank, R. (2015). {\it Theoretical Foundations of Functional Data Analysis, with an Introduction to Linear Operators}. John Wiley \& Sons.

\end{thebibliography}
\end{document}